\documentclass[10pt,journal]{IEEEtran}      
\pdfoutput=1

\usepackage{amsthm}
\usepackage{amsmath}
\usepackage{mathtools}
\usepackage{amssymb}
\usepackage{etoolbox}
\usepackage{enumerate}
\usepackage{bm}
\usepackage{subcaption}
\usepackage{cite}
\usepackage{graphicx}
\usepackage{url}

\usepackage{color}
\usepackage{amsmath}
\usepackage{amssymb}
\usepackage{graphicx}
\usepackage{comment,xspace}
\usepackage{fancybox}

\newtheorem{assump}{Assumption}
\newtheorem{theorem}{Theorem}[section]
\newtheorem{definition}[theorem]{Definition}

\newtheorem{lemma}[theorem]{Lemma}
\newtheorem{corollary}[theorem]{Corollary}
\newtheorem{remark}[theorem]{Remark}
\newtheorem{claim}[theorem]{Claim}

\newcommand{\real}{{\mathbb{R}}}
\newcommand{\reals}{\real}

\renewcommand{\baselinestretch}{0.93}

\usepackage{epstopdf}
\usepackage[usenames,dvipsnames]{xcolor}

\newcommand{\Asym}{A_\mathrm{sym}}
\newcommand{\CR}{\mathcal{C}_{L\otimes R}}

\newtoggle{arxiv}
\toggletrue{arxiv}

\title{Decentralized Algorithms for 3D Symmetric Formations in Robotic Networks:\\ a Contraction Theory Approach}

\author{Sumeet~Singh,~\IEEEmembership{Student Member,~IEEE,}
	  Edward~Schmerling,~\IEEEmembership{Student Member,~IEEE,}
	 Marco~Pavone,~\IEEEmembership{Member,~IEEE}%
\thanks{Sumeet Singh and Marco Pavone are with the Department of Aeronautics and Astronautics; Edward Schmerling is with the Institute for Computational and Mathematical Engineering, Stanford University, Stanford, CA, 94305, \{{\tt\small ssingh19, schmrlng, pavone}\} {\tt \small @stanford.edu}.}%
\thanks{This work was supported in part by the Stanford Graduate Fellowship (SGF) and the King Abdulaziz City for Science and Technology (KACST).}}

\begin{document}
\maketitle

\begin{abstract}
This paper presents decentralized algorithms for formation control of multiple robots in three dimensions. Specifically, we leverage the mathematical properties of cyclic pursuit along with results from contraction and partial contraction theory to design decentralized control algorithms that ensure global convergence  to symmetric formations. We first consider regular polygon formations as a base case, and then extend the results to Johnson solid and other polygonal mesh formations. The algorithms are further augmented to allow control over formation size and avoid collisions with other robots in the formation. The robustness properties of the algorithms are assessed in the presence of bounded additive disturbances and their effect on the \emph{quality} of the formation is quantified. Finally, we present a general methodology for embedding the control laws on complex dynamical systems, in this case, quadcopters, and validate this approach via simulations and experiments on a fleet of quadcopters.
\end{abstract}

\section{Introduction}
Recently there has been much interest in decentralized, multi-robot systems due to the plethora of possible applications and performance advantages. These systems have many attractive properties such as robustness to single-point failures, scalable implementation, and potentially lower operating costs when compared with monolithic approaches. Furthermore, multiple robots can together accomplish tasks which may be impossible or very difficult for any single robot on its own \cite{WS-QY-JT-NX:06, RSS-FYH:02, FA-HH-SC-GS:10}.

A prototypical problem in this context is decentralized formation control, that is, the problem of ensuring convergence to a desired formation via control algorithms amenable to a distributed implementation \cite{AJ-JL-ASM:03, JRTL-RWB-BJY:03, WR-RB:04}. Most existing results consider two-dimensional formations \cite{JM:05, JC-DS-JY-HC:10, LC-FM-DP-MT:08}, however, a number of robotic applications would require three-dimensional (3D) formations. For example, it may be desired to deploy a set of small satellites  to surround a larger damaged spacecraft to provide complete 3D reconstruction and visualization of the external structure. Additional examples include differential atmospheric, deep space, or underwater measurements of environmental phenomena, and filming sport and entertainment events \cite{DPS-FYH-SRP:04, YC-WY-WR-GC:13,FG-RS-VL-EN:14}. 

Accordingly, the objective of this paper is to design and rigorously analyze decentralized control algorithms for 3D formations. Our approach leverages the simple, yet effective strategy of cyclic pursuit. Essentially, the cyclic pursuit strategy entails letting each robot $i$ follow its leading neighbor $i+1$ modulo $n$, where $n$ is the number of robots. This approach is attractive due to its decentralized nature and low information requirements, namely relative position, and has led to its popularity in recent years \cite{JAM-MEB-BAF:04, MP-EF:07, JLR-MP-EF-DWM:10}. In this paper, we adopt the cyclic pursuit strategy proposed in~\cite{JLR-JJES:12} (that, in turn, generalizes earlier results in \cite{MP-EF:07, JLR-MP-EF-DWM:10}), where each robot follows its neighbors on both sides, that is, robots $i+m$ and $i-m$, where $m = \{1,\ 2,\ \ldots N\}$ is a look-ahead parameter up to a horizon $N$. This is referred to as the \emph{symmetric} cyclic control algorithm. To prove stability (i.e., convergence to a desired formation), we exploit contraction theory.  Contraction theory is a relatively recent innovation in control system design~\cite{WL-JJES:98}, and hinges upon  an exact differential analysis of convergence. At its core, contraction theory studies exponential convergence of pairs of system trajectories towards each other, and by extension, to a desired target trajectory. This represents a generalization with respect to traditional Lyapunov analysis which studies convergence to the origin, i.e., to a zero trajectory \cite{QCP-NT-JJES:09}. Partial contraction theory extends this concept by considering the convergence of trajectories to a ``set of properties," for example, a flow-invariant subspace~\cite{WW-JJES:05, QCP-JJES:07}. 

The contributions of this paper are as follows. First, we  study convergence of the symmetric cyclic control algorithm to a polygon formation in  3D. Our analysis differs from the analysis in~\cite{JLR-JJES:12} in that it considers a refined notion of the target formation (necessary to properly constrain a polygon formation). Our analysis (as in~\cite{JLR-JJES:12}) leverages partial contraction theory and, in addition, yields insights into the various tuning parameters for the controller and their impact on performance metrics such as convergence rate and control saturation. Second, we extend the definition of the target formation to accommodate combinations of regular polygons such as the set of convex regular polyhedra known as Johnson solids, as well as more general polygonal mesh formations. Control laws and sufficient conditions for convergence are derived for these relatively complex polyhedron formations, paving the way for innovative mission concepts \cite{OW-KE-RF:07}. Third, we provide extensions to the algorithm to control the size of the formation, and prevent collisions with other robots. Fourth, we analyze the robustness properties of the symmetric cyclic controller under bounded additive disturbances and provide performance bounds with respect to the formation error. These results provide theoretical insights into experiments performed on board the International Space Station which highlighted the remarkable robustness properties of cyclic controllers \cite{JLR-MP-EF-DWM:10}. Finally, we present a control hierarchy for implementing the formation control laws on a complex dynamical system, in this case, a fleet of quadcopters.

The rest of the paper proceeds as follows: Section \ref{sec:math_prelim} introduces the mathematical background for contraction and partial contraction theory and circulant matrices which forms the basis for all subsequent analysis. Section \ref{sec:planar} formalizes the problem, introduces the formation subspace and symmetric cyclic control law, and provides conditions for convergence to a  regular polygon. In Section \ref{sec:complex}, we extend the results to polyhedral formations. Section \ref{sec:additional} provides extensions to the algorithm for controlling formation size and preventing collisions among the robots. We also assess the robustness properties of the symmetric cyclic controller under bounded additive disturbances. In Section \ref{quad_cyclic}, we demonstrate the application of our formation control laws on quadcopters, providing verification via simulations and hardware experiments. Finally, in Section \ref{sec:conclusions}, we draw our conclusions and provide avenues for future work.

A preliminary version of this paper appeared as  \cite{SS-ES-MP:15}. This extended and revised version contains as additional contributions: (1) extensions to the control algorithms for forming more complex polygonal mesh formations, controlling formation size, and avoiding collisions, (2) additional numerical results, and (3) hardware experimental results on a fleet of quadcopters.

\section{Mathematic Preliminaries} \label{sec:math_prelim}
\subsection{Notation}
Given a square matrix $A$, the symmetric part of $A$, i.e.,  $(1/2)(A + A^T)$, is denoted by $\Asym$. The smallest and largest eigenvalues of $\Asym$ are denoted, respectively, by $\lambda_{\mathrm{min}}(A)$ and $\lambda_{\mathrm{max}}(A)$. Accordingly, the matrix $A$ is positive definite (denoted $A\succ0$) if $\lambda_{\mathrm{min}}(A)>0$, and negative definite (denoted $A\prec0$) if $\lambda_{\mathrm{max}}(A)<0$. Let $I_k$ denote the $k$-by-$k$ identity matrix. Let $A \otimes B$ denote the Kronecker product between matrices $A$ and $B$. The null space of a matrix $A$ is denoted by $ \mathcal{N}(A)$. Finally, the set of eigenvalues of $A$ is denoted by $\lambda(A)$.

\subsection{Contraction and Partial Contraction Theory}\label{contraction_theory}

This section reviews basic results in nonlinear contraction theory \cite{WL-JJES:98, JJ-JJES:04}. The basic stability result in contraction theory reads as follows.

\begin{theorem}[Contraction \cite{WL-JJES:98}]\label{them:contr}
Consider a general system of the form 
\begin{equation}
\label{eq:System}
\dot{\bm{x}} = \bm{f}(\bm{x},t),
\end{equation}
where $\bm{x}$ is the $n \times 1$ state vector and $\bm{f}$ is a $n \times 1$ nonlinear, continuously differentiable vector function with Jacobian $\partial \bm{f}/\partial \bm{x}$. If there exists a square matrix $\Theta(\bm{x},t)$ such that $\Theta(\bm{x},t)^{T} \, \Theta(\bm{x},t)$ is uniformly positive definite and the matrix 
\[
F := \left(\dot{\Theta} + \Theta \frac{\partial \bm{f}}{\partial \bm{x}} \right) \Theta^{-1}
\]
is uniformly negative definite, then all system trajectories converge exponentially to a single trajectory.  In this case, the system is said to be contracting.
\end{theorem}
A few remarks are in order. First, a matrix $\Theta(\bm{x},t)$ is uniformly positive definite if there exists $\beta>0$ such that $\forall (\bm{x},t),\  \lambda_{\mathrm{min}}(\Theta(\bm{x},t)) \geq \beta$, and is  uniformly negative definite if there exists $\beta>0$ such that $\forall (\bm{x}, t),\  \lambda_{\mathrm{max}}(\Theta(\bm{x},t)) \leq -\beta$. Second, the matrix $F$ is referred to as the \emph{generalized Jacobian} for system \eqref{eq:System}. Third, the matrix $M(\bm{x},t):= \Theta(\bm{x},t)^{T} \, \Theta(\bm{x},t)$ is referred to as the \emph{contraction metric}. 

We next discuss partial contraction theory, which allows one to address questions more general than trajectory
convergence. Consider system \eqref{eq:System}, and assume there exists a flow-invariant linear subspace $\mathcal{M} \subset \reals^n$, i.e., a linear subspace with the property that, for all $t$,
\[
\bm{x} \in \mathcal{M} \Rightarrow  \bm{f}(\bm{x},t) \in \mathcal M.
\]
Assume that the dimension of $\mathcal M$ is $p$ and let $(\bm{e}_1,\ldots, \bm{e}_n)$ be an orthonormal basis of $\reals^n$ where the first $p$ vectors form a basis for $\mathcal M$. Let $\bar{U}$ be a $p\times n$ matrix whose rows are   $\bm{e}_{1}^T, \ldots, \bm{e}_{p}^T$. Let $\bar{V}$ be an $(n-p)\times n$ matrix whose rows are   $\bm{e}_{p+1}^T, \ldots, \bm{e}_{n}^T$. One can easily verify that matrix $\bar{V}$ is sub-unitary and satisfies the properties $\bar{V}^T \bar{V} + \bar{U}^T  \bar{U} = I_n$, $\bar{V}\bar{V}^T = I_{n-p}$, and $\bm{x}\in \mathcal M$ if and only if $\bar{V} \bm{x} = \bm{0}$. We will refer to $\bar{V}$ as the projection matrix of $\mathcal M$. The main theorem in partial contraction theory can be stated as:

\begin{theorem}[Partial Contraction \cite{QCP-JJES:07}]
\label{manifolds_and_partial}
Consider a flow-invariant linear subspace $\mathcal M$ and its associated projection matrix $\bar{V}$. A particular solution $\bm{x}_p(t)$ of system \eqref{eq:System} converges exponentially to $\mathcal M$ if the auxiliary system
\[
\dot{\bm{y}} = \bar{V}\bm{f} \left (\bar{V}^T\bm{y}+\bar{U}^T \bar{U} \, \bm{x}_p(t),t \right) 
\]
is contracting with respect to $\bm{y}$.  If this is true for all particular solutions $\bm{x}_p$, all trajectories of system (\ref{eq:System})  will exponentially converge to $\mathcal M$ from all initial conditions. 
\end{theorem}

Combining Theorems and \ref{them:contr} and \ref{manifolds_and_partial}, one obtains a powerful tool to prove convergence to a desired flow-invariant subspace.

\begin{corollary}[Convergence to Flow-invariant Subspace \cite{QCP-JJES:07}]\label{co:applied_contraction}
A sufficient condition for global exponential convergence to $\mathcal{M}$ is
	\begin{equation*}
	\bar{V} \, \frac{\partial \bm{f}}{\partial \bm{x}}\, \bar{V}^T \prec 0, \qquad \text{uniformly}.
	\end{equation*}
\end{corollary}
In this paper, the subspace $\mathcal M$ will represent a desired robot formation. 

\begin{remark}[Non-orthonormal $V$]\label{V_Vbar_T}
Note that the application of partial contraction theory requires the rows of the projection matrix $\bar{V}$ to be orthonormal. However, the  matrix $V$ characterizing a subspace $\mathcal{M}$ may not be row-wise orthonormal, e.g., when it is obtained by combining a set of linearly independent equations. Nonetheless, as long as the equations are independent, the matrix $V$ will be full row rank and can be transformed via an invertible matrix $T$ into an orthonormal counterpart $\bar{V}$ which satisfies (1) $V = T\,\bar{V}$, (2) $\bar{V} \bm{x} = \bm{0} \Leftrightarrow\bm{x}\in \mathcal{M}$, and (3) $\bar{V} \, \frac{\partial \bm{f}}{\partial \bm{x}}\, \bar{V}^T \prec 0 \Leftrightarrow V \, \frac{\partial \bm{f}}{\partial \bm{x}}\, V^T \prec 0$ \cite{QCP-JJES:07}.
\end{remark}


\subsection{Circulant and Block-Circulant Rotational Matrices}\label{circulant_section}
A circulant matrix of order $n$ is a square matrix with the following structure:
\begin{equation}
C = \begin{bmatrix}c_1 & c_2 & \cdots & c_{n} \\c_{n} & c_1 & \cdots & c_{n-1}\\\vdots & \vdots &  & \vdots \\c_2 & c_3 & \cdots & c_1\end{bmatrix}.
\end{equation}
The elements of each row are identical to the row above, but shifted one position to the right and wrapped around. Thus, a circulant matrix can be compactly denoted as:
\[ C = \mathrm{circ}[c_1 \  c_2 \  \ldots \  c_{n}]. \]
A useful circulant matrix with dimensions $n\times n$ used in this paper is defined below, parametrized by the integer $m \in\{0,\ldots, n-1\}$:
\[
L_m := \mathrm{circ}[1,0,\ldots,0,\underbrace{-1}_{(m+1)\text{st element}},0,\ldots,0].
\]
\subsection{Kronecker Product}
Let $A$ and $B$ be matrices of size $m\times n$ and $p\times q$ respectively. The Kronecker product of $A$ and $B$ is defined as:
\[
 A\otimes B = \begin{bmatrix}a_{11}B & \cdots & a_{1n}B \\\vdots & & \vdots \\a_{m1}B & \cdots & a_{mn}B\end{bmatrix},
 \]
where $A\otimes B$ has dimensions $mp\times nq$. If $\lambda_{A}$ is an eigenvalue of $A$ with eigenvector ${\bf v_{A}}$, and similarly $\lambda_{B}$ and ${\bf v_{B}}$ are an eigenvalue and eigenvector pair for $B$ then the corresponding eigenvalue and eigenvector pair for $A\otimes B$ is $\lambda_{A}\lambda_{B}$ and $\bf v_{A}\otimes \bf v_{B}$. We also note the following identity: $(A\otimes B)(C\otimes D) = AC\otimes BD $, where $A,B,C$ and $D$ have compatible dimensions in order for the products $AC$ and $BD$ to be well defined.

In the following, we will denote by $\CR$ the set of block-circulant matrices that can be written as $L\otimes R_{\beta}$, where $L$ is a circulant matrix and $R_{\beta}$ is a rotation matrix about the axis $\bm{e}_z:=(0,0,1)^T$ with rotation angle $\beta$. From the properties of the Kronecker product, the eigenvalues and eigenvectors of a matrix in $\CR$ are given, respectively, by the product of eigenvalues and Kronecker product of eigenvectors of the circulant matrix $L$ and the rotation matrix $R_{\beta}$. 

\section{Symmetric Planar Formations}\label{sec:planar}

\subsection{Problem Setup}
Consider $n\geq 3$ mobile robots (uniquely labelled by an integer $i \in  \{1, \dots, n\})$. Denote the position of robot $i$ at time $t$ as $\bm{x}_i(t)$, where $\bm{x}_i(t) \in \reals^3$. The overall state vector is denoted by $\bm{x} = (\bm{x}_1^T, \bm{x}_2^T, \ldots, \bm{x}_n^T)^T$.  The dynamics of each robot are given by 
\begin{equation}\label{eq:sys_dynamics}
\bm {\dot{x}}_i = \bm{g}_i(\bm{x}) + \bm{u}_i(\bm{x},t),
\end{equation}
where $\bm{u}_i$ is the control action. It is desired to design the control actions so that (a) they drive the global state vector $\bm{x}$ to a desired symmetric formation, and (b) they are amenable to a decentralized implementation. As in~\cite{JLR-JJES:12}, our strategy is to ``encode" a symmetric formation as a formation subspace as discussed next. The proofs for all theorems and lemmas introduced in this section are provided in \iftoggle{arxiv}{the appendix}{\cite{SS-ES-MP:15EVb}}.

\subsection{Formation Subspace}\label{subsec:subspace}
In this section we consider \emph{regular polygons} in 3D as desired symmetric formations -- the extension to non-planar formations (which represents the main contribution of this paper) is discussed in Section \ref{sec:complex}. Consider the case where the direction normal to the desired formation polygon is aligned with the vector $\bm{e}_z$  (the general case can be reduced to this case via a coordinate transformation). We encode such a formation via the subspace: 
\begin{subequations}
\begin{align}
&\mathcal{M}_n = \{\bm{x} \in \reals^{3n}: \nonumber\\
&\, \, \, (\bm{x}_{i+1}\! -\! \bm{x}_i) = R_{2\pi/n}( {\bm{x}_{i+2}\! - \!\bm{x}_{i+1}} ), i=1, \ldots, n\!-\!2, \label{manifold_def1}\\ 
&\,\, \,  \bm{e}_z^T (\bm{x}_{n}\! -\! \bm{x}_{n-1})  = \bm{e}_z^T ( {\bm{x}_{1}\! - \!\bm{x}_{n}} )\}, \,  \label{manifold_def2}
\end{align}
\end{subequations}
where the indices are considered modulo $n$, and $R_{2\pi/n}$ denotes a counterclockwise rotation around $\bm{e}_z$. The $n-2$ constraints in \eqref{manifold_def1} will be referred to as \emph{rotational constraints} while the single constraint in \eqref{manifold_def2} will be referred to as the \emph{in-plane} constraint. The in-plane constraint, not considered in \cite{JLR-JJES:12}, is needed in order to ensure that all robots lie in the same plane (which is not ensured by the rotational constraints alone as they permit spiral formations). The following lemma shows that the constraints  \eqref{manifold_def1} and \eqref{manifold_def2} are indeed necessary and sufficient for the definition of a regular polygon. 

\begin{lemma}[Polygon Constraints]\label{Lemma:manifold_defined}
The set of constraints  \eqref{manifold_def1} and \eqref{manifold_def2} are necessary and sufficient for the definition of a regular polygon with normal direction $\bm{e}_{z}$. \end{lemma}

\begin{remark}[Polygon Degrees of Freedom]\label{rot_freedom}
The constraints \eqref{manifold_def1} and \eqref{manifold_def2} together form a set of $3n -5$ linearly independent equations in $3n$ variables. The five missing equations correspond to five distinct degrees of freedom: three in translation, one in scaling, and one in in-plane rotation. That is, the polygon may be translated anywhere in space, scaled in size, or rotated within the desired plane about the plane normal.
\end{remark}

Both \eqref{manifold_def1} and \eqref{manifold_def2} can be compactly represented as the null space of a certain matrix, as shown in the following lemma. 

\begin{lemma}[Compact Constraints]\label{Lemma:rot}
Let $W_{r_{n}} := [I_{n-2}, 0_{(n-2)\times 2}]$ and  $W_{p_{n}} := [0_{1\times (n-2)}, 1, 0]$. Define the $3 (n-2)+1 \times 3n$ matrix
\begin{equation}\label{eq:Vr_formula}
V := \underbrace{\begin{bmatrix}W_{r_{n}}\otimes I_{3} \\ W_{p_{n}} \otimes \bm{e}_{z}^{T} \end{bmatrix}}_{:= \mathcal{W}_n} \underbrace{(L_1 \otimes I_3 + (L_1 - L_2)\otimes R_{2\pi/n})}_{:=\mathcal{P}_{n}\in\CR}. 
\end{equation}
Then the rotational and in-plane constraints in equation \eqref{manifold_def1} and \eqref{manifold_def2} are equivalent to the equation $V \bm{x} = 0$.
\end{lemma}
Note that $W_{r_{n}}$ captures the rotational constraints, while 	$W_{p_{n}}$ captures the in-plane constraint. The subspace $\mathcal {M}_n$ can then be characterized as the null space of matrix $V$, i.e., $V \, \bm{x} = \bm{0} \Leftrightarrow\bm{x}\in \mathcal{M}_n$.

Our standing assumption throughout this paper is that the robots' internal dynamics, i.e., $\bm{g}(\bm{x})$, are flow-invariant with respect to $ \mathcal{M}_n$:
\begin{assump}[Flow-invariance]\label{ass:inv}
The internal dynamics are flow-invariant with respect to the desired formation, that is for all $\bm{x} \in \mathcal{M}_n$ one has $V\bm{g}(\bm{x}) = \bm{0}$.
\end{assump}

\subsection{Symmetric Cyclic Controller} \label{general_cyclic}
We consider the class of symmetric cyclic controllers proposed in~\cite{JLR-JJES:12} (in turn generalizing the pursuit controllers introduced in \cite{MP-EF:07}):
\begin{equation}\label{eq:symmetric_control}
\bm{u}_i = \sum\limits_{m = 1}^{N}k_m \big[R_m (\bm{x}_{i+m} \! -\! \bm{x}_{i})\! + \!R_m^T (\bm{x}_{i-m} \!- \!\bm{x}_{i})\big],
\end{equation}
where $N$ is the look-ahead horizon $(0<N<n-1)$, $k_m>0$ is a gain, $R_m $ is a rotation matrix around $\bm{e}_z$ with rotation angle $\alpha_m$, and $(\bm{x}_{i+m} - \bm{x}_i)$ and $(\bm{x}_{i-m} - \bm{x}_i)$ denote the relative coordinates among robot $i$ and its $i+m$ and $i-m$ neighbors, where $m=\{1,\ldots, N\}$ and summation is modulo $n$. Note that the control law is spatially-decentralized, as each robot only requires relative position information from a set of neighboring robots.

Then, the symmetric cyclic controller can be written in compact form as
\begin{align}
\bm{u}&= -\sum\limits_{m=1}^{N} k_m \, [L_m\otimes R_m \,+ L_m^T\otimes R_m^T ]\bm{x}  \nonumber \\
			 &= - \sum\limits_{m=1}^{N} k_m \, \mathcal{L}_{m}\, \bm{x} = -  \mathcal{L} \, \bm{x}, 
\label{eq:symmetric_control_simplified}
\end{align}			 			 
where $\mathcal{L}_{m} := {L}_m\otimes R_{m} +  {L}_m^T\otimes R^T_{m} \in \CR$ and $\mathcal{L} := \sum\limits_{m=1}^{N} k_m \, \mathcal{L}_{m}$. In order to apply partial contraction theory, the control law needs to be flow-invariant (note that, by Assumption \ref{ass:inv}, the internal dynamics are flow invariant). The next lemma shows that this is indeed the case.

\begin{lemma}[Flow-invariance] \label{cyclic_invariant}
Subspace $\mathcal{M}_n$ is flow-invariant with respect to the symmetric cyclic control law given in equation \eqref{eq:symmetric_control_simplified}.
\end{lemma}

We are now in a position to apply partial contraction theory to show that, under some assumptions, the symmetric cyclic controller drives the system to the desired formation subspace $\mathcal{M}_n$. 

As the encoded constraints are linearly independent, $V$ is full row rank, guaranteeing the existence of an orthonormal counterpart $\bar{V}$ (see Remark \ref{V_Vbar_T}). In the following,  to prove convergence to $\mathcal M_n$, we will apply partial contraction theory using $\bar{V}$. Let $\bar{U}$ be a matrix whose rows represent an orthonormal basis for the orthogonal complement of the subspace defined by the rows of $\bar{V}$. According to Theorem \ref{manifolds_and_partial}, we want to show that for system
\[
\dot{\bm{x}} = \bm{g} (\bm{x}) - \mathcal L \, \bm{x},
\]
the associated auxiliary system 
\[
\begin{split}
\dot{\bm{y}} = & \bar{V} \bigg( \bm{g}(\bar{V}^T\bm{y} + \bar{U}^T\bar{U}\bm{x}_p) -  \mathcal{L}\,(\bar{V}^T\bm{y} + \bar{U}^T\bar{U}\bm{x}_p) \bigg), 
\end{split}
\]
is contracting. Note that by Assumption \ref{ass:inv} and Lemma \ref{cyclic_invariant} the closed-loop dynamics are invariant with respect to $\mathcal M_n$. Then, according to Corollary \ref{co:applied_contraction}, one requires 
\[
\bar{V} \bigg( \frac{\partial \bm{g}}{\partial \bm{x}}  - \mathcal{L}\bigg)\bar{V}^T \prec 0, \qquad \text{ uniformly}. 
\]
By Remark \ref{V_Vbar_T}, since $V$ and $\bar{V}$ are related by an invertible transformation, the above stability requirement can be re-formulated in terms of $V$ directly. By factoring $V$ as $ \mathcal{W}_{n} \mathcal{P}_{n} $ (see equation \eqref{eq:Vr_formula}) one obtains the condition
\begin{equation} \label{eq:convergence_cond}
\mathcal{W}_{n} \mathcal{P}_{n} \bigg( \frac{\partial \bm{g}}{\partial \bm{x}}  - \mathcal{L}\bigg)  \mathcal{P}^{T}_{n} \mathcal{W}^{T}_{n}  \prec 0,\qquad \text{ uniformly}.
\end{equation}
Performing an eigenvalue analysis of \eqref{eq:convergence_cond} 
yields the main result of this section. The details of the proof are provided in \iftoggle{arxiv}{the appendix}{\cite{SS-ES-MP:15EVb}}.

\begin{theorem}[Polygon Convergence]
\label{contraction_th_orig}
Assume
\begin{equation}
\begin{split}
&\sup_{\bm{x},t}\! \bigg(\! \lambda_{\mathrm{max}}\!\left(\mathcal{P}_{n} \frac{\partial \bm{g}}{\partial \bm{x}}  \mathcal{P}^T_{n}\right)\! -\! \! \!\! \!  \! \min\limits_{\substack{1\leq i\leq n \\ k\in\{-1,0,1\}}} \! \sum\limits_{m=1}^{N} k_m \, \lambda^{(m)}_{ik}  \! \bigg)\!  < \! 0, 
\end{split}
\label{eq:eig_Theorem}
\end{equation}
where
{\small
\begin{equation*}
\begin{split}
\lambda^{(m)}_{ik}  =\dfrac{2}{{e^{\frac{2\pi(2(i-1)+k)j}{n}}}} &\bigg[\cos(k\alpha_m) - \,\,\cos\left(k\alpha_m \!+\! \frac{2\pi m(i-1)}{n}\right)\bigg]\!\\
&  \!  \left[\left(e^{\frac{2\pi(i-1+k)j}{n}}\!-\!1\right)\!\!\!\left(e^{\frac{2\pi(i-1)j}{n}}\!-\!1\right)\right]^2\!\!.
\end{split}
\end{equation*}
}
Then the robots governed by system \eqref{eq:sys_dynamics} under the symmetric cyclic controller \eqref{eq:symmetric_control} globally converge to a regular polygon formation, i.e., to the formation subspace $\mathcal{M}_n$.																	
\end{theorem}
Note that due to the necessary inclusion of the in-plane constraint, the eigenvalues in Theorem \ref{contraction_th_orig} differ from those in \cite{JLR-JJES:12}.

\begin{remark}[Formations of Fixed Size] \label{Remark:cyclic_fixed}
Theorem \ref{contraction_th_orig} provides a sufficient condition to ensure that the robots converge to a regular polygon formation. Note that depending on the internal dynamics $\bm{g}(\bm{x})$ and the rotation angle $\alpha_m$, the polygon formation will be contracting, expanding, or of fixed size. For instance, if $\bm{g}(\bm{x}) = \bm{0}$, by setting $\alpha_m<m\pi/n$  the polygon formation will contract towards a point; by setting  $\alpha_m > m\pi/n$  the polygon formation will expand as the robots travel along rays emanating from the center of mass of the formation; by setting $\alpha_m = m\pi/n$ the polygon formation will be of fixed size, where the size depends on the initial conditions. The freedom in scaling the formation will be utilized later to control the formation size.
\end{remark}
For a given number of robots, the primary design parameters in the symmetric cyclic controller include the gains $k_m$ and the look-ahead horizon $N$. It is clear that increasing the gains uniformly scales all eigenvalues of the projected Jacobian, thereby increasing the exponential convergence rate. However, this is at the expense of a more aggressive controller which increases the risk of control saturation. On the other hand, increasing the look-ahead horizon, while imposing greater information requirements for each robot, also increases the convergence rate but at a lower risk of control saturation. This is due to the fact that the symmetric cyclic controller leverages the degree of asymmetry at each node of the polygon. Increasing the look-ahead horizon therefore improves the ``estimate" of this asymmetry, rather than simply magnifying the control law by increasing the gain. This intuition is exemplified by the following simulation. 

For simplicity, we assume that the internal dynamics $\bm{g}(\bm{x})$ are zero. In order to compare the control signals due to varying choices in gains and look-ahead, we maintain the same lower bound on the closed-loop convergence rate, i.e., $\lambda_{\min}\left(\bar{V}\mathcal{L}\bar{V}\right)$. Figure \ref{fig:cyclic_effort} shows the Euclidean norm of the net control signal, i.e., $\|\mathcal{L}\bm{x}(t)\|_2$ for six robots starting at the same initial conditions. The control signal in Figure \ref{fig:norm_horizon} demonstrates a significantly lower peak value than the signal in Figure \ref{fig:norm_gain} despite the fact that the two closed-loop systems possess the same lower on bound on convergence rate.
\begin{figure}[htbp]
\centering
\begin{subfigure}[t]{0.24\textwidth}
	\includegraphics[width=\textwidth]{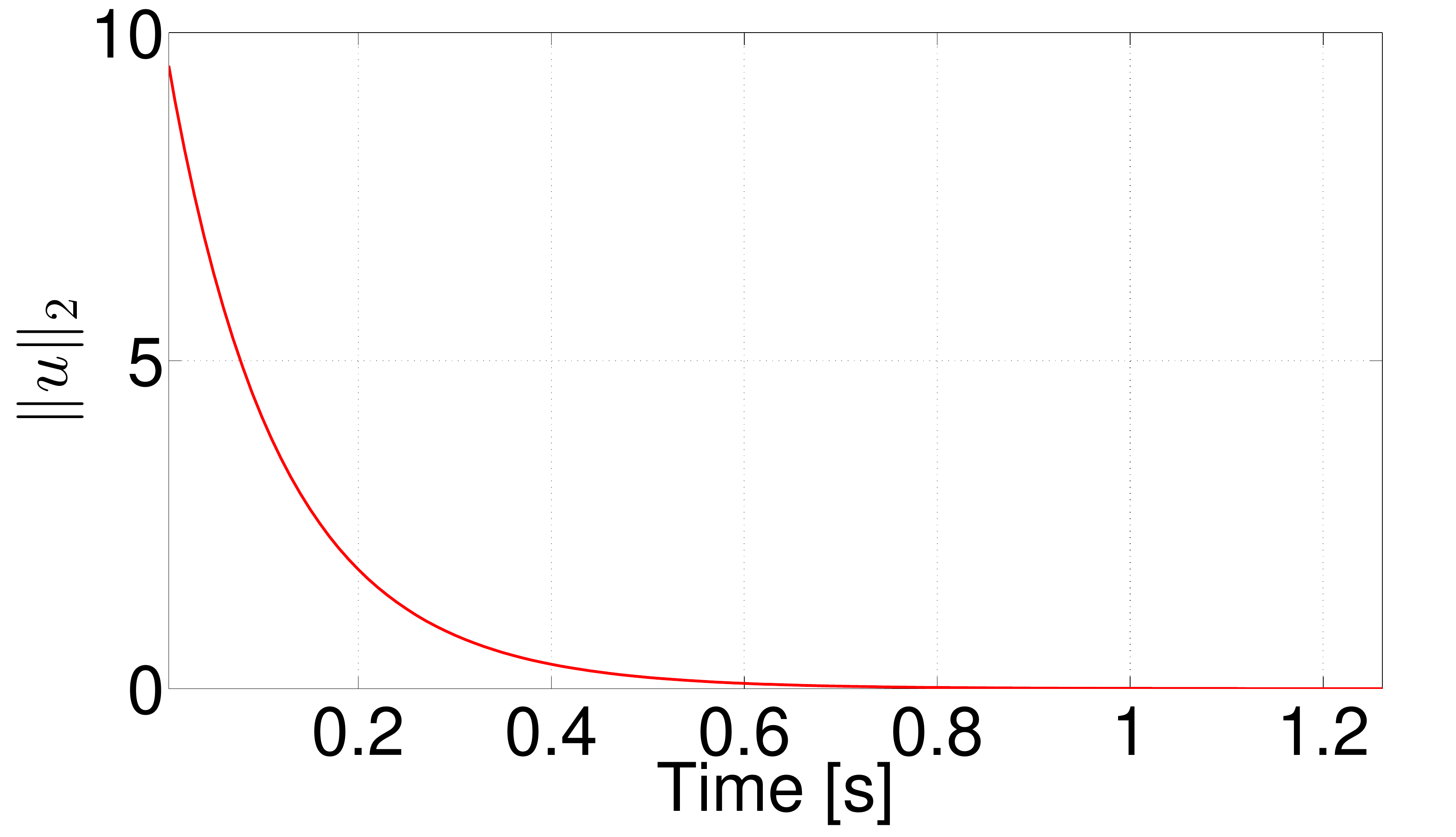}
	\caption{Control norm with $N = 2$, $k_m = 2$.}
	\label{fig:norm_horizon} 
\end{subfigure}
\begin{subfigure}[t]{0.24\textwidth}
	\includegraphics[width=\textwidth]{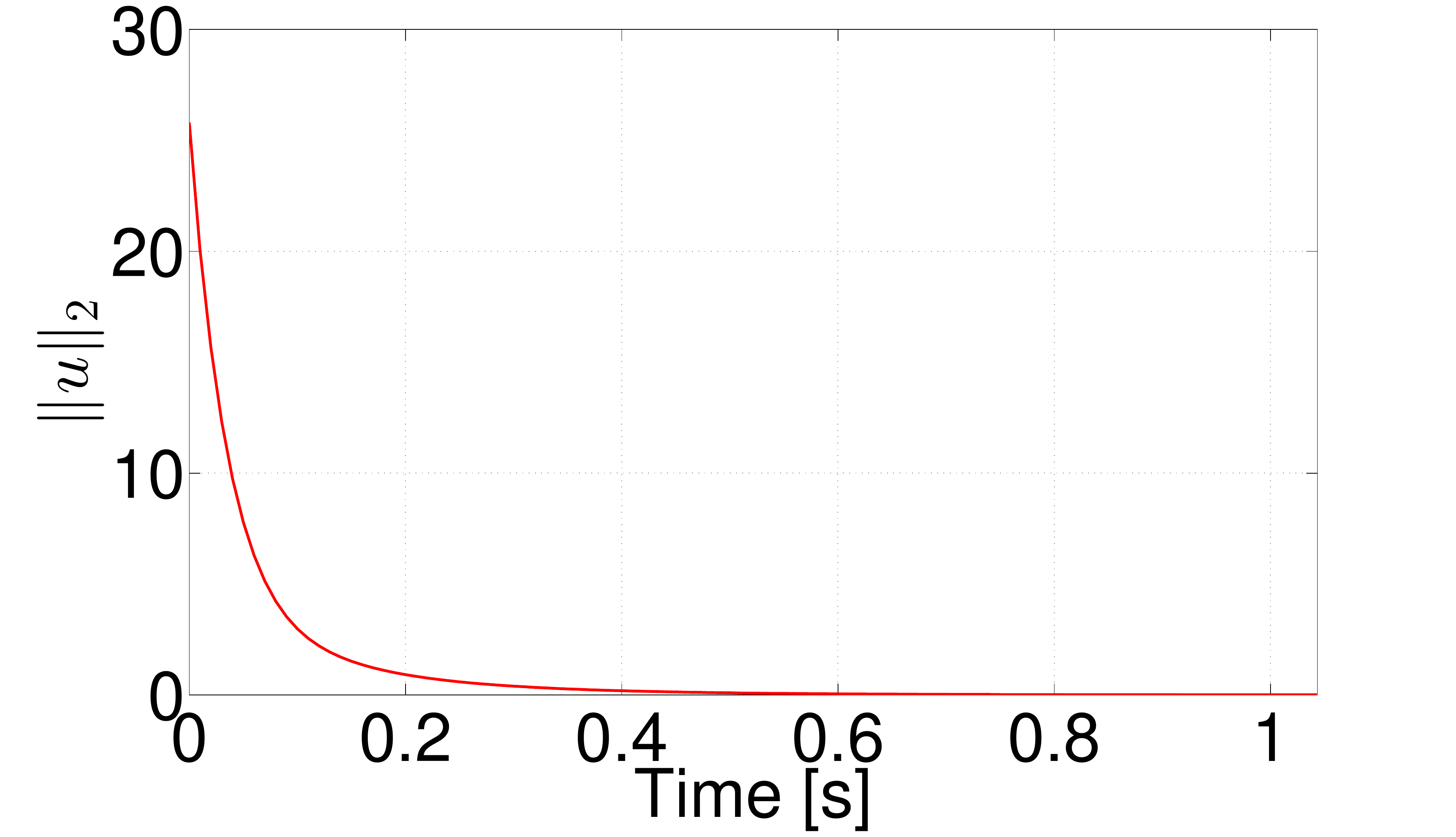}
	\caption{Control norm with $N = 1$, $k_m = 6.928$.}
	\label{fig:norm_gain} 
\end{subfigure}
\caption{Comparison of control effort with varying gains and look-ahead horizon.}
\label{fig:cyclic_effort}
\end{figure}

\begin{remark}[General Plane of Convergence]\label{Remark:Reta}
For  polygons not orthogonal to $\bm{e}_z$, one can use a similarity transformed version of the control rotation matrix $R_{m}$ to obtain analogous convergence results. Let $R_{\eta}^T\bm{e}_z$ be the desired polygon normal where $R_{\eta}$ is an arbitrary rotation matrix. Define  $R_{m_s} := R_{\eta}^{T} R_{m} R_{\eta}$ to be the similarity transform of $R_m$. We now replace $R_{m}$ in the individual control law given in (\ref{eq:symmetric_control}) by $R_{m_s}$, and the expression in (\ref{eq:symmetric_control_simplified}) becomes:
\begin{equation} \label{cyclic_modified}
\begin{split}
\bm{u} &= -\underbrace{ ( I_{n} \otimes R^{T}_{\eta}) }_{:=\mathcal{R}^{T}_{\eta} }  \sum\limits_{m=1}^{N} k_m \mathcal{L}_{m} \underbrace{( I_{n} \otimes R_{\eta}) }_{:=\mathcal{R}_{\eta}}\bm{x} \\
	   &= -\sum\limits_{m=1}^{N} k_m \mathcal{L}_{m\eta} \bm{x} = -\mathcal{L}_{\eta}\bm{x},
\end{split}
\end{equation}
where $\mathcal{L}_{m\eta}:= \mathcal{R}^{T}_{\eta}\mathcal{L}_{m}\mathcal{R}_{\eta}$ and $\mathcal{L}_{\eta} := \sum\limits_{m=1}^{N} k_m \, \mathcal{L}_{m\eta}$. The expression for the projection matrix $V$ becomes:
\begin{equation} \label{Vr_modified}
V = \mathcal{W}_{n} \mathcal{P}_{n} \mathcal{R}_{\eta}.
\end{equation}
The convergence analysis under control law \eqref{cyclic_modified} is virtually identical to the one provided above for the case where the polygon is orthogonal to $\bm{e}_z$. The generalized control law \eqref{cyclic_modified} will be leveraged in  Section \ref{sec:complex} to design control laws for polyhedral formations.
\end{remark}

\section{Polyhedral Formations}\label{sec:complex}
In this section we extend the results of Section~\ref{sec:planar} to \emph{polyhedral formations}.  Specifically, we first focus on convex polyhedra having regular faces and equal edge lengths, referred to as Johnson solids. Generalizations to other polygonal mesh formations are discussed at the end of this section.
\begin{definition}[Convex Polyhedral Solid and Polyhedral Surface] \label{def:polyhedron}
A convex solid polyhedron, or simply {polyhedron}, is defined as a body bounded by a finite number of polygons such that it lies on one side of the plane of each polygon. Equivalently, it is the convex hull of the polygon vertices. A {complete polyhedral surface} is the set of bounding faces for a polyhedron. A {partial polyhedral surface (PPS)} is a subset of the faces of a complete polyhedral surface.
\end{definition}
Figure \ref{fig:solid_reps} illustrates the distinction between a complete and partial surface of a convex solid in the definition above. In Figure \ref{fig:cube_incomp}, only a portion of the surface of the Johnson solid (cube) is specified, and thus the PPS formed by the faces in Figure \ref{fig:cube_incomp} contains un-matched edges. The set of these unmatched edges (shown in dashed form in Figure \ref{fig:cube_incomp}) is termed the \emph{boundary} of the PPS. The complete polyhedral surface displayed in Figure \ref{fig:cube_comp} encompasses the entire surface of the solid cube.
\begin{figure}[htbp]
\vspace{-3mm}
\begin{subfigure}[t]{.24\textwidth}
	\centering
	\includegraphics[width=0.9\textwidth]{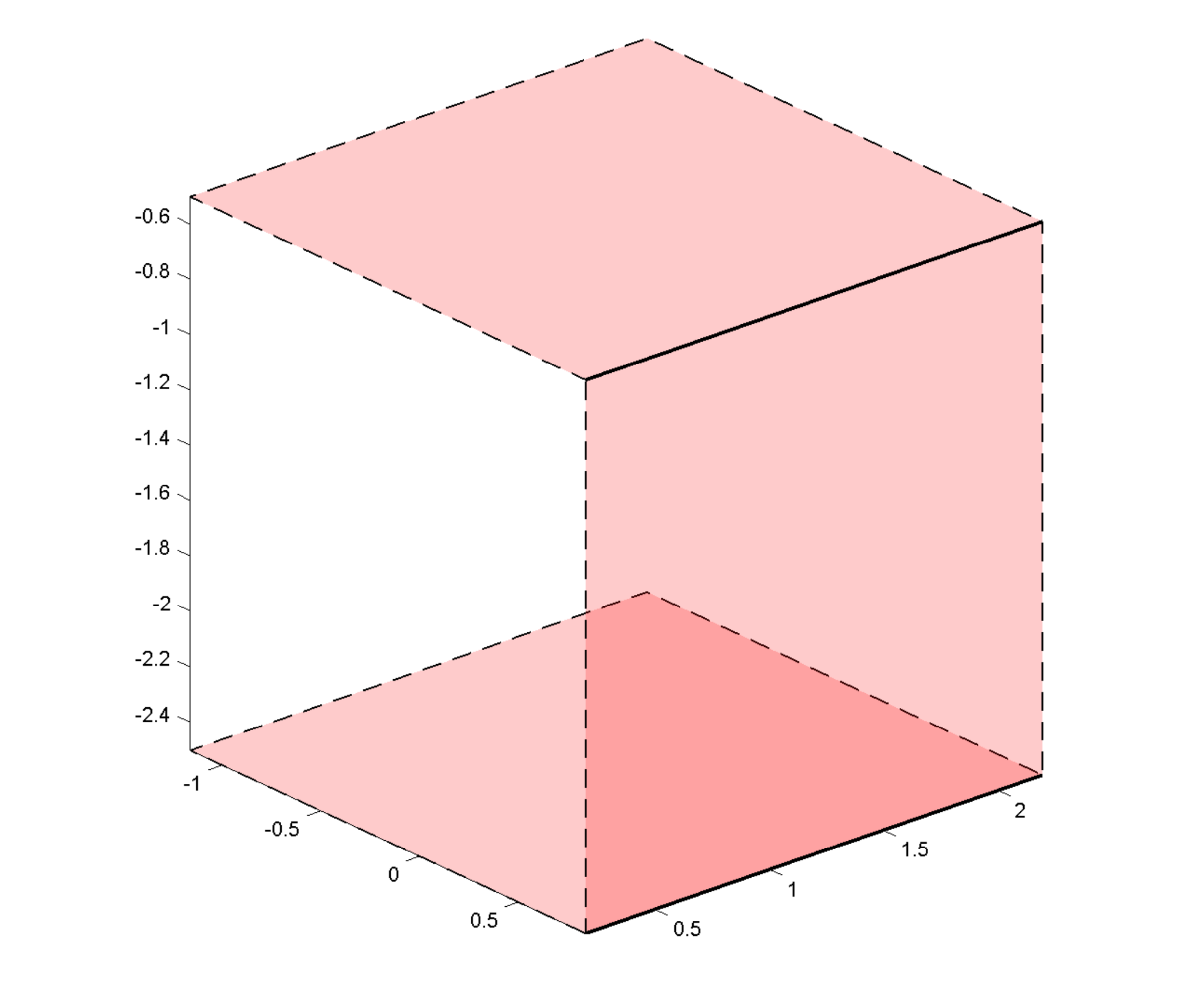}
	\caption{Partial polyhedral surface of a cube.}
	\label{fig:cube_incomp} 
\end{subfigure}
\begin{subfigure}[t]{.24\textwidth}
	\centering
	\includegraphics[width=1.05\textwidth]{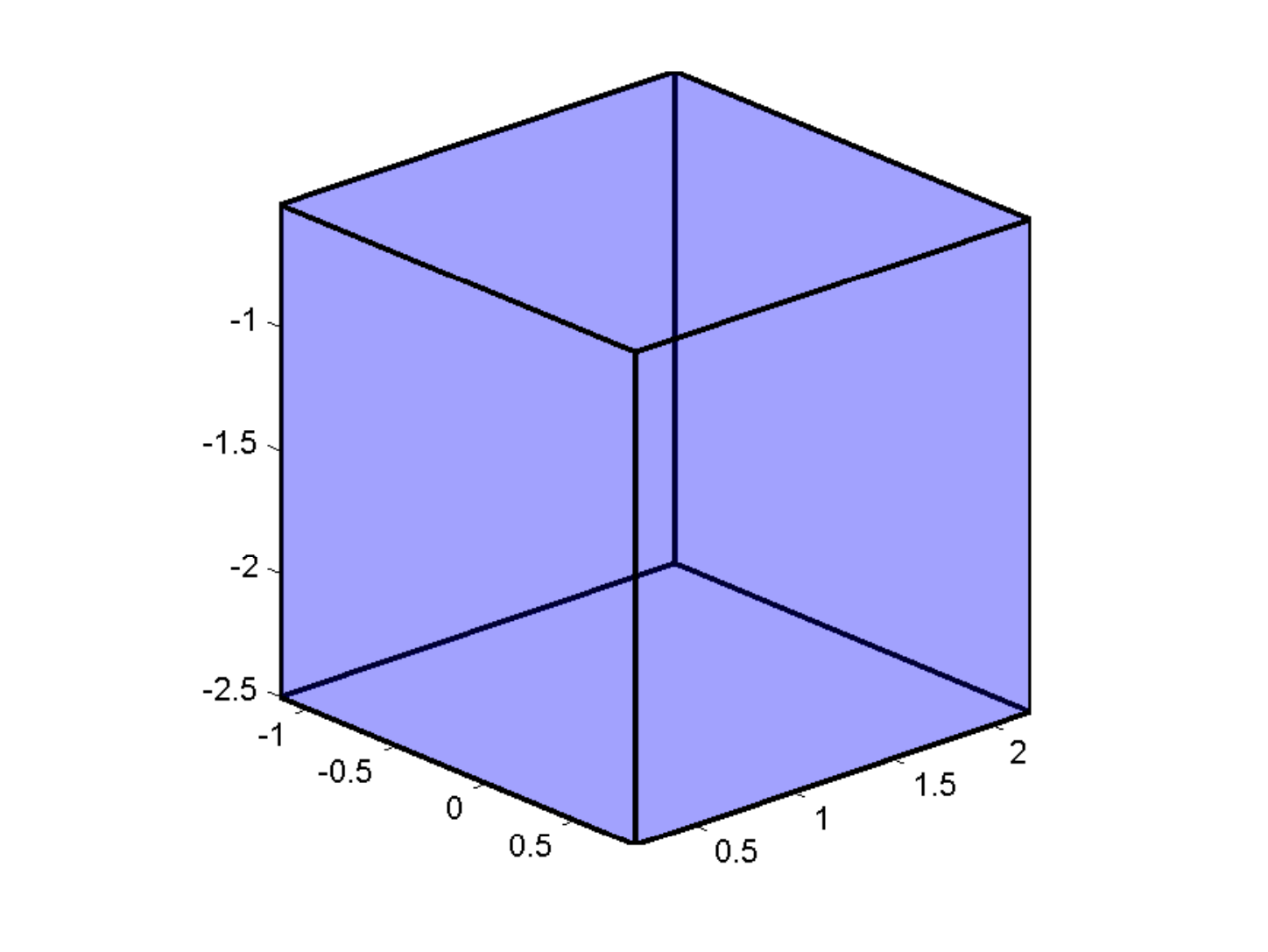}
	\caption{Complete polyhedral surface of a cube.}
	\label{fig:cube_comp} 
\end{subfigure}
\caption{Partial and complete representations of a bounded polyhedron (cube) with same vertex set.}
\label{fig:solid_reps}
\vspace{-2mm}
\end{figure}

Consider now a polyhedron $\mathcal{P}$ with complete polyhedral surface $\mathcal{F}$. As an extension to the discussion in Section~\ref{sec:planar}, each face of $\mathcal{F}$ can be described by a set of linearly independent constraints on its vertices. When considering these faces in combination, we must ensure that the full set of constraints (i.e., due to all faces in $\mathcal{F}$) remain linearly independent. This motivates the search for a PPS representation of the polyhedron that is \emph{minimal} in the sense that it does not induce redundant constraints yet still retains enough information to fix the positions of each vertex (robot). We note that this reduced set of linear constraints for the PPS will retain the translational and scaling degrees of freedom associated with the PPS's constituent polygons (see Remark~\ref{rot_freedom}). Thus all further discussion of the terms ``face'' or ``polygon'' in this section will implicitly assume translation and scale invariance. In the following subsection, we show that a suitable minimal PPS representation can always be constructed for any given convex polyhedron.

\subsection{Formulation of a Minimal Partial Polyhedral Surface}
We first introduce some terminology (adapted from \cite{ADA:05}):

\begin{definition}[Development of a Polyhedral Surface \cite{ADA:05}]\label{def:develop}
A development $\mathcal{D}$ is defined as the tuple $(\mathcal{Q},\mathcal{R})$ where $\mathcal{Q}$ is a set of polygons as defined in Section~\ref{sec:planar} and $\mathcal{R}$ is a set of rules that identify common edges between polygons. The sets $\mathcal{Q}$ and $\mathcal{R}$ satisfy the following properties:
\begin{enumerate}
    \item For every pair of adjacent polygons $\mathcal{Q}_i,\mathcal{Q}_j \in \mathcal{Q}$, there is a rule $r_{ij} \in \mathcal{R}$ which encompasses two pieces of information, namely, it identifies (a) the two shared vertices between the two polygons, and (b) the dihedral angle formed between the planes of the two polygons.
	\item It is possible to create a path between any two polygons in the collection by following the common edges (i.e., a development cannot be split into disconnected parts).
	\item Each side of every polygon is either matched with no edge of any other polygon, or with exactly one edge of another polygon in the collection.
\end{enumerate}
\end{definition}
Every development defines a \emph{unique} (complete or partial) polyhedral surface $\mathcal{F}$, but multiple developments may correspond to the same surface. The surface in Figure \ref{fig:cube_incomp} corresponds to a development with three squares, while the surface in Figure \ref{fig:cube_comp} corresponds to a development with six squares (this development is commonly referred to as the polyhedral net \cite{MJW:74}). Our goal is to be able to use the development for the PPS in Figure \ref{fig:cube_incomp} as a viable representation of the complete polyhedral surface in Figure \ref{fig:cube_comp} by leveraging the following result:

\begin{lemma}[Completion of a Polyhedral Surface \cite{ADA:05}] \label{lem:completion}
Every bounded
convex PPS $\mathcal{F}'$ gives rise to a unique complete polyhedral surface $\mathcal{F}$ without adding new vertices. If $\mathcal{A}'$ denotes the set of vertices for $\mathcal{F}'$, then the completion $\mathcal{F}$ is given by the surface of the convex hull of  $\mathcal{A}'$.
\end{lemma}
We are now ready to present the first step in constructing a PPS representation of a complete polyhedral surface where the set of linear constraints induced by the faces of the PPS are linearly independent. 
\begin{lemma}[Minimal Development]\label{lem:reduced_poly}
Let $\mathcal{F}$ be a complete polyhedral surface and let $\mathcal{A}$ denote its set of vertices. Suppose $\mathcal{D} = (\mathcal{Q}, \mathcal{R})$ is a development for $\mathcal{F}$. Choose a set of polygons $\mathcal{Q}' \subset \mathcal{Q}$ that satisfies the following properties:
\begin{enumerate}
	\item Tree: Within the dual graph of $\mathcal{F}$, the subgraph induced by $\mathcal{Q}'$ is a tree,
	\item Vertex span: The set of vertices spanned by $\mathcal{Q}'$ is identical to the set of vertices spanned by $\mathcal{Q}$.
\end{enumerate}
Given $\mathcal{Q}'$ satisfying the properties above, we now construct the set of rules $\mathcal{R}' \subset \mathcal{R}$ corresponding to the polygons selected in $\mathcal{Q}'$:
\[
\mathcal{R}' = \left\{r_{ij} \in \mathcal{R} \mid Q_i, Q_j \in \mathcal{Q}' \right\}.
\]
Then, the tuple $\mathcal{D}' := (\mathcal{Q}', \mathcal{R}')$ is a well-posed development (essentially a spanning tree of the graph $\mathcal{D}$) and therefore defines a PPS $\mathcal{F}'$. Furthermore, the completion of $\mathcal{F}'$ is the original surface $\mathcal{F}$.
\end{lemma}
\begin{proof}
Properties 1) and 3) of Definition~\ref{def:develop} hold for $\mathcal{Q}', \mathcal{R}'$ by construction. Property 2) holds because the `tree' assumption for $\mathcal{Q}'$ ensures that there exists a path between any two polygons in $\mathcal{Q}'$ by following common edges. Thus $\mathcal{D}'$ is well-defined and therefore describes a PPS $\mathcal{F}'$.

The `vertex span' assumption ensures that the convex hulls of the vertices of $\mathcal{F}$ and $\mathcal{F}'$ are identical. Since $\mathcal{F}$, a complete polyhedral surface, is the surface of this convex hull, by Lemma~\ref{lem:completion} the completion of $\mathcal{F}'$ is $\mathcal{F}$.
\end{proof}

The key message of Lemma~\ref{lem:reduced_poly} is that a PPS is sufficient to describe a polyhedron (i.e., the desired formation) provided that it encompasses all vertices of the polyhedron and its faces are connected. More precisely, each face within the PPS defines a set of linear constraints on its vertices, the combination of which completely determine the convex hull matching the desired formation. The further restriction that the faces induce a tree graph is an important step, as shown next, towards ensuring that the combined linear constraints are independent. 
For the rest of this section, we assume that a PPS $\mathcal{F}$ described by the development $\mathcal{D} = (\mathcal{Q}, \mathcal{R})$ has been fixed. Denote $L := |\mathcal{Q}|$, the number of polygons in $\mathcal{F}$. We now make explicit the notion of rules of a development.
\begin{figure}[htbp]
\vspace{-2mm}
\centering
	\includegraphics[scale=0.65]{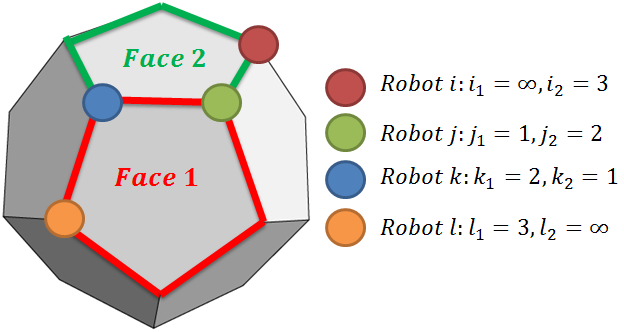}
	\caption{Indexing of robots for a general polyhedron.}
	\label{fig:3d_poly}
\vspace{-3mm}
\end{figure}

Let $\mathcal{Q}_k$ denote the $k^{\text{th}}$ polygon in $\mathcal{Q}$. Each $\mathcal{Q}_k$ is associated with a set of robots $\mathcal V_k$, and a unit outward normal $\bm{n}_k$ (pointing outward from the interior of the polyhedron). The number of robots $n$ is equal to the number of vertices in the polyhedron and each robot is uniquely associated with a vertex. Each face induces a counterclockwise (with respect to $\bm{n}_k$) sub-indexing of the robots, specifically, let $i_k$ denote the index of robot $i$ in face $k$, where $i=1,\ldots, n$ and $k=1, \ldots, L$, see Figure  \ref{fig:3d_poly}. If robot $i$ does not belong to face $k$, then we use the convention $i_k = \infty$. 


For each $\mathcal{Q}_k \in \mathcal{Q}$, define the matrix $\bar E^{(k)} \in \reals^{|\mathcal{V}_k|\times n}$ as
\begin{equation}
[\bar E^{(k)}]_{ij} = \begin{cases} 1 &\mbox{if } j_k = i, \\
					    0 &\mbox{otherwise.}
			\end{cases}
\label{robot_subset}
\end{equation}
Thus $[\bar E^{(k)}]_{ij}$ equals one if robot $j$'s index within face $k$ is equal to $i$. Define $\bm{x}^{(k)} \in \reals^{3|\mathcal{V}_k|}$ to be the positions of the robots in face $\mathcal{Q}_k$. Using \eqref{robot_subset}, we may write $\bm{x}^{(k)}$ as
\begin{equation} \label{face_distr}
\bm{x}^{(k)} = \left( \bar E^{(k)}\otimes I_3\right)\bm{x} = E^{(k)}\bm{x},
\end{equation}
where $\bm{x}\in\reals^{3n}$ is the global state vector, and $E^{(k)} := \bar E^{(k)}\otimes I_3$. As a direct extension of Section \ref{subsec:subspace}, the face $\mathcal{Q}_k$ (by assumption a regular polygon) is defined as the subspace $\mathcal{M}_{|\mathcal{V}_k|}^{(k)}$ as follows:
\begin{subequations}
\begin{align}
&\mathcal{M}_{|\mathcal{V}_k|}^{(k)} = \{\bm{x}^{(k)} \in \reals^{3|\mathcal{V}_k|}: \nonumber\\
&\, \, \,R_{\eta_k}(\bm{x}_{i+1}^{(k)}\! -\! \bm{x}_i^{(k)}) = R_{2\pi/n}R_{\eta_k}( {\bm{x}_{i+2}^{(k)}\! - \!\bm{x}^{(k)}_{i+1}} ), \label{manifold_def1_p} \\
&\hspace{37mm} i=1, \ldots, |\mathcal V_k|\!-\!2, \nonumber \\ 
&\,\, \,  \bm{e}_z^T R_{\eta_k}(\bm{x}^{(k)}_{|\mathcal V_k|}\! -\! \bm{x}^{(k)}_{|\mathcal V_k|-1})  = \bm{e}_z^T R_{\eta_k}( {\bm{x}^{(k)}_{1}\! - \!\bm{x}^{(k)}_{|\mathcal V_k|}} )\},\label{manifold_def2_p}
\end{align}
\end{subequations}
where $R_{\eta_k}$ is a rotation matrix defined such that $R_{\eta_k}^T\bm{e}_z = \bm{n}_k$. The desired Johnson polyhedral formation is the convex hull of a vertex set which can  be represented as a subspace $\mathcal{M}_n \subset \reals^{3n}$ defined by the polygonal constraints that stem from each face in $\mathcal{Q}$, that is  
\begin{align}
 \mathcal{M}_{n} &= \bigcap\limits_{k:\mathcal{Q}_k\in \mathcal{Q}} \mathcal{M}_{|\mathcal{V}_k|}^{(k)}.
\label{poly_manifold_net}
\end{align}
By Lemma \ref{Lemma:rot}, the constraints in \eqref{manifold_def1_p} and \eqref{manifold_def2_p} can be compactly represented as the null space of a $3(|\mathcal{V}_k|-2)+1\times 3n$ matrix, denoted by $\widetilde{V}^{(k)}$. In other words, $\bm{x} \in \mathcal{M}_{|\mathcal{V}_k|}^{(k)} \Leftrightarrow \widetilde{V}^{(k)} \bm{x} = \bm{0}$. Using \eqref{Vr_modified} and \eqref{face_distr}, one can write $\widetilde{V}^{(k)}$ as
\begin{equation}
\widetilde{V}^{(k)} = \mathcal{W}_{|\mathcal{V}_k|}\mathcal{P}_{|\mathcal{V}_k|}\mathcal{R}_{\eta_k}E^{(k)}.
\label{poly_manifold}
\end{equation}
Thus, the subspace $\mathcal{M}_n$ in \eqref{poly_manifold_net} may be equivalently re-stated as the intersection of null spaces:
\[
	\mathcal{M}_n = \bigcap\limits_{k:\mathcal{Q}_k \in \mathcal{Q}}\mathcal{N}(\widetilde{V}^{(k)}).
\]
Note that the rules $\mathcal{R}$ are implicitly encoded in the indexing used to write down the constraint matrices $\widetilde{V}^{(k)}$ and their nullspace intersection. Given our representation of the PPS in terms of the set $\mathcal{Q}$, we can now combine \eqref{poly_manifold_net} and \eqref{poly_manifold} to form the global constraint matrix $\widetilde{V}$:
\begin{equation}\label{V_poly}
	\widetilde{V}:= \begin{bmatrix} \widetilde{V}^{(1)^T} & \widetilde{V}^{(2)^T} & \cdots & \widetilde{V}^{(L)^T} \end{bmatrix}^T.
\end{equation}
The global formation subspace $\mathcal{M}_n$ is then defined to be the null space of $\widetilde{V}$, i.e., 
\[
	\bm{x} \in \mathcal{M}_n   \Leftrightarrow \widetilde{V}\bm{x} = \bm{0}.
\]
As constructed, $\widetilde{V}$ is not yet full row rank. A geometric explanation for this linear dependence is that once the orientations of two adjacent PPS faces are fixed, the in-plane rotational degrees of freedom for the remaining faces vanish. This notion is formalized in the following (relatively straightforward) corollary of Lemma \ref{Lemma:manifold_defined}, the proof of which is provided in \iftoggle{arxiv}{the appendix.}{\cite{SS-ES-MP:EVb}.}
\begin{corollary}[Reduced Polygon Constraints]\label{corr:manifold_defined_red}
Assume two neighboring robots $j$ and $j+1$ of face $\mathcal{Q}_k$ are constrained to lie in a plane normal to $\bm{n}_k$, where the indices $\{j, j+1\} \in \{1,\ldots,|\mathcal{V}_k|\}$ are modulo $|\mathcal{V}_k|$. Then the rotational constraints given in \eqref{manifold_def1_p} are necessary and sufficient for the definition of a regular polygon with normal direction $\bm{n}_k$. 
\end{corollary}

Since stability analysis using partial contraction is contingent upon a full row rank constraint matrix, we present a reduction of $\widetilde{V}$ that discards redundant constraints. Consider the following reduced constraint matrix $V$ with similar structural properties as $\widetilde{V}$:
\begin{align}
    &{V} := \begin{bmatrix} {V}^{(1)^T} & {V}^{(2)^T} & \cdots & {V}^{(L)^T} \end{bmatrix}^T,\\
    &V^{(k)} = \begin{cases} \widetilde{V}^{(k)} &\mbox{if } k = 1, 2, \\
                        (W_{r_{|\mathcal{V}_k|}}\otimes I_{3}) \mathcal{P}_{|\mathcal{V}_k|}\mathcal{R}_{\eta_k}E^{(k)} & k = 3,\ldots,L,
            \end{cases}
            \label{V_block_red}
\end{align}
where, as defined in Lemma \ref{Lemma:rot}, $W_{r_{|\mathcal{V}_k|}} =  [I_{|\mathcal{V}_k|-2}, 0_{(|\mathcal{V}_k|-2)\times 2}]$. That is, in-plane constraints have been removed from all but two (adjacent) faces within $\mathcal{Q}$, specifically, $\mathcal{Q}_1$ and $\mathcal{Q}_2$. The following lemma shows that the reduced constraints that stem from \eqref{V_block_red} are equivalent to the original set of constraints using $\widetilde{V}^{(k)}$, and that the resulting global constraint matrix $V$ is full row rank. 

\begin{lemma}[Minimal Representation of $V$]\label{Lemma:V_reduced}
Denote the set of equations $\widetilde{V}^{(k)}\bm{x} = \bm{0}$ where $\widetilde{V}^{(k)}$ has the form given in \eqref{poly_manifold} for $k = 1,\ldots,L$ as the full-constraint set. Similarly, denote the set of equations $V^{(k)}\bm{x} = \bm{0}$ where $V^{(k)}$ has the form given in \eqref{V_block_red} as the reduced-constraint set. Then, the solutions to the two sets of equations are identical. That is, 
\[
	\bm{x} \in \mathcal{M}_n   \Leftrightarrow V\bm{x} = \bm{0}.
\]
Furthermore, $V$ is full row rank.
\end{lemma}
\begin{proof}
We assume without loss of generality that the faces $\mathcal{Q}_k$ (and thus the corresponding constraint matrices $V^{(k)}$) are ordered so that for any $j = 1, ..., L$ all faces $\mathcal{Q}_i$ along the unique path between $\mathcal{Q}_1$ and $\mathcal{Q}_j$ satisfy $i \leq j$.
Let $\bm{\mathcal{Q}_k} = \{\mathcal{Q}_1,\ldots,\mathcal{Q}_k\}$ denote a partial list of faces, constrained by the set of equations $[V^{(k)}]\bm{x} = \bm{0}$, where $[V^{(k)}] = \begin{bmatrix}V^{(1)^T} & \cdots & V^{(k)^T}\end{bmatrix}^T,\ k\leq L$. Similarly, let $[\widetilde{V}^{(k)}] = \begin{bmatrix}\widetilde{V}^{(1)^T} & \cdots & \widetilde{V}^{(k)^T}\end{bmatrix}^T$. We proceed by induction.

Base case $k = 2$: It is clear that $\mathcal{N}\left([V^{(2)}]\right) = \mathcal{N}([\widetilde{V}^{(2)}])$ as $[V^{(2)}] = [\widetilde{V}^{(2)}]$. To see that $[V^{(2)}]$ is full row rank, we note that $V^{(1)}$ is full row rank (Lemma~\ref{Lemma:manifold_defined}) and each successive rotational constraint within $V^{(2)}$ (counting around the polygon starting from the shared edge) involves at least one robot that does not belong to any previous constraints in $V^{(1)}$ or $V^{(2)}$. 
The in-plane constraint for $V^{(2)}$ further reduces the dimension of the solution set by constraining the rotational degree of freedom for $\mathcal{Q}_1$: there is only one orientation of $\mathcal{Q}_1$ (up to reflection, i.e. negative scaling) under which the edge shared by the two faces is orthogonal to $\bm{n}_2$. Therefore all constraints in $[V^{(2)}]$ are independent.

Inductive Hypothesis: Assume that the set of equations $[\widetilde{V}^{(k)}]\bm{x} = \bm{0}$ and $[V^{(k)}]\bm{x} = \bm{0}$ are equivalent and that $[V^{(k)}]$ is full row rank.

We now prove that $[\widetilde{V}^{(k+1)}]\bm{x} = \bm{0}$ and $[V^{(k+1)}]\bm{x} = \bm{0}$ are also equivalent and that $[V^{(k+1)}]$ is full row rank.
The inductive  hypothesis $[V^{(k)}]\bm{x} = \bm{0}$ encodes rules which assemble the faces of $\bm{\mathcal{Q}_k}$ (which form a tree by the ordering assumption) into a unique PPS.
Now, from our ordering assumption, $\mathcal{Q}_{k+1}$ must share exactly one edge (i.e., exactly 2 robots) with a face within $\bm{\mathcal{Q}_k}$, and that the edge must lie in a plane normal to $\bm{n}_{k+1}$. The tree structure of $\bm{\mathcal{Q}_k}$ prevents any additional shared edges. From Corollary \ref{corr:manifold_defined_red}, the rotational constraints encoded in $V^{(k+1)}$ are indeed sufficient to fully constrain $\mathcal{Q}_{k+1}$.
Thus, $[V^{(k+1)}]\bm{x} = \bm{0}$ and $[\widetilde{V}^{(k+1)}]\bm{x} = \bm{0}$ are equivalent.

To see that $[V^{(k+1)}]$ is full row rank, we note that each successive rotational constraint within $V^{(k+1)}$ involves at least one robot not represented by any of the existing constraints within $[V^{(k)}]$ or $V^{(k+1)}$. Then, by similar reasoning as for the base case, all constraints in $[V^{(k+1)}]$ are independent. 

By induction, both claims are proven.
 \end{proof}
\begin{remark}[Polyhedron Degrees of Freedom]
Note that $V$ has a total of $\left(3\sum\limits_{k=1}^{L}|\mathcal{V}_k| - 6L + 2\right)$ linearly independent equations for a total of $\left(3\sum\limits_{k=1}^{L}|\mathcal{V}_k| - 6L + 6\right)$ variables. Three degrees of freedom correspond to the three translational degrees of freedom. The final degree of freedom corresponds to scaling the polyhedron by some constant $\alpha \in \reals$. Applying a negative scaling factor may be interpreted as inverting the entire formation.
\end{remark}

\subsection{Polyhedral Formation Control Law}
Having characterized the desired formation as the null space of a full row rank matrix, we turn our attention to the control laws. The control law for each robot is given by the sum of a number of contributions, one for each face in which the robot represents a vertex. Specifically, let $\bm{u}_i^{(k)}$ denote the contribution to the control law for robot $i$ due to face $k$, obtained by generalizing the symmetric cyclic controller presented in Section \ref{general_cyclic}, namely:
\begin{equation}
\bm{u}_i^{(k)} \!=\! \sum\limits_{m=1}^{N_k} k^{(k)}_m \big[R_{m_s}^{(k)}(\bm{x}^{(k)}_{i+m} - \bm{x}^{(k)}_{i})  + R_{m_s}^{(k)^T}(\bm{x}^{(k)}_{i-m} - \bm{x}^{(k)}_{i})\big], 
\label{poly_cyclic_ind}
\end{equation}
where $N_k<|\mathcal{V}_k|-1$  is the look-ahead horizon for face $k$, $k_{m}^{(k)}>0$ is a gain, $R_{m_s}^{(k)} = R_{\eta_k}^TR_{m}^{(k)}R_{\eta_k}$ with $R_{m}^{(k)} = R_{m\pi/|\mathcal{V}_k|}$, and $\bm{x}^{(k)}_{i+m}$ and $\bm{x}^{(k)}_{i-m}$ are neighboring robots within face $k$ (modulo $|\mathcal{V}_k|$). The choice $R_{m}^{(k)} = R_{m\pi/|\mathcal{V}_k|}$ stems from the requirement that robots  converge to a polygon of fixed size within each face (see Remark \ref{Remark:cyclic_fixed}). 

The net control for each robot is then given by the superposition of the contributions from each face in which the robot is a vertex, i.e., 
\begin{equation} \label{poly_control_robot}
 \bm{u}_i = \sum\limits_{k:i_k\neq \infty}^{L} \bm{u}_i^{(k)}, \qquad i=1,\ldots,n. 
\end{equation}
The control laws in \eqref{poly_cyclic_ind} and \eqref{poly_control_robot} highlight the information requirements for each robot in the formation. In particular, we assume that each robot knows: (a) what face(s) it belongs to, (b) relative position with respect to its forward and rear neighbors within each face, (c) the number of robots in these faces, and (d) the rotation matrix $R_{\eta}$ that accompanies the relative position measurements. In a scenario where there are a large number of robots in the overall formation, these information requirements are fairly minimal, emphasizing the decentralized framework that underpins this work.

Using \eqref{poly_cyclic_ind} and \eqref{cyclic_modified}, the overall control vector stemming from face $k$ is given by
\[ 
\bm{u}^{(k)} =  -\sum\limits_{m=1}^{N_k} k_m^{(k)}\, \mathcal{L}_{m\eta}^{(k)}\bm{x}^{(k)} := -\mathcal{L}_{\eta}^{(k)}E^{(k)}\bm{x},
 \]
where $\mathcal{L}_{m\eta}^{(k)} = \mathcal{R}^{T}_{\eta_{k}}\mathcal{L}_{m}^{(k)}\mathcal{R}_{\eta_{k}} $ and $\mathcal{L}_{m}^{(k)} = {L}_m\otimes R_{m}^{(k)}  +  {L}_m^T\otimes R^{(k)^T}_{m}$. The overall closed-loop dynamics are then given by
\begin{equation} 
\dot{\bm{x}} = \bm{u} = \sum\limits_{k=1}^{L}  -E^{(k)^T}\mathcal{L}_{\eta}^{(k)}E^{(k)}\bm{x}.
\label{poly_dynamics} 
\end{equation}
For simplicity, we are considering zero internal dynamics. A sufficient condition for convergence to the desired formation, i.e., to the subspace $\mathcal M_n$, is provided by the following theorem.

\begin{theorem}[Polyhedron Convergence]\label{Theorem:poly_converge_th}
Assume 

\begin{equation} \label{3d_poly_converge}
J = \begin{bmatrix} V\left( \sum\limits_{k=1}^{L}  -E^{(k)^T}\mathcal{L}_{\eta}^{(k)}E^{(k)} \right) V^{T} \end{bmatrix} \prec 0. 
\end{equation}
Then, the closed-loop dynamics  \eqref{poly_dynamics} globally converge to a Johnson polyhedral formation, i.e., to the formation subspace $\mathcal M_n$.
\end{theorem}
\begin{proof}
Let $\bar{V}$ represent the orthonormal counterpart of $V$, whose existence is guaranteed given the results of Lemma \ref{Lemma:V_reduced}. First, we need to show that the dynamics in \eqref{poly_dynamics} are flow invariant. Indeed, as an immediate consequence of Remark \ref{Remark:cyclic_fixed} and the fact that the rotation angle is set equal to $m\pi/|\mathcal V_k|$, one has
\begin{equation}
\bm{x} \in \mathcal{M}_{n} \Rightarrow \bm{u}^{(k)}(\bm{x}) = \bm{0}, \qquad \text{for all } k = 1,\ldots,L,
\label{poly_prop}
\end{equation} 
and hence $\bm{x} \in \mathcal{M}_{n}  \Rightarrow \bar{V} \bm{u} =\bm{0}$, i.e., the dynamics in \eqref{poly_dynamics} are flow invariant.

Consider, then, the following auxiliary system for \eqref{poly_dynamics}:
\[ \dot{\bm{y}} = \bar{V} \left( \sum\limits_{k=1}^{L}  -E^{(k)^T}\mathcal{L}_{\eta}^{(k)} E^{(k)} \left(\bar V^T \bm{y} + \bar{U}^{T}\bar{U}\bm{x}_{p}\right) \right). \]
 Since the dynamics are flow invariant, and an invertible transformation exists between $V$ and $\bar{V}$, then by Remark \ref{V_Vbar_T} and by applying Corollary \ref{co:applied_contraction}, one obtains the claim. 

 \end{proof}

\begin{remark}[Conditions for Polyhedron Convergence]
Note that the condition given in \eqref{3d_poly_converge} is markedly more complex than \eqref{eq:eig_Theorem}, required for convergence to a single regular polygon. In particular, we now have the addition of cross-terms in the projected Jacobian due to the interactions between the various faces. Additionally, when working with the auxiliary system, we did not leverage the results of flow-invariant subspaces as they were introduced in Section \ref{sec:planar}. This is because if the global state vector $\bm{x}$ converges to one of the subspaces defined by $V^{(k)}$, then the global dynamics will not necessarily be flow-invariant with respect to the given face or any other faces in $\mathcal{M}_n$. This is due to two reasons: (a) only the rotational set of constraints are used to define faces other than the ones indexed by $k=1$ and $k=2$, which we know are insufficient on their own to describe a regular planar polygon, and (b) the symmetric cyclic control law for each robot as defined in \eqref{poly_control_robot} introduces coupling between various faces. Thus, convergence is predicated on the decay of the symmetric cyclic controller to zero for all faces.
\end{remark}

Figure \ref{fig:3dfor} shows simulation results demonstrating convergence to example 3D formations. 
\begin{figure}[htbp]
\vspace{-3mm}
\begin{subfigure}[t]{.24\textwidth}
	\centering
	\includegraphics[width=1\textwidth]{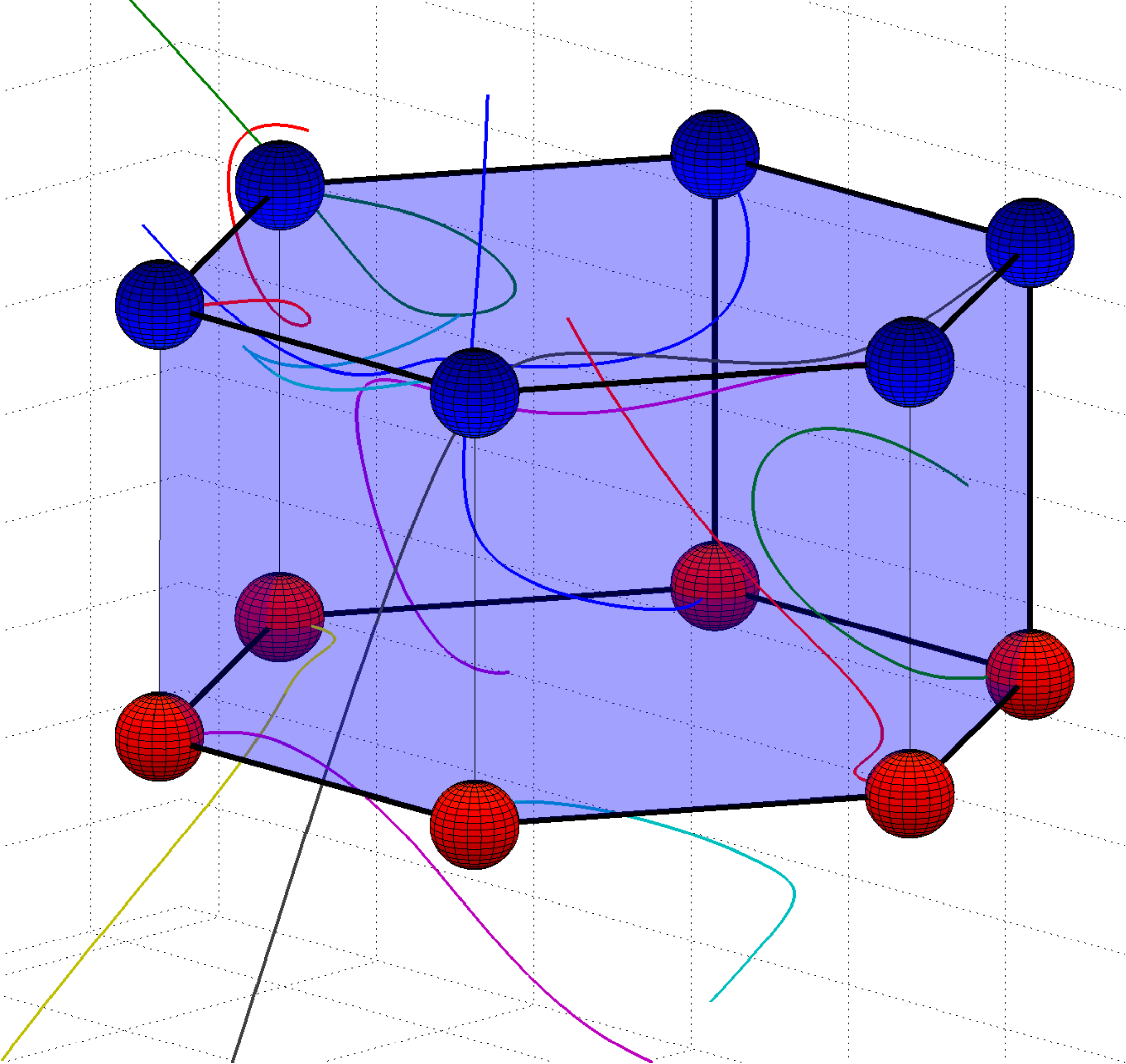}
	\caption{Hexagonal Box (12 robots)}
	\label{fig:hex_box} 
\end{subfigure}
\begin{subfigure}[t]{.24\textwidth}
	\centering
	\includegraphics[width=1\textwidth]{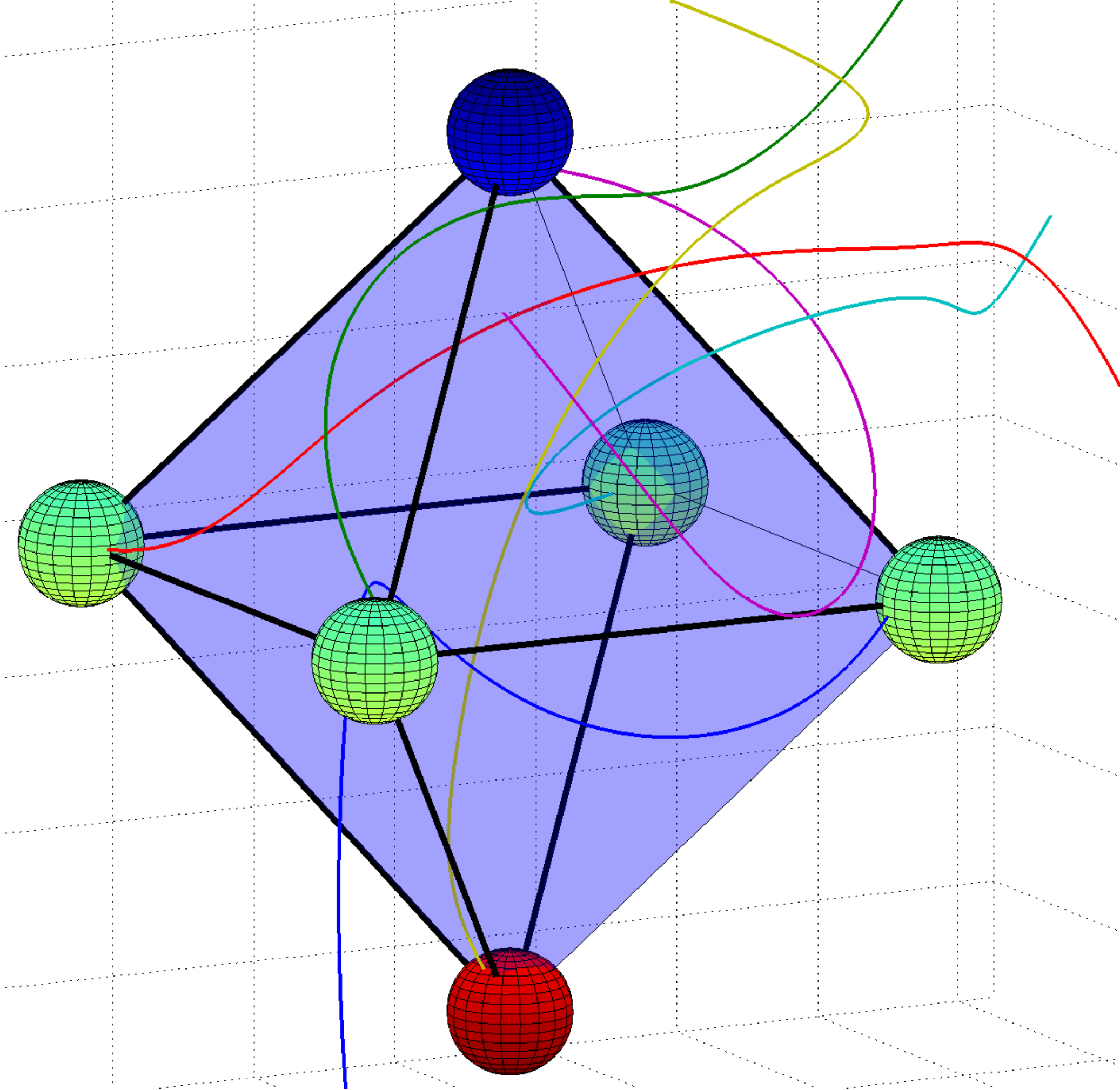}
	\caption{Octahedron (6 robots)}
	\label{fig:octahedron} 
\end{subfigure}
\caption{Convergence to polyhedra. Faces with bold edges correspond to the set $\mathcal{Q}$.}
\label{fig:3dfor}
\vspace{-2mm}
\end{figure}

\begin{remark}[Nonsymmetric Tree-Based Formations]
Lemma~\ref{lem:reduced_poly} sets forth the essential properties of polyhedral surface developments required for the preceding control law and convergence analysis to apply. That is, viable targets under this theory consist of any formation of robots expressable as the vertices of spatially-oriented regular polygons arranged according to a tree structure. Johnson solid formations are a particularly motivating example, given their symmetry, but other grid and mesh formations are also possible.
For instance, Figure \ref{fig:dome_shape} illustrates an example where a curved dome is represented by the vertices of a tree of squares ($\mathcal{Q}$), and the symmetric cyclic controller is defined and assembled using equations \eqref{poly_cyclic_ind} and \eqref{poly_control_robot}. Note that although there are 20 robots in this formation, each robot only requires relative position measurements with respect to three other robots (the robots in the top square require relative measurements with respect to four other robots) in the group. Thus, in addition to highlighting the applicability of our approach to the formation control of nonsymmetric formations, this example also underlines the sparsity of the decentralized formation control laws.
\begin{figure}[htbp]
\vspace{-5mm}
\begin{subfigure}[t]{.24\textwidth}
	\centering
	\includegraphics[width=1\textwidth]{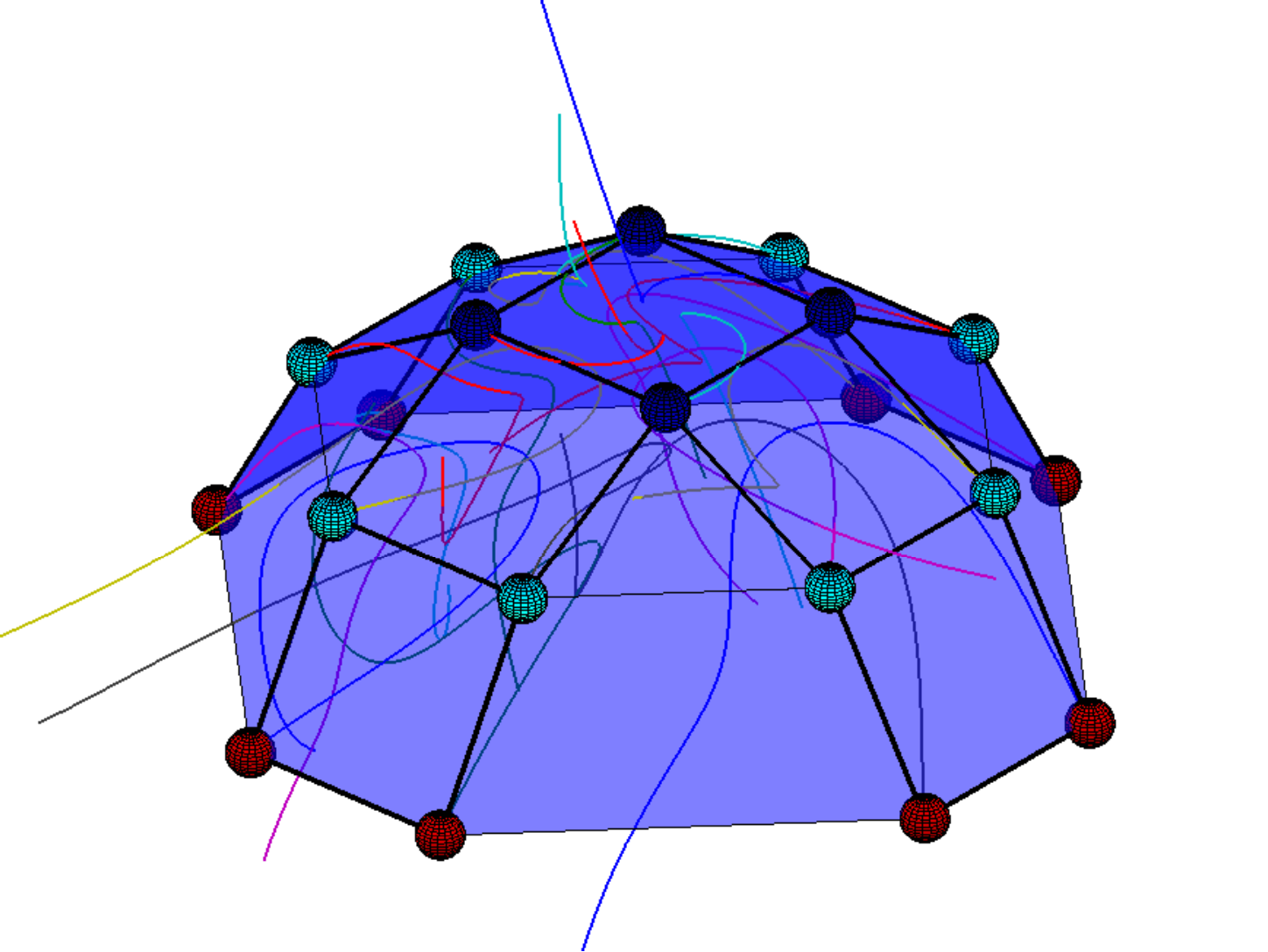}
	\caption{Dome shape formed by the convex hull of the robots.}
	\label{fig:dome} 
\end{subfigure}
\begin{subfigure}[t]{.24\textwidth}
	\centering
	\includegraphics[width=1\textwidth]{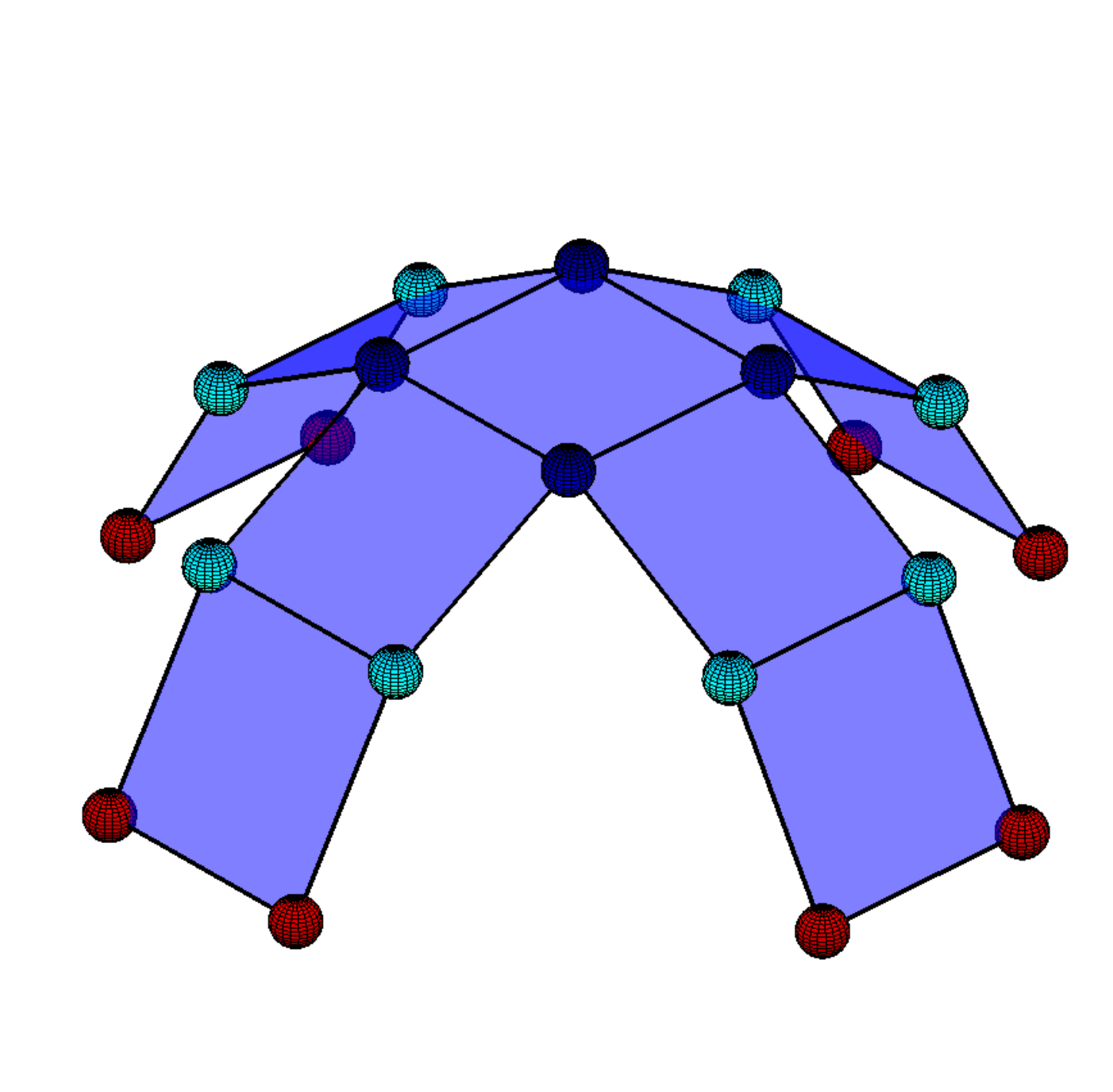}
	\caption{Tree of polygons (squares) $\mathcal{Q}$ used to derive the formation control laws.}
	\label{fig:octahedron} 
\end{subfigure}
\caption{Dome-shaped formation with 20 robots. The formation may be extended to collectively film sports or entertainment events.}
\label{fig:dome_shape}
\vspace{-5mm}
\end{figure}  
\end{remark}

In future work, we intend to investigate strategies for accommodating local robot failures. For instance, in the regular polygon case, a single point of failure simply requires the robots to reinitialize their indices within the group and subsequently converge to a smaller polygon. In the polyhedron case this can be problematic if the failed robot is coupled to multiple faces. In this scenario, a global restart may be required with a new configuration that fits the reduced set of robots.

\section{Extensions}\label{sec:additional}

In this section we present a number of extensions to the symmetric cyclic controller discussed  in Sections~\ref{sec:planar} and~\ref{sec:complex}. In Subsection~\ref{section_size}, we provide extensions to control the formation size, and investigate convergence properties under delays. Subsection~\ref{coll_avoid} details a heuristic control law for preventing collisions through use of Repulsive Potential Functions. Finally, in Subsection~\ref{sec:robust}, we study the disturbance rejection properties of the symmetric cyclic controller given in (\ref{eq:symmetric_control_simplified}) with respect to bounded additive disturbances.

\subsection{Control Over Formation Size}\label{section_size}
We first address control over the size of a regular polygon formation. By setting $\alpha_m = m\pi/n$ in (\ref{eq:symmetric_control}), the formation converges to a regular polygon of fixed, but uncontrolled size (in particular, the size depends on the initial conditions). In \cite{JLR-JJES:12}, the authors introduce an additional term to (\ref{eq:symmetric_control}) to control the size of the formation, however, the auxiliary system used to study convergence using partial contraction does not admit the necessary stationary zero solution. We present a different approach to the size control problem by simply modifying the rotation angle $\alpha_m$ embedded within the cyclic controller.

Let $\rho$ be the desired polygon side-length for the formation and define $p_i(\bm{x}) := 1 - \|\bm{x}_{i+1}-\bm{x}_i\|/\rho$ to be the normalized error in the inter-robot distance between robots $i$ and $i+1$, where $\|\cdot\|$ denotes the Euclidean norm. Let $\bar{p}(\bm{x})$ be the arithmetic average over all $p_i(\bm{x})$ (hereby referred to as the average inter-robot distance error), and define $f_{s}(\bar{p})$ to be a continuous odd scalar function that satisfies the properties: (a) $\bar{p}f_{s}(\bar{p})>0$ for $\bar{p}\neq 0$, (b) $|f_{s}(\bar{p})|\leq 1$, and (c) $\exists a \in \reals_{>0}$ such that $f'_{s}(0) > a$. The last (technical) property is required to ensure sufficient control authority near $\bar{p} = 0$ and will be formalized in Theorem \ref{thm:size_final}. Some examples for $f_s$ include the hyperbolic tangent function, or the saturation function ($|f_s(p_i)| = \min\{|p_i|, 1\}$). Consider the following modification to the symmetric cyclic control law introduced in (\ref{eq:symmetric_control}), now expressed as a \emph{hybrid} dynamical system:
\begin{equation}
\begin{split}
\bm{u}_i(t) = \sum\limits_{m = 1}^{N}k_m \big[R_m(\bm{x}) (\bm{x}_{i+m} \! &-\! \bm{x}_{i})\! + \!R_m(\bm{x})^T (\bm{x}_{i-m} \!- \!\bm{x}_{i})\big], \\ & t \in [k\tau, (k+1)\tau),\ k \in \mathbb{N}, 
\end{split}
\label{size_center_control}
\end{equation}
where $R_m(\bm{x})$ is a rotation matrix around $\bm{e}_z$ with rotation angle $m\pi/n + \alpha_s(\bm{x})$, where $\alpha_s(\bm{x}) := \alpha_{s_0}f_{s}(\bar{p}(\bm{x}((k-1)\tau)))$, $\alpha_{s_0} \in \reals_{>0}$ is a gain, and $\tau \in \reals_{>0}$ is a constant time lag. 

In equation \eqref{size_center_control} the rotation angle has been modified \emph{equally} for all robots by the amount $\alpha_{s}$, which is a function of the average inter-robot distance error $\bar{p}$. The time lag parameter $\tau$ accounts for the finite time required to compute $\bar{p}$, e.g., by using discrete average consensus techniques \cite{NL:96}. That is, the angle $\alpha_s$ used in the time interval  $[k\tau, (k+1)\tau)$ is computed using the average inter-robot distance error at time $(k-1)\tau$. To simplify notation, we will use $(\bm{x},\tau)$ to reference the continuous/discrete hybrid dynamics expressed in \eqref{size_center_control}.

As the rotation angle modification is the same for all robots, the net control law can still be expressed using circulant matrices:
\begin{equation}
\dot{\bm{x}} = \bm{u} = -\underbrace{\sum\limits_{m=1}^{N} k_m \, \underbrace{[L_m\otimes R_m(\bm{x}, \tau) \,+ L_m^T\otimes R_m(\bm{x}, \tau)^T ]}_{:=\mathcal{L}_{m}(\bm{x}, \tau)}}_{:=\mathcal{L}(\bm{x},\tau)}\bm{x}.
\label{symmetric_control_simplified_size}
\end{equation}
To analyze convergence to the desired formation, first define $\mathcal{M}_{n\rho}$ to be a subset of the original invariant subspace $\mathcal{M}_n$, corresponding to the space of regular polygon formations with inter-robot neighbor distance equal to $\rho$. Thus, 
\begin{equation}
\mathcal{M}_{n\rho} = \{\bm{x} : \bm{x} \in \mathcal{M}_{n} \wedge \|\bm{x}_{i+1}-\bm{x}_{i}\| = \rho, \forall i \}.
\label{size_manifold}
\end{equation}
To prove convergence to this subspace, we first start with a lemma to prove that the subspace $\mathcal{M}_n$ is indeed flow-invariant with respect to \eqref{symmetric_control_simplified_size}. 

\begin{lemma}[Flow-invariance] \label{Lemma:size_invariant}
The subspace $\mathcal{M}_n$ is flow-invariant with respect to the dynamics given in \eqref{symmetric_control_simplified_size}.
\end{lemma}
\begin{proof}
Given that the symmetric cyclic controller can still be expressed using circulant matrices as in \eqref{symmetric_control_simplified_size}, by Lemma \ref{cyclic_invariant}, the claim follows.
 \end{proof}

Having shown that $\mathcal{M}_n$ (which contains $\mathcal{M}_{n\rho}$)  is flow-invariant with respect to our modified dynamics, we now present sufficient conditions to ensure that $\bm{x}$ converges to the manifold $\mathcal{M}_{n\rho}$. The proof for the following theorem relies on several intermediate results which are detailed in \iftoggle{arxiv}{the appendix}{\cite{SS-ES-MP:15EVb}}. 

\begin{theorem}[Polygon Convergence with Desired Size]\label{thm:size_final}
Let $C = 2\Gamma$, where 
\[
	\Gamma = \begin{cases} \dfrac{\sqrt{2(1-\cos(2\pi/n))}}{\sin(\frac{\pi}{n})}\sum\limits_{m=1}^{N}k_m &\mbox{ if }  $n$ \text{ is even,} \\
					     \dfrac{\sqrt{2(1-\cos(2\pi/n))}}{2\sin(\frac{\pi}{2n})}\sum\limits_{m=1}^{N}k_m &\mbox{ if }  $n$ \text{ is odd.}
			\end{cases}
\]
Assume
\begin{equation}
\begin{split}
&\inf_{\alpha_m\in[m\pi/n-\alpha_{s_0}, m\pi/n+\alpha_{s_0}]}\bigg(\min\limits_{\substack{ 1\leq i\leq n \\ k\in\{-1,0,1\}}} \  \sum\limits_{m=1}^{N} k_m \, \lambda^{(m)}_{ik}\ \bigg)   > \! 0, 
\end{split}
\label{eig_Theorem_size}
\end{equation}
where $\lambda^{(m)}_{ik}$ has the form given in Theorem \ref{contraction_th_orig}, and suppose the time lag parameter $\tau$ satisfies the following bound:
\begin{equation}
	\tau < \min \left\{ \dfrac{1}{C}, \dfrac{1}{8CT}\right\},
\label{tau_constraint_final}
\end{equation}
where $T$ denotes a positive constant such that $(1/2)T|\bar{p}| \leq |\sin\left(\alpha_{s_0}f_s(\bar{p})\right)| \leq T |\bar{p}|$ for $|\bar{p}|<1$ (if $f_s$ is the saturation function, $T = \alpha_{s_0}$). 

Then the robots converge to a regular polygon with the desired side-length $\rho$. That is, $\bm{x}$ converges to the manifold $\mathcal{M}_{n\rho}$.
\end{theorem}

\begin{remark}[Time Lag Bound]
Note that the bound for the time lag parameter $\tau$ given in \eqref{tau_constraint_final} can be relaxed by using a value for $\Gamma$ that is strictly smaller than the one given in Theorem \ref{thm:size_final}. The derivation of this relaxed bound is detailed in \iftoggle{arxiv}{the appendix}{\cite{SS-ES-MP:15EVb}}. 
\end{remark}
\begin{remark}[Control Over Formation Center]
The control laws may be trivially extended to also permit control over the geometric center of the formation. For instance, let $\bm{x}_c$ represent the desired center and let $\bm{x}_0:= (1/n)\sum_{i=1}^{n}\bm{x}_i$ be the instantaneous geometric center. In similar spirit to the continuous/discrete dynamics in \eqref{size_center_control}, the center control law $\bm{u}_c$ for robot $i$ may be expressed as:
\[
	\bm{u}_c(t) = k_c \left(  \bm{x}_c - \bm{x}_0((k-1)\tau) \right), t \in [k\tau, (k+1)\tau), k\in \mathbb{N},
\]
where $k_c \in \reals_{>0}$ is a gain and the time lag parameter $\tau$ accounts for the finite time required to compute $\bm{x}_0$ in a decentralized fashion (e.g., using message passing). By leveraging analysis similar to that presented in the proof for Theorem \ref{thm:size_final}, it is possible to derive a bound on the steady state error $\|\bm{x}_c - \bm{x}_0\|$. 
\end{remark}
\begin{remark}[Size Control Extension to Polyhedra]
The size controller discussed in this subsection may be trivially extended to the formation control of polyhedra by decoupling the polygon $\mathcal{Q}_1$ from all other polygons $\mathcal{Q}$ in the PPS. The net control for all robots in $\mathcal{Q}_1$ is given by the modified symmetric cyclic controller in \eqref{symmetric_control_simplified_size} plus a ``rotational controller" that rotates the robots within the plane such that the first face has an orientation compatible with adjacent faces within the PPS. The control for all other robots in the polyhedron is given by the standard symmetric cyclic controller as given in \eqref{poly_cyclic_ind} and \eqref{poly_control_robot}. The proof for convergence is then also decoupled by treating $\mathcal{Q}_1$ as an independent polygon with respect to the remaining PPS. An example rotational controller along with a proof of correctness is provided in \iftoggle{arxiv}{the appendix}{\cite{SS-ES-MP:15EVb}}.
\end{remark}

\subsection{Collision Avoidance} \label{coll_avoid}

For collision avoidance, one must account for the size of each robot in the formation. In particular, we assume that each robot is completely contained within a sphere of radius $r_1$. To design  a collision avoidance controller, we make use of Repulsive-Potential-Functions (RPF) \cite{DHK-HW-SS:06}. 

Consider a ``collision-detection" zone around each robot parametrized by two variables $r_1<r_2$, where $r_2$ denotes the radius of the boundary of a robot's detection region. The RPF between any two robots $i$ and $j$ can be written as follows:
\begin{align}
V_{ij}(d_{ij}) :=&  -\dfrac{(r_2-r_1)^2}{r_1-d_{ij}} + 2(r_2-r_1)\log (r_1 - d_{ij}) - d_{ij} \nonumber\\	
							&	-v(r_1,r_2), \ \mbox{if }  r_{1} <d_{ij}\leq r_2,
\label{eq:potential_fnc}
\end{align}
 and equals zero for $d_{ij}>r_2$, where $d_{ij} = \|\bm{x}_j-\bm{x}_i\| $ and $v(r_1, r_2) := 2(r_2 - r_1)\log(r_1 - r_2) - r_1$ is a constant. The normalized force magnitude can be derived from the RPF by taking the derivative with respect to $d_{ij}$ giving:
\begin{align}
f_{c_{ij}}(d_{ij}) &= \begin{cases} \text{undefined} &\mbox{if } d_{ij} \leq r_1, \\
				\dfrac{(d_{ij}-r_2)^2}{d_{ij}(d_{ij}-r_1)^2}&\mbox{if } r_1 < d_{ij} \leq r_2, \\
				0 &\mbox{if } d_{ij} >r_2.
			\end{cases}
\label{eq:coll_force}
\end{align}

Notice that both $V_{ij}$ and $f_{c_{ij}}$ asymptotically approach $+\infty$ as $d_{ij} \rightarrow r_1$, and are $\mathcal{C}^1$ continuous for $d_{ij}>r_1$.  We now formulate the collision avoidance controller for robot $i$ as follows:
\begin{equation}
\label{eq:u_coll_i}
\bm{u}_{\mathrm{coll}_i} =- \sum\limits_{j\neq i }^{n}f_{c_{ij}}(d_{ij})(\bm{x}_j-\bm{x}_i).
\end{equation}
Thus the collision avoidance controller adds a non-zero (flow-invariant with respect to $\mathcal{M}_n$) ``escape" velocity to the existing symmetric cyclic controller for robot $i$ if and only if there exists another robot $j$ within the $i^{th}$ robot's collision avoidance zone.

Notice that $f_{c_{ij}}$ is unbounded as $d_{ij} \rightarrow r_1$ which makes convergence analysis with respect to the subspace $\mathcal{M}_n$ intractable. In particular, it becomes difficult to upper bound the maximum eigenvalue of the Jacobian of the collision avoidance controller. In lieu of this analysis, we present Monte Carlo simulation results in Section \ref{quad_cyclic} for six quadcopters attempting to converge to a polygon while implementing the collision avoidance controller above.

\subsection{Robustness Analysis}
\label{sec:robust}
To conclude our analysis of the symmetric cyclic control algorithm, we investigate the robustness properties of the closed-loop dynamics given in \eqref{eq:symmetric_control_simplified} in the presence of additive disturbances. We quantify robust performance by the Euclidean norm of the deviation of the perturbed system trajectory from the nominal contracting trajectory with respect to the metric $\bar{V}^{T}\bar{V}$, where $\bar{V}$ is the orthonormal counterpart of the nominal projection matrix as defined in \eqref{eq:Vr_formula}. The analysis presented in this section is restricted to zero internal dynamics (i.e., $\bm{g}(\bm{x}) = \bm{0}$) and to convergence to regular polygons. An analogous proof can be derived for the Johnson polyhedron case.

Let $\bm{z}(t) = \bar{V}\bm{x}(t)$ represent the dynamics of the unperturbed system under the action of the symmetric cyclic controller \eqref{eq:symmetric_control_simplified}. Then,
\begin{equation}
 \dot{\bm{z}}(t) = -\bar{V}\mathcal{L}\bm{x}(t). 
\label{eq:robust_nom}
\end{equation}
Consider the perturbed dynamics $\bm{z}_d(t)$ under an additive state and/or time-dependant disturbance: 
\begin{equation}
\dot{\bm{z}}_d(t) =  -\bar{V}\mathcal{L}\bm{x}_d(t) + \bm{d}(\bm{x}_d,t).
\label{eq:robust_pert}
\end{equation}
Here the disturbance is measured in the transformed $\bm{z}$ coordinates. Let $\bm{\delta}(t)$ represent the difference between the nominal and perturbed trajectories at time $t$. Specifically, let $\bm{\delta_z}(t) := \bm{z}_d(t) - \bm{z}(t)$ and $\bm{\delta_x}(t) := \bm{x}_d(t) - \bm{x}(t)$. Note that both $\bm{\delta_z}$ and $\bm{\delta_x}$ equal $\bm{0}$ at $t = 0$ as the trajectories are assumed to start from the same initial conditions. Then, 
\begin{equation}
\dot{\bm{\delta_z}}(t) = -\bar{V}\mathcal{L}\,\bm{\delta_x}(t) +  \bm{d}(\bm{x}_d,t).
\label{eq:robust_diff}
\end{equation}
Let  $\bar{R}^2(t) := \bm{\delta_z}^T(t)\bm{\delta_z}(t)$; $\bar{R}$ represents the Euclidean distance of the perturbed trajectory from the nominal one in the transformed coordinates, referred to as the \emph{formation error}. The following theorem characterizes the robustness of the symmetric cyclic controller.

\begin{theorem}[Cyclic Control Robustness]\label{Theorem:robustness_bound}
Assume that the nominal system is contracting with contraction rate $\Lambda = -\lambda_{\min}\left(\bar{V}\mathcal{L}\bar{V}^T\right) < 0$ and that the disturbance is norm bounded with  upper bound $\bar{d}$. Then the formation error  in the transformed $\bm{z}$ coordinates, that is $\bar{R}$, is upper bounded as
\begin{equation}
\bar{R}(t) \leq \dfrac{\bar{d}}{\Lambda}(e^{\Lambda t} - 1).
\label{robust_bound}
\end{equation}
\end{theorem}
\begin{proof}
Starting from (\ref{eq:robust_diff}), $\bm{\delta_x} = (\bar{V}^T\bar{V} + \bar{U}^T\bar{U})\bm{\delta_x} = \bar{V}^T\bm{\delta_z} + \bar{U}^T\bar{U}\bm{\delta_x}$. Now, $\bar{U}^T\bar{U}\bm{\delta_x} \in \mathcal{M}_n$, so $\bar{V}\mathcal{L}\,(\bar{U}^T\bar{U}\,\bm{\delta_x}) = \bm{0}$. Thus, 
\[ 
\dot{\bm{\delta_z}}(t) = -\bar{V}\mathcal{L} \bar{V}^T\bm{\delta_z} +  \bm{d}(\bm{x}_d,t).
 \]
Pre-multiplying by $2\bm{\delta_z}(t)$ we obtain:
\[ 
2\bm{\delta_z}^T\dot{\bm{\delta_z}} = \dfrac{d(\bm{\delta_z}^T\bm{\delta_z})}{dt} = -2\bm{\delta_z}^T\left(\bar{V}\mathcal{L}\bar{V}^T\right)\bm{\delta_z} +  2\bm{\delta_z}^T\bm{d}(\bm{x}_d,t).
 \]
Letting $\bar{R}^2 = \bm{\delta_z}^T\bm{\delta_z}$, one can write
\[ 
\dfrac{d\bar{R}^2}{dt} = 2\bar{R}\dfrac{d\bar{R}}{dt} =  -2\bm{\delta_z}^T\left(\bar{V}\mathcal{L}\bar{V}^T\right)\bm{\delta_z} +  2\bm{\delta_z}^T\bm{d}(\bm{x}_d,t).
\]
Noting that:
\[ 
-\bm{\delta_z}^T\left(\bar{V}\mathcal{L}\bar{V}^T\right)\bm{\delta_z} \leq -\lambda_{\min}\left(\bar{V}\mathcal{L}\bar{V}^T\right)\bm{\delta_z}^T\bm{\delta_z} = \Lambda\bar{R}^2,
\]
and
\[ 
\bm{\delta_z}^T\bm{d}(\bm{x}_d,t) \leq \|\bm{\delta_z}\|\|\bm{d}(\bm{x}_d,t)\| = \bar{R}\|\bm{d}(\bm{x}_d,t)\| ,
\] 
we get:
\[ 
\bar{R} \dot{\bar{R}} \leq \Lambda\bar{R}^2 + \bar{R}\|\bm{d}(\bm{x}_d,t)\| .
 \]
Thus, 
\begin{subequations}
\begin{align}
\dot{\bar{R}} &\leq \Lambda\bar{R} + \|\bm{d}(\bm{x}_d,t)\|, \label{robust_diff_bound} \\ 
&\leq \Lambda \bar{R} + \bar{d} \label{robust_diff_bound2}
\end{align}
\end{subequations}
By the comparison theorem \cite{HKK:02}, the claim follows. 
 \end{proof}
\begin{remark}
The bound $\bar{R}(t)$ describes the time-varying deviation from a nominal contracting trajectory. In the limit $t\rightarrow \infty$, the steady state bound $\bar{R}_{ss} := \bar{d}/|\Lambda|$ provides a measure of the steady state formation error. Physically, it represents a bound on the degree of asymmetry at each node of the formation, as measured by the set of linear constraints used to define the subspace, that is, \eqref{manifold_def1} and \eqref{manifold_def2}. The inverse proportionality between $\bar{R}_{ss}$ and $|\Lambda|$ highlights a trade-off between increased robustness, control saturation, and sensor requirements. 
\end{remark}


Figure \ref{fig:robust_plot} provides insights into the tightness of the bound. Here, disturbances were introduced into the system via error in the control rotation angles. The rotation angle in $R_{m}$ for robot $i$ is equal to $\tilde{\alpha}_i(t) = \alpha_0 + \delta_{\alpha_i}(t)$ where $\alpha_0 = m\pi/n$ is the nominal control angle and $\delta_{\alpha_i}(t)$ is a perturbation, sampled randomly from a uniform distribution over the range $[\underline{\delta}_{\alpha_i},\bar{\delta}_{\alpha_i}]$ where $\underline{\delta}_{\alpha_i},\ \bar{\delta}_{\alpha_i}\in[-1, 1]^{\circ}$. The plotted curves in Figure \ref{fig:robust_plot_actual} correspond to the actual formation error $\bar{R}$, the bound in \eqref{robust_bound}, and the numeric integration of the expression in \eqref{robust_diff_bound}. Figure \ref{fig:robust_plot_percent} quantifies the tightness of \eqref{robust_bound} and \eqref{robust_diff_bound} by plotting the percentage difference with respect to the {actual deviation. 

\begin{figure}[htbp]
\vspace{-3mm}
\centering
\begin{subfigure}[t]{0.23\textwidth}
	\includegraphics[width=1\textwidth]{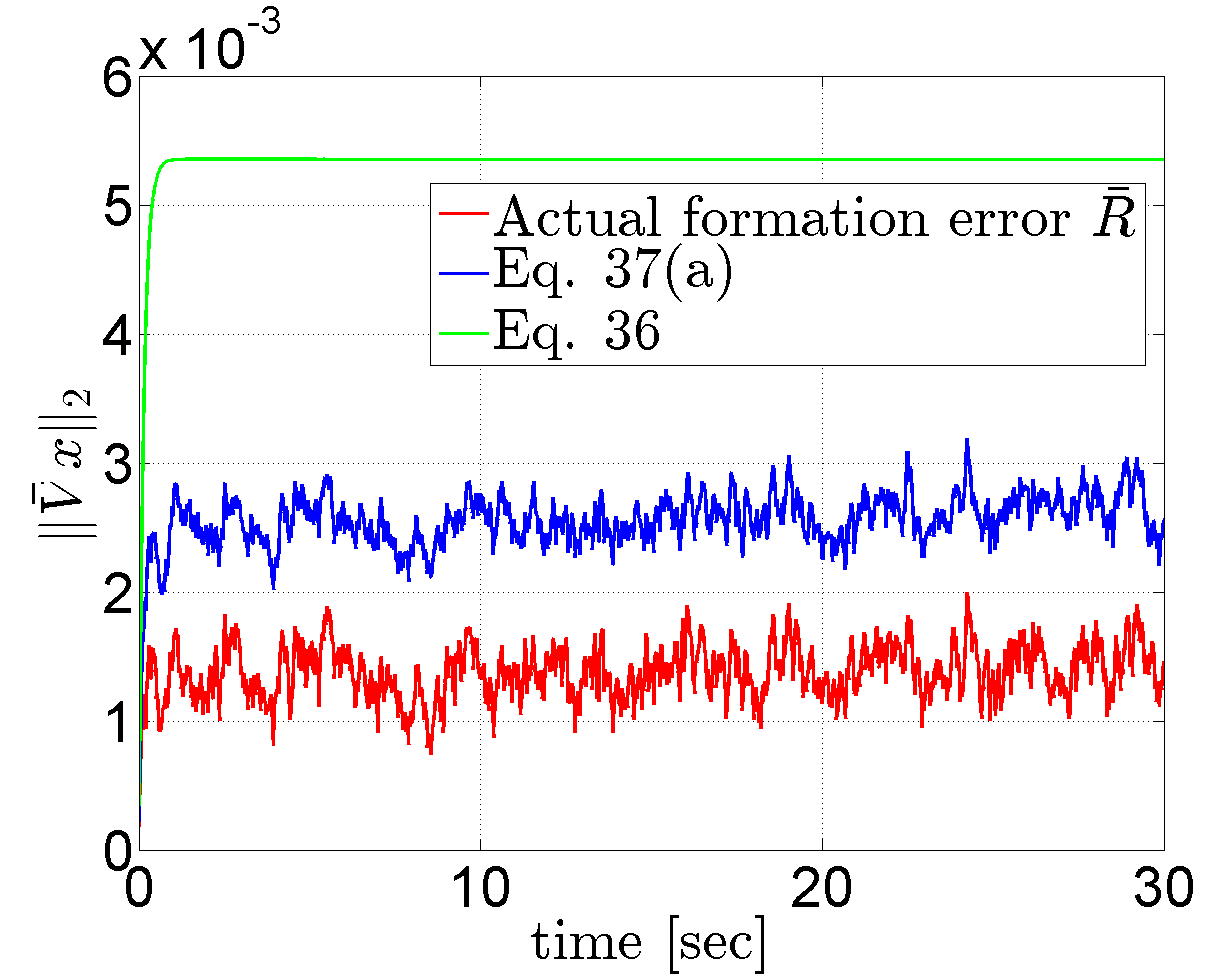}
	\caption{Comparison for $\|\bm{\delta}_{z}\|$.}
	\label{fig:robust_plot_actual} 
\end{subfigure}
\begin{subfigure}[t]{0.23\textwidth}
	\includegraphics[width=1.1\textwidth]{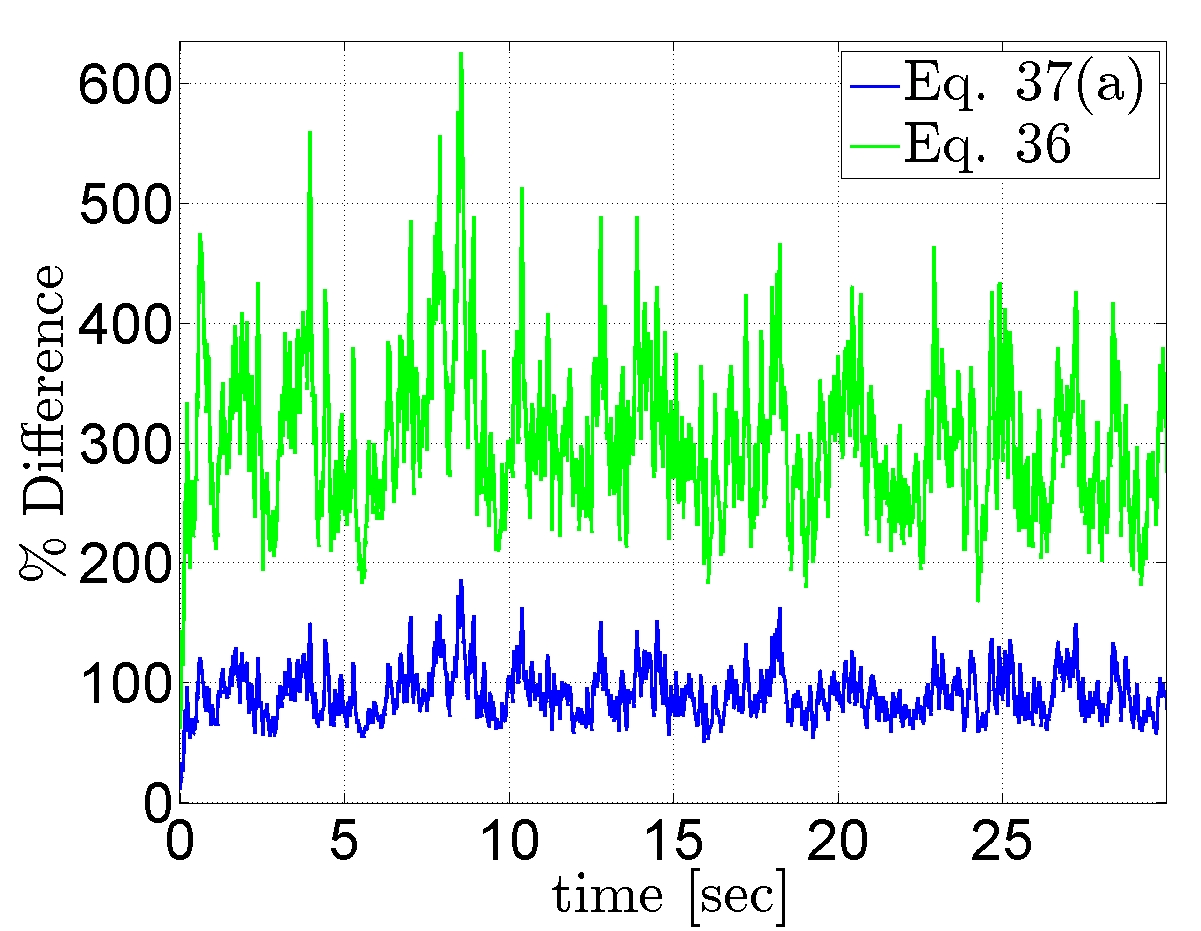}
	\caption{Bound tightness.}
	\label{fig:robust_plot_percent} 
\end{subfigure}
\caption{Tightness for robustness bounds. Parameters: $n = 6, N = 2, \Lambda = -6.928,\bar{d} = 0.065$. }
\label{fig:robust_plot}
\vspace{-2mm}
\end{figure}

\section{Formation Control of Quadcopters}\label{quad_cyclic}

In this section we demonstrate the use of the symmetric cyclic control laws developed in Sections \ref{sec:planar},  \ref{sec:complex} and \ref{section_size} on a fleet of quadcopters. 

\subsection{Quadcopter Dynamics and Control}

Let $\bm{x}_i$ and $\bm{v}_i$ denote the position and inertial velocity respectively of the $i^{th}$ quadcopter, resolved in an inertial frame with a downward pointing $Z$ axis. Denote $\bm{q}_i$ to be the attitude quaternion of the $i^{th}$ quadcopter with associated Euler angles: Yaw ($\psi_i$), Pitch ($\theta_i$), and Roll ($\phi_i$) with rotation order $ZYX$. The angular velocity of the $i^{th}$ quadcopter with respect to the inertial frame, resolved in the quadcopter body frame, is given by $\bm{\omega}_i$. Let $R_{ib}^i$ denote the rotation matrix to transform vectors in the inertial frame to the $i^{th}$ quadcopter body frame. Finally let $m_{q_i}$ and $I_i$ denote the mass and moment of inertia tensor, respectively, for the $i^{th}$ quadcopter. 

The translational dynamics of the quadcopter are given by:
\begin{equation}
\begin{split}
	\dot{\bm{x}}_i &= \bm{v}_i, \\
	\dot{\bm{v}}_i &= R_{ib}^{i^T} \begin{bmatrix} 0 \\ 0 \\ -T_i/m_{q_i} \end{bmatrix} + \begin{bmatrix} 0 \\ 0 \\ g \end{bmatrix},
\label{quad_trans}
\end{split}
\end{equation}
where $T_i$ is the net thrust produced by all four propellers and $g$ is the gravitational acceleration. Note that the direction of positive thrust is opposite the body $z-$axis. The rotational dynamics of the quadcopter are given by:
\begin{equation}
\begin{split}
	\dot{\bm{q}}_i &= \dfrac{1}{2}\begin{bmatrix}0 \\ \bm{\omega}_i \end{bmatrix} \odot \bm{q}_i, \\
	\dot{\bm{\omega}}_{i} &= I_i^{-1}\bigg( \bm{M}_{i} - \bm{\omega}_{i} \times I_i\bm{\omega}_{i} \bigg),
\label{quad_rot}
\end{split}
\end{equation}
where $\odot$ denotes the quaternion product and $\bm{M}_i$ is the net torque on the quadrotor as a result of differential thrust forces and reaction moments generated by the motor/propeller pairs. As the dynamics of the motors are significantly faster than those of the quadcopter as a rigid body, we assume that the net thrust force and moment generated by the propellers can be commanded instantaneously. 

Notice that the rotation matrix $R_{ib}^i$ couples the attitude and translational dynamics, thereby resulting in a fairly complex nonlinear dynamical system. In order to incorporate the symmetric cyclic control laws for formation control of quadcopters, we employ a two-level control hierarchy where, at the high level, the symmetric cyclic control laws in (\ref{eq:symmetric_control}), \eqref{poly_control_robot} or \eqref{size_center_control} are used to generate a \emph{desired} inertial velocity reference for each quadcopter. The lower level control system then simply tracks this reference, and incorporates both thrust and attitude control. Several velocity tracking controllers for quadcopters have already been presented in the literature \cite{RM-VK-PC:12,AK-DM-CP-VK:13}. In \iftoggle{arxiv}{the appendix}{\cite{SS-ES-MP:15EVb}} we present an example controller used for numerical simulations. 

\subsection{Numerical Simulations and Experiments}

We now present simulation and experimental results demonstrating the application of the above two-level control hierarchy on a group of six quadcopters. The quadcopters were initialized randomly in hover mode ($\bm{v}_i = \bm{0}$) around the point $[0,0,-2]^T$. The size control parameters were: desired inter-quadcopter distance $\rho = 2.0m$, function $f_s(\bar{p}) = \tanh(\bar{p})$, and gain $\alpha_{s_0} = 5\pi/180$. The desired formation center $\bm{x}_c$ was selected as the point $[0,0,-10]^{T}m$. The symmetric cyclic controller used a look-ahead horizon $N= 1$ and gain $k_1 = 0.5$. The desired plane of convergence was tilted $42^o$ with respect to the horizontal. We assumed a time lag $\tau = 0.1s$. The collision avoidance parameters were $r_1 = 0.4m$,  $r_2 = 1.2m$, and the controller was slightly modified to take advantage of relative velocity measurements:
\begin{equation}
	\begin{split}
	\bm{u}_{\mathrm{coll}_i} = - &\sum_{j\neq i}^{n} \dfrac{(d_{ij}-r_2)^2}{d_{ij}(d_{ij}-r_1)^2}(\bm{x}_j - \bm{x}_i),\ \\ &\mbox{if } r_1 < d_{ij} \leq r_2 \mbox{ and } v_{s_{ij}} <0,
	\end{split}																		
\end{equation}
where $v_{s_{ij}} =   (\bm{v}_{j} - \bm{v}_i)\cdot (\bm{x}_{j}-\bm{x}_{i})/ d_{ij}$ is the line-of-sight velocity between robots $i$ and $j$. Additionally, the desired velocity (given by the sum of the cyclic, center and collision avoidance controllers) was subject to saturation constraints to prevent aggressive maneuvering and increased collision risk. In particular, each quadcopter was limited to a speed of $3m/s$. Figure \ref{fig:quad_sim} shows the resulting 3D trajectories as well as a shaded plane that outlines the final positions of the quadcopters. Figure \ref{fig:quad_sim2} shows the evolution of the inter-quadcopter distances with respect to quadcopter 1. Notice how two neighboring quadcopters entered the collision detection region for quadcopter 1 before converging to the desired separation.

\begin{figure}[htbp]
\begin{subfigure}[t]{0.22\textwidth}
	\includegraphics[width=1\textwidth]{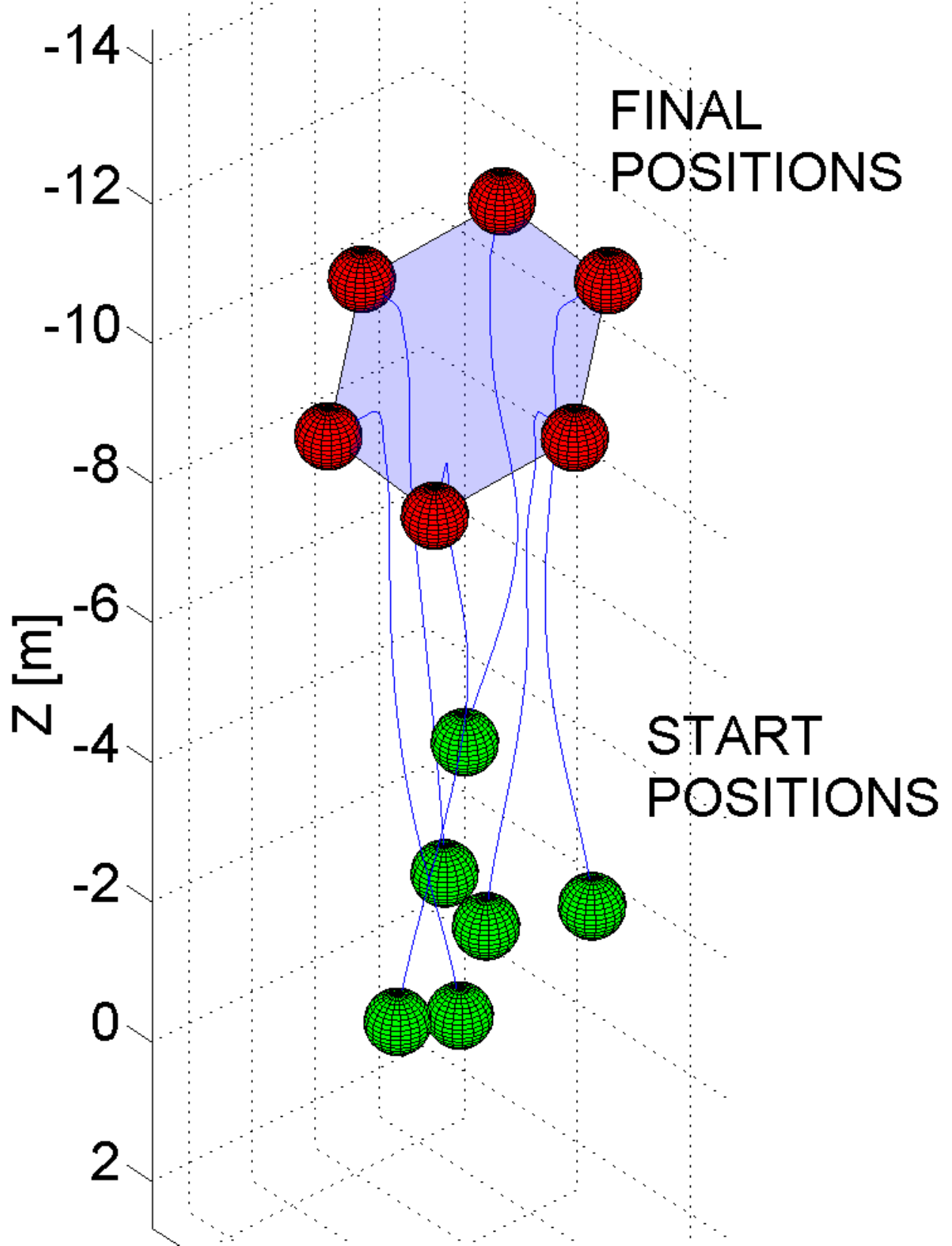}
	\caption{Quadcopter 3D trajectories with size and center control.}
	\label{fig:quad_sim} 
\end{subfigure}
\begin{subfigure}[t]{0.26\textwidth}
	\includegraphics[width=1.1\textwidth]{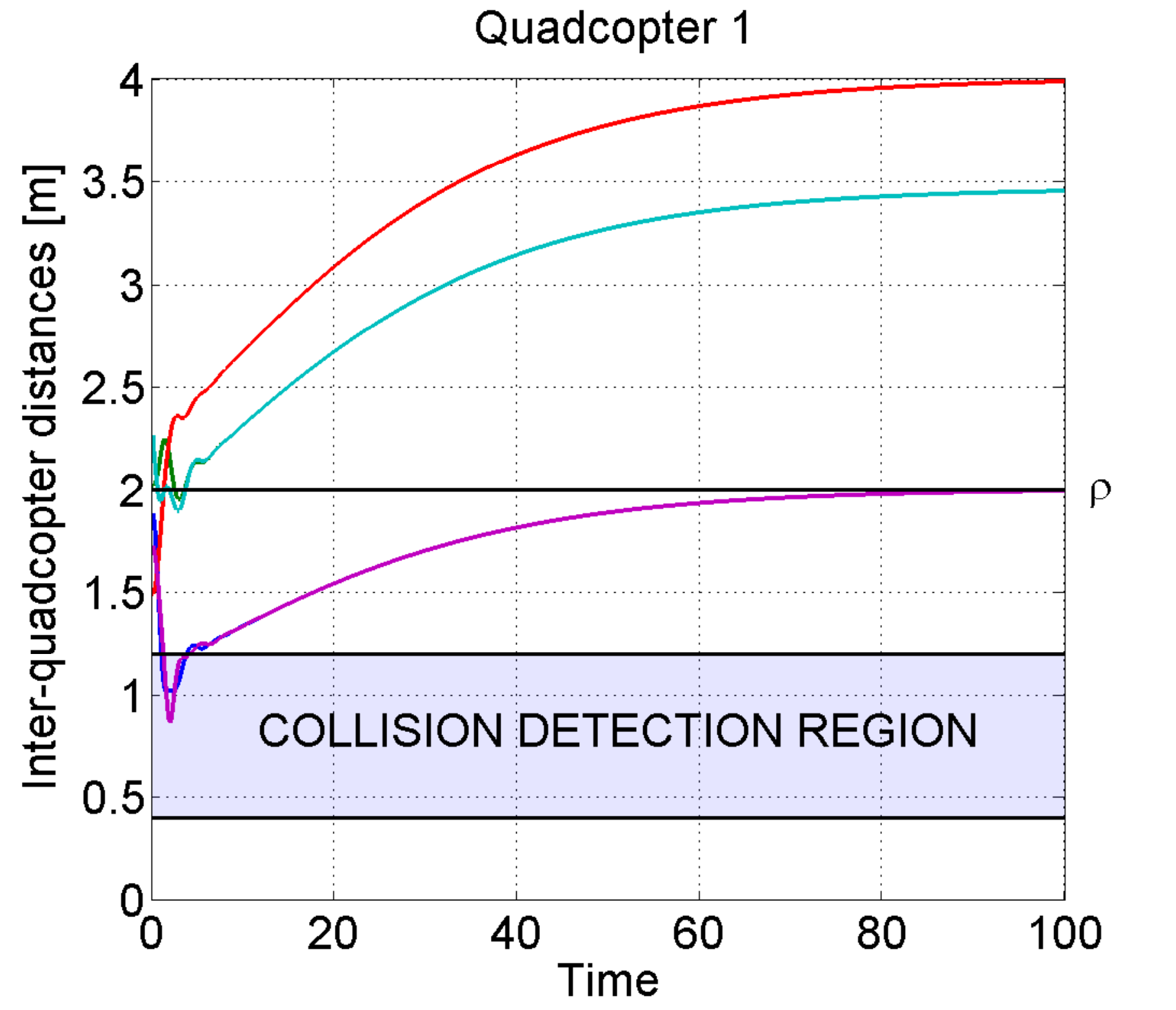}
	\caption{Inter-quadcopter distances (in metres) for Quadcopter 1.}
	\label{fig:quad_sim2} 
\end{subfigure}
\caption{Planar formation simulation.}
\end{figure}

To test the collision avoidance controller, we ran a Monte Carlo simulation with 100 samples randomizing the initial position of the quadcopters within a ball of radius $5m$ centered at $\bm{x}_c$ (all simulations initialized the quadcopters in hover mode). Additionally, upon initialization, the quadcopters were assigned indices such that their projection onto the desired plane of convergence formed a non-intersecting polygon. This is an example of a systematic (and decentralizable) method for the quadcopters to identify their nearest neighbors for symmetric cyclic control. In spite of the same speed constraint of $3m/s$, no collisions were recorded while the formation consistently converged in terms of shape, size and geometric center. 

For experimental validation, we implemented the control algorithms on a group of four quadcopters, each equipped with a PX4 Pixhawk autopilot \cite{PX4:15}. The low level velocity tracking controller detailed in \iftoggle{arxiv}{the appendix}{\cite{SS-ES-MP:15EVb}} is virtually identical to the on-board Pixhawk controller. Relative position measurements between neighboring quadcopters were provided by a Vicon motion capture system. The desired inter-robot neighbor separation $\rho$ was set to $1.65m$. Due to limitations in testing space, a less aggressive collision avoidance controller was implemented:
\[
	\bm{u}_{\mathrm{coll}_i} = - \sum_{j\neq i}^{n} k_{\mathrm{coll}}\tanh(\rho - d_{ij}) (\bm{x}_j - \bm{x}_i),\  \mbox{if } d_{ij} \leq r_2,
\]
where $k_{\mathrm{coll}} $ is a gain with value $1.2$, and $r_2 = 0.9m$. In addition, the speed constraint was further restricted to $0.7m/s$, while the cyclic ($k_1$) and center ($k_c$) control gains were set to $0.25$ and $0.15$ respectively. Note that in all of the experiments where the quadcopters were initialized outside of the collision detection regions, the collision avoidance controller did not become active.

Figures \ref{fig:quad_exp_1a} and \ref{fig:quad_exp_1b} show the quadcopters in initial and final (tetrahedron) configurations respectively. Figure \ref{fig:quad_exp_1c} shows the time evolution of the formation error $\|\bar{V}\bm{x}\|$ over all three phases of flight -- autonomous takeoff, formation control, and landing. Figures \ref{fig:quad_exp_2b} and \ref{fig:quad_exp_2c} show the time evolution of the inter-quadcopter distances and the error in geometric center. The quadcopters demonstrated relatively quick convergence in terms of formation shape, size and center, thereby validating the two-level formation control methodology \footnote{Videos of the experiments are available at \url{https://www.youtube.com/playlist?list=PL8-2mtIlFIJr4fVjJeyQqymbkKfhsKKYX}}. 
\begin{figure}[htbp]
\centering
\begin{subfigure}[t]{0.4\textwidth}
	\includegraphics[width=1\textwidth]{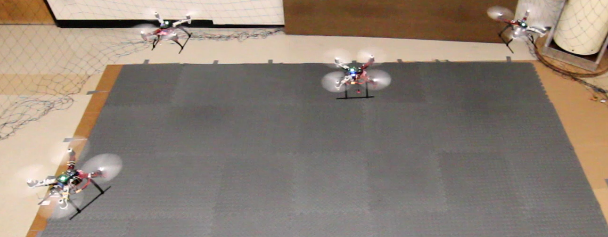}
	\caption{Quadcopters upon initialization in hover mode.}
	\label{fig:quad_exp_1a} 
\end{subfigure}\\
\begin{subfigure}[t]{0.3\textwidth}
	\includegraphics[width=1\textwidth]{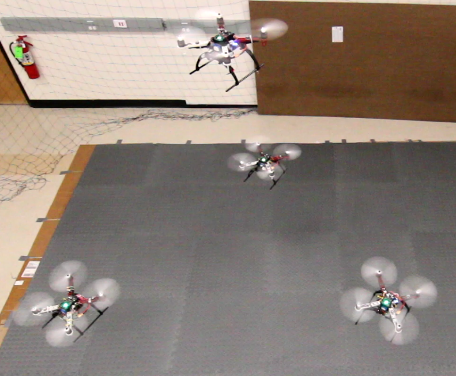}
	\caption{Quadcopters in final tetrahedron configuration.}
	\label{fig:quad_exp_1b} 
\end{subfigure}\\
\begin{subfigure}[t]{0.45\textwidth}
	\includegraphics[width=1\textwidth]{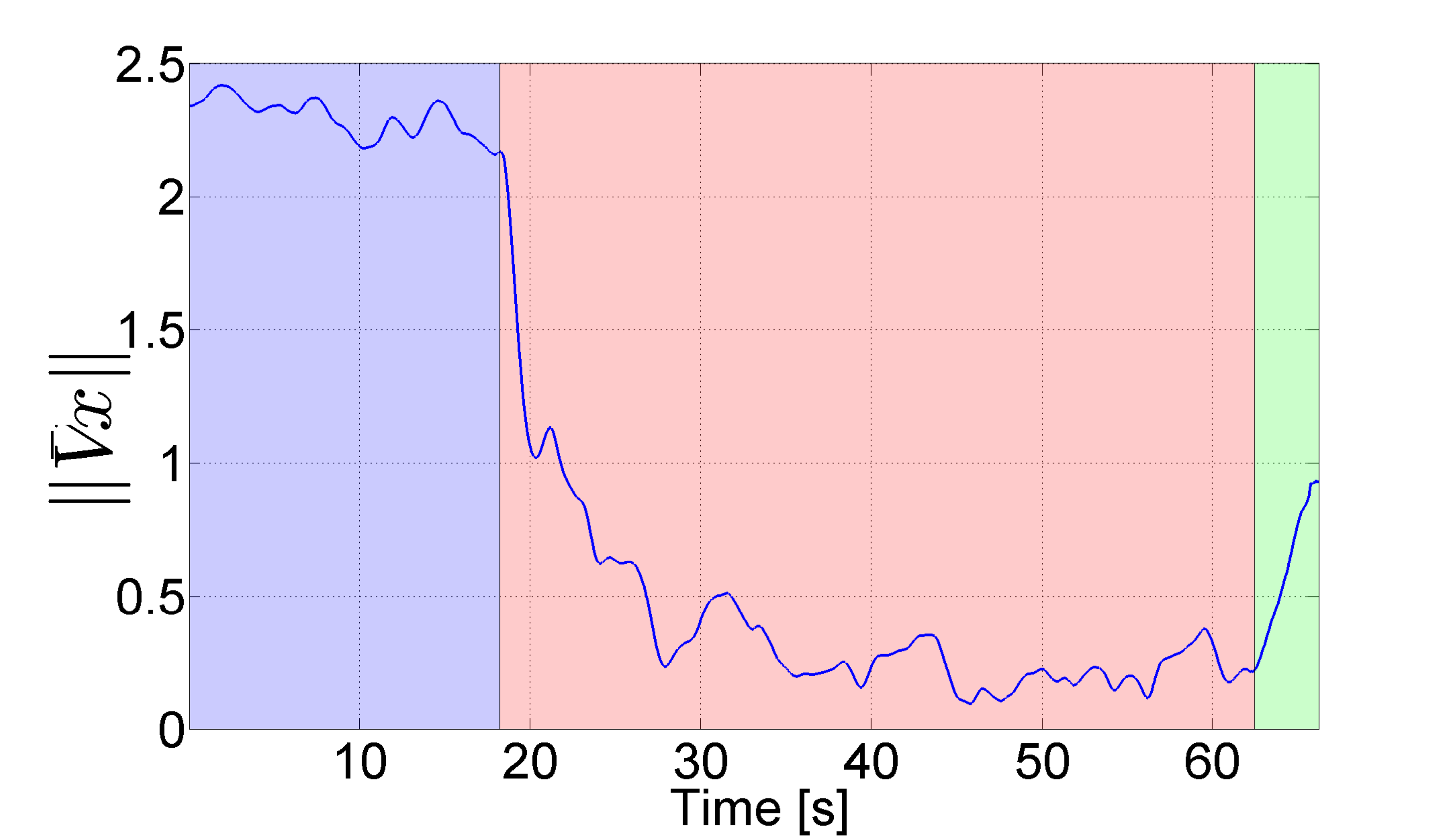}
	\caption{Tetrahedron formation error. The three shaded regions indicate the flight modes (from left to right: takeoff, formation control, landing).}
	\label{fig:quad_exp_1c} 
\end{subfigure}
\caption{Quadcopter tetrahedron formation.}
\end{figure}
\begin{figure}[h]
\centering
\begin{subfigure}[t]{0.5\textwidth}
	\includegraphics[width=1\textwidth]{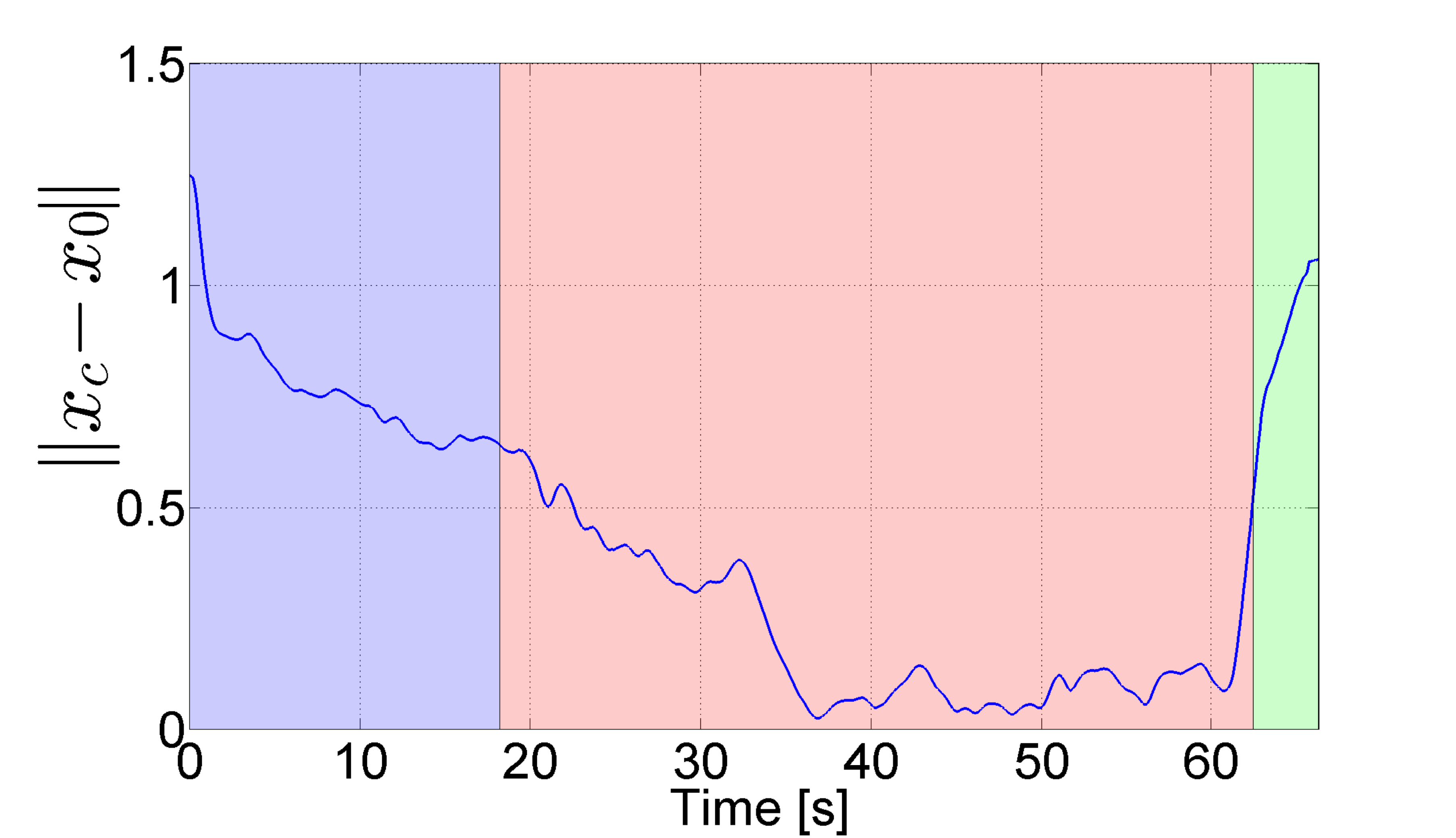}
	\caption{Error in geometric center.}
	\label{fig:quad_exp_2b} 
\end{subfigure}\\
\begin{subfigure}[t]{0.5\textwidth}
	\includegraphics[width=1\textwidth]{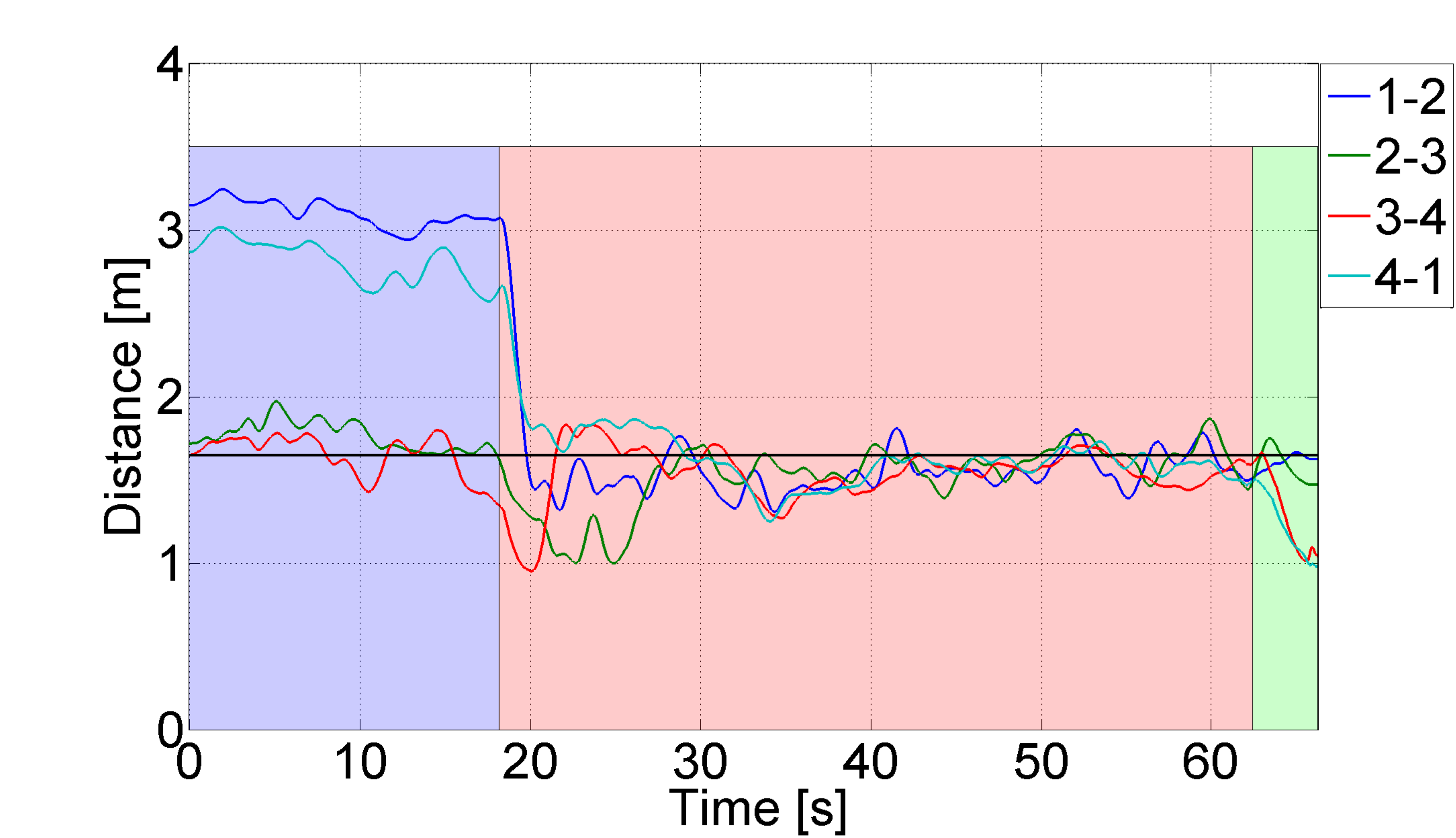}
	\caption{Inter-quadcopter neighbor separations. $\rho = 1.65m$.}
	\label{fig:quad_exp_2c} 
\end{subfigure}
\caption{Formation accuracy.}
\end{figure}

\section{Conclusions}\label{sec:conclusions}
In this paper we studied decentralized control algorithms for 3D formations such as regular polygons, Johnson solids, and more generally, connected regular polygonal meshes. The algorithms are desirable for their simplicity in that they only require relative position measurements between neighbors and an agreement between the robots on the desired orientation in space. Our approach leveraged the mathematical properties of cyclic pursuit along with results from contraction and partial contraction theory. In addition, we provided extensions to the algorithm for controlling the formation size and preventing collisions, and quantified the robustness of the symmetric cyclic controller under bounded additive disturbances. Finally, we validated our algorithms via numerical simulations and hardware experiments on a fleet of quadcopters. 

This work allows several possible avenues for extension. First, we would like to study  (sub-)optimality properties of the symmetric cyclic controller by introducing an objective cost function to be minimized (e.g., time or control effort), and using it to guide gain selection and/or look-ahead horizon. Second, we plan to investigate various examples for internal dynamics $\bm{g}(\bm{x})$ and their interplay with the symmetric cyclic controller when analyzing convergence. Third, the modification of the symmetric cyclic controller for formation size is a little cumbersome in that it requires all robots to compute a shared quantity. Consequently, we would like to investigate size controllers that 
do not rely upon consensus techniques.

\bibliographystyle{IEEEtran}
\bibliography{../../../bib/alias,../../../bib/main}

\vspace{-0.9cm}
\begin{IEEEbiography}[{\includegraphics[width=\textwidth]{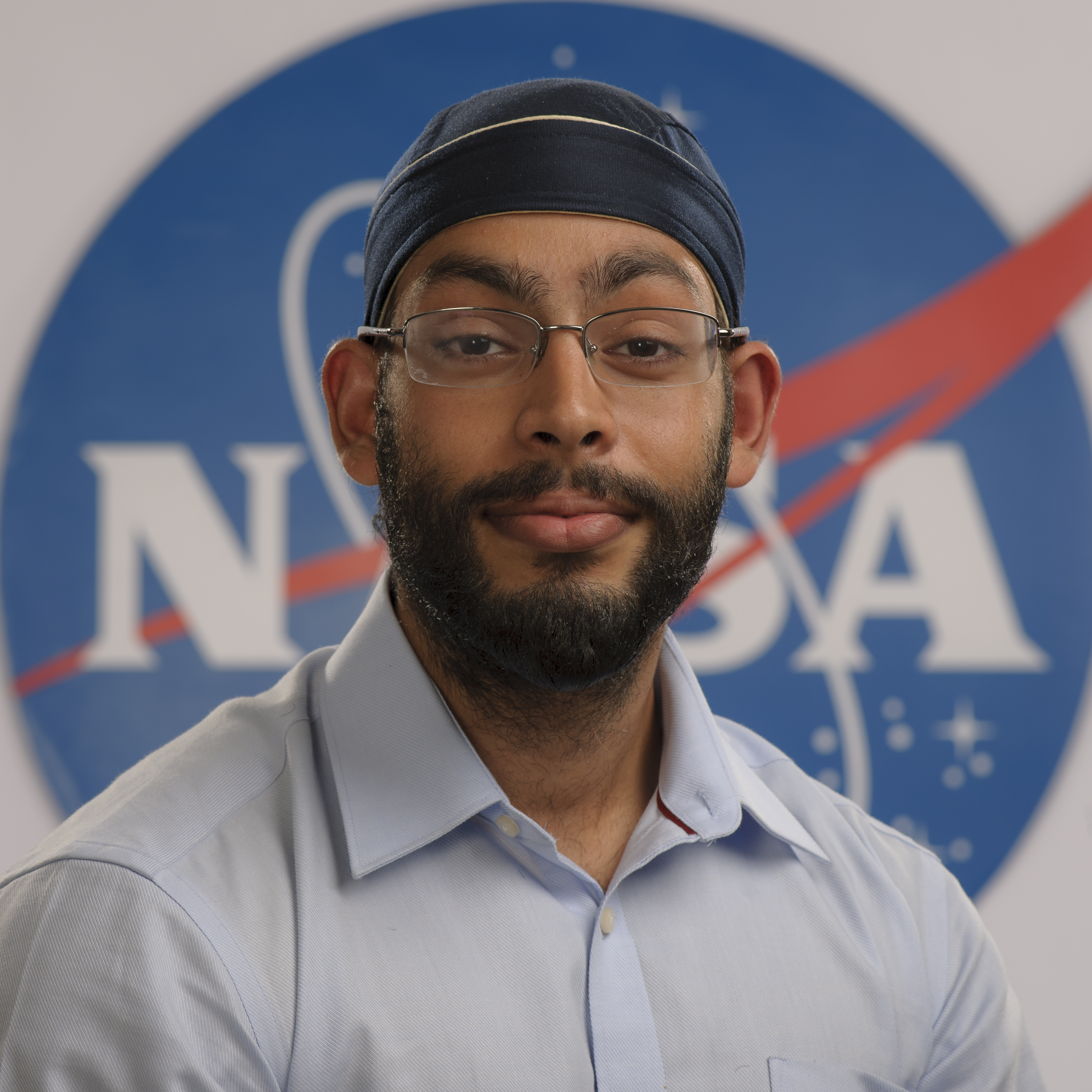}}]%
{Sumeet Singh} is is a Ph.D. candidate in the Autonomous Systems Laboratory at Stanford University. He received a B.Eng. in Mechanical Engineering and a Diploma of Music (Performance) from the University of Melbourne in 2012, and a M.Sc. in Aeronautics and Astronautics from Stanford University in 2015. Prior to joining Stanford, Sumeet worked in the Berkeley Micromechanical Analysis and Design lab at University of California Berkeley in 2011, and the Aeromechanics Branch at NASA Ames in 2013. Sumeet�'s current research interests include decentralized algorithms for formation control, and stochastic and robust nonlinear model predictive control, with specific applications to spacecraft control.
\end{IEEEbiography}
\begin{IEEEbiography}[{\includegraphics[width=\textwidth]{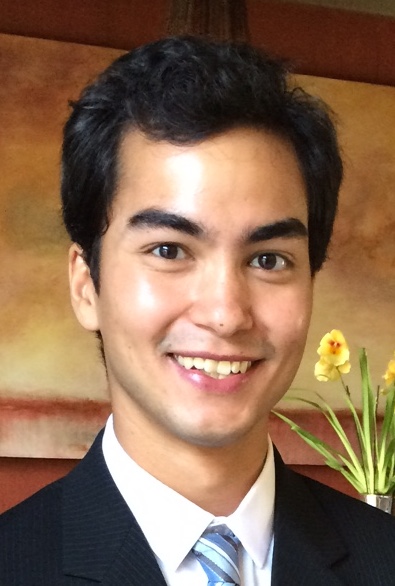}}]%
{Edward Schmerling} received the B.Sc. degree in mathematics and physics from Stanford University, Stanford, CA, USA, in 2010. He is currently working towards the Ph.D. degree at the Institute for Computational and Mathematical Engineering, Stanford University. His research interests include developing algorithms for robotic motion planning subject to general dynamic and cost constraints.
\end{IEEEbiography}
\begin{IEEEbiography}[{\includegraphics[width=\textwidth]{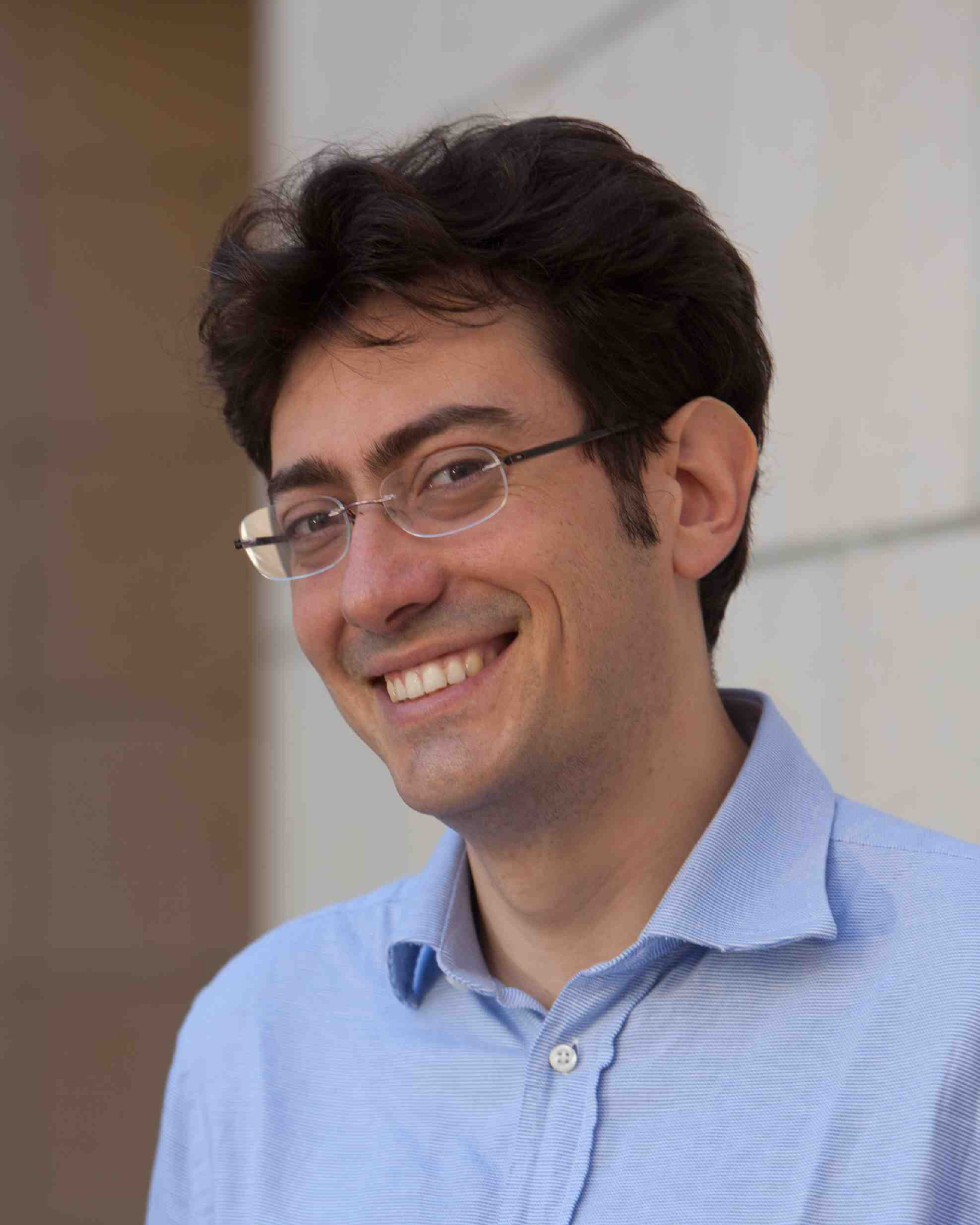}}]%
{Dr. Marco Pavone} is an Assistant Professor of Aeronautics and Astronautics at Stanford University, where he is the Director of the Autonomous Systems Laboratory. Before joining Stanford, he was a Research Technologist within the Robotics Section at the NASA Jet Propulsion Laboratory. He received a Ph.D. degree in Aeronautics and Astronautics from the Massachusetts Institute of Technology in 2010.  Dr. Pavone's areas of expertise lie in the fields of controls and robotics. His main research interests are in the development of methodologies for the analysis, design, and control of autonomous systems, with an emphasis on autonomous aerospace vehicles and large-scale robotic networks. He is a recipient of an NSF CAREER Award, a NASA Early Career Faculty Award, a Hellman Faculty Scholar Award, and was named NASA NIAC Fellow in 2011. He is currently serving as an Associate Editor for the IEEE Control Systems Magazine.
\end{IEEEbiography}

\iftoggle{arxiv}{

\clearpage

\section*{Appendix}
\section{Mathematic Preliminaries}

In this section we present the eigendecomposition of circulant matrices. The following theorem will be instrumental to the proof for Theorem \ref{contraction_th_orig}.

\begin{theorem}[Adapted from Theorem 3.1 in~\cite{RMG-01}]
\label{circulant_eig}
The eigenvalues of a circulant matrix are the Discrete Fourier Transform of the first row of the matrix. Specifically, every circulant matrix $C= \mathrm{circ}[c_1 \  c_2 \  \ldots \  c_{n}]$ has eigenvectors:
\begin{equation}
\label{eq:circulant_evectors}
\begin{split}
{\bf v_k} = \dfrac{1}{\sqrt{n}} (1, e^{-2\pi j (k-1)/n},\ldots, e^{-2\pi j (k-1)(n-1)/n})^{T}, \\ k= 1,2,\ldots, n, \end{split}
\end{equation}
and corresponding eigenvalues:
\begin{equation}
\label{eq:circulant_evalues}
\lambda_k = \sum\limits_{p=1}^{n}c_p e^{2\pi j(k-1)(p-1)/n}, 
\end{equation}
and can be expressed in the form $C = U\Lambda U^{*}$ where $U$ is a unitary matrix with the $k^{th}$ column equal to the $k^{th}$ eigenvector and $\Lambda$ is a diagonal matrix of corresponding eigenvalues. Note that the same matrix $U$ diagonalizes all circulant matrices. Thus if $C$ and $B$ are $n\times n$ circulant matrices, with eigenvalues $\{\lambda_{B,k}\}_{k=1}^{n}$ and $\{\lambda_{C,k}\}_{k=1}^{n}$ respectively, then:
\begin{enumerate}
\item $C$ and $B$ commute ($CB=BC$) and $CB$ is also a circulant matrix with eigenvalues equal to $\{\lambda_{B,k}\lambda_{C,k}\}_{k=1}^{n}$
\item $C+B$ is also circulant with eigenvalues $\{\lambda_{B,k}+\lambda_{C,k}\}_{k=1}^{n}$
\end{enumerate}
\end{theorem} 

\section{Proofs}

\subsection{Proofs for Section \ref{sec:planar}}

\begin{proof}[Proof of Lemma \ref{Lemma:manifold_defined}]
Necessity is trivial. We then consider sufficiency. First, we prove that constraints \eqref{manifold_def1} and \eqref{manifold_def2} ensure that all robots lie in a common plane, i.e., 
\begin{equation}\label{ed:des_p}
\bm{e}^{T}_{z} \, \bm{x}_i= a, 
\end{equation} 
for all $i\in \{1,\ldots,n\}$ and some $a\in \reals$. Since $\bm{e}_{z}$ is a left eigenvector of $R_{2\pi/n}$, with eigenvalue equal to 1, the rotational constraints imply
$\bm{e}_{z}^{T}(\bm{x}_{i+1}\! -\! \bm{x}_i) = \bm{e}_{z}^{T}( {\bm{x}_{i+2}\! - \!\bm{x}_{i+1}} )$,
for $i = 1,\ldots,n-2$. Combining this set of equations with the in-plane constraint, one can write
\[
\bm{e}_{z}^{T}(\bm{x}_{i+1}\! -\! \bm{x}_i) = \bm{e}_{z}^{T}( {\bm{x}_{i+2}\! - \!\bm{x}_{i+1}} ),
\]
for $i = 1,\ldots,n-1$ and where the indices are modulo $n$.  One can then readily show that 
\begin{equation} 
\bm{e}_{z}^{T} \bm{x}_{i+2} = \bm{e}_{z}^{T}\bm{x}_{2} +i\, \bm{e}_{z}^{T}(\bm{x}_2- \bm{x}_1), 
\label{manifold_defined_1}
\end{equation}
for $i = 1,\ldots,n-1$ and where the indices are modulo $n$. Hence, for $i=n-1$, one obtains
\[
\bm{e}_{z}^{T} \bm{x}_{1} = \bm{e}_{z}^{T}\bm{x}_{2} +(n-1)\, \bm{e}_{z}^{T}(\bm{x}_2- \bm{x}_1), 
\]
which implies that $\bm{e}_{z}^{T}(\bm{x}_2- \bm{x}_1) = \bm{0}$. Hence, by setting $\bm{e}_{z}^{T}\bm{x}_{2} :=a$, one immediately obtains equation \eqref{ed:des_p}. We have then proven that all robots lie in a common plane. Now, for $n$ points lying in a common plane, the rotational constraints represent the definition of a regular polygon.

For independence, we note that each successive rotational constraint involves at least one robot that does not belong to any of the previous constraints, proving that the rotational constraints are independent. The in-plane constraint is necessary to ensure that the robots do not form a spiral formation out of the desired plane, thereby reducing the dimensionality of the solution subspace. Thus, the in-plane constraint must be independent to the rotational constraints.  
 \end{proof}

\begin{proof}[Proof of Lemma \ref{Lemma:rot}]
Equations \eqref{manifold_def1} and \eqref{manifold_def2} can be written in matrix form as $V \, \bm{x} = \bm{0}$, where
\begin{align}
&V = \left[
								\begin{array}{ccc} I_3 & -(I_3 + R_{2\pi/n}) & R_{2\pi/n}  \\
								0_{3\times3} & I_3 & -(I_3+ R_{2\pi/n})  \\
								\vdots       & \vdots  &  \vdots    \\
								\bm{e}_{z}^{T}R_{2\pi/n} & 0_{1\times 3} & 0_{1\times 3}                     
								\end{array}\right. \ldots \nonumber \\
								&\qquad \qquad  \left. 
								\begin{array}{cccc}
								0_{3\times3} & \cdots & \cdots &  0_{3\times3} \\
								R_{2\pi/n} & 0_{3\times3} & \cdots & \vdots  \\
								\vdots            & \vdots   &\vdots &\vdots  \\
								\cdots & \cdots &\bm{e}_{z}^{T} & -\bm{e}_{z}^{T}(I_{3} + R_{2\pi/n} )
								\end{array}
								\right].
\label{eq:Vr_def}								
\end{align}
Simplifying the above expression, one obtains
\begin{align}
V =  & \mathcal{W}_{n}\mathrm{circ}[I_3, -I_3, 0_{3\times3}, \ldots, 0_{3\times3}] + \\
&\quad\mathcal{W}_{n}\mathrm{circ}[0_{3\times3}, -R_{2\pi/n},  R_{2\pi/n},   0_{3\times3}, \ldots, 0_{3\times3}]  \nonumber \\
 =  &\mathcal{W}_{n}(L_1 \otimes I_3) + \mathcal{W}_{n}((L_1 - L_2)  \otimes R_{2\pi/n}I_{3}) \nonumber \\
  			=& \mathcal{W}_{n}\left( L_{1}I_{n}\otimes I_{3}I_{3} + (L_1 - L_2)I_{n}\otimes R_{2\pi/n}I_{3}\right) 	\nonumber \\
				=&\mathcal{W}_{n}(L_1 \otimes I_3 + (L_1 - L_2)\otimes R_{2\pi/n})(I_n\otimes I_3), \nonumber		     
\end{align}
where the last equality follows from the properties of the Kronecker product. The claim then follows.										  												
\end{proof}

\begin{proof}[Proof of Lemma \ref{cyclic_invariant}]
Note that the rotational constraints in \eqref{manifold_def1} and the in-plane constraint \eqref{manifold_def2} define the subspace $\mathcal{M}_n$ to be a regular polygon with normal $\bm{e}_{z}$. The following analysis therefore assumes that all robots lie in a single plane and satisfy the rotational constraints. Then, 
\begin{align}
\bm{\dot{x}}_{i+1} - \bm{\dot{x}}_{i} =&  \sum\limits_{m=1}^{N} k_m(\bm{x},t)\big[R_m(\bm{x},t)(\bm{x}_{i+1+m} - \bm{x}_{i+1}) \nonumber \\
									&	+ R_m^T(\bm{x},t)(\bm{x}_{i+1-m} - \bm{x}_{i+1})\big] \nonumber \\
									&	-\sum\limits_{m=1}^{N} k_m(\bm{x},t)\big[R_m(\bm{x},t)(\bm{x}_{i+m} - \bm{x}_{i}) \nonumber \\
									&	+ R_m^T(\bm{x},t)(\bm{x}_{i-m} - \bm{x}_{i})\big].
\label{eq:prove_invariant_1}
\end{align}										
Now by the constraints described by \eqref{manifold_def1}, we can state:
\begin{align*}
\bm{x}_{i+1+m} - \bm{x}_{i+1} &= R_{2\pi/n}(\bm{x}_{i+2+m} - \bm{x}_{i+2}), \text{ and } \\
\bm{x}_{i+1-m} - \bm{x}_{i+1} &= R_{2\pi/n}(\bm{x}_{i+2-m} - \bm{x}_{i+2}).
\end{align*}
Thus, the first term in  (\ref{eq:prove_invariant_1}) can be written as:
\begin{align}
\bm{\dot{x}}_{i+1}  =& \sum\limits_{m=1}^{N} k_m(\bm{x},t)\big[R_{m}R_{2\pi/n}(\bm{x},t)(\bm{x}_{i+2+m} - \bm{x}_{i+2})  \nonumber\\
				 &	+ R_m^TR_{2\pi/n}(\bm{x},t)(\bm{x}_{i+2-m} - \bm{x}_{i+2})\big] \nonumber \\
  		 		= &  R_{2\pi/n}\sum\limits_{m=1}^{N} k_m(\bm{x},t)\big[R_{m}(\bm{x},t)(\bm{x}_{i+2+m} - \bm{x}_{i+2}) \nonumber \\
				&	+ R_m^T(\bm{x},t)(\bm{x}_{i+2-m} - \bm{x}_{i+2})\big]  \nonumber \\
				= & R_{2\pi/n}\bm{\dot{x}}_{i+2},
\end{align}					
where the order of $R_{m}$ and $R_{2\pi/n}$ can be swapped since they are rotation matrices about the same axis. Similarly, it can be shown that $\bm{\dot{x}}_{i} = R_{2\pi/n}\bm{\dot{x}}_{i+1}$. Thus,  rewriting (\ref{eq:prove_invariant_1}):
\[
\bm{\dot{x}}_{i+1} - \bm{\dot{x}}_{i} = R_{2\pi/n}(\bm{\dot{x}}_{i+2} - \bm{\dot{x}}_{i+1}) .
\]
To prove the arbitrary normal case, one need only replace $R_{m}$ and $R_{2\pi/n}$ with their similarity transformed versions with respect to $R_{\eta} \neq I_3$. 
\end{proof}

\begin{proof}[Proof of Theorem \ref{contraction_th_orig}]
We provide the proof for the general case where the desired formation plane is not necessarily the horizontal plane. The proof leverages the following Lemma.
\begin{lemma}\label{W_n_Lemma}
Given a symmetric matrix $\mathcal{X}$ such that $\mathcal{X} \prec 0$, then $\mathcal{W}_{n}\mathcal{X}\mathcal{W}^T_{n} \prec 0$.
\end{lemma}

\begin{proof}
If $\mathcal{X} \prec 0$ then $\lambda_{max}(\mathcal{X}) < 0$. Noting that $\mathcal{W}_{n}$ is sub-unitary, by the Cauchy Interlacing Theorem, $\lambda_{max}(\mathcal{W}_n\mathcal{X}\mathcal{W}_n^T) \leq \lambda_{max}(\mathcal{X}) < 0$.
\end{proof}
Having proved Lemma \ref{W_n_Lemma}, we are now in a position to prove Theorem \ref{contraction_th_orig}. Specifically, the strategy is to only consider the uniform negative-definiteness of the inner term (not involving $\mathcal{W}_{n}$) in (\ref{eq:convergence_cond}). 

For the scenario where the desired formation plane is not the horizontal plane, by the discussion in Remark \ref{Remark:Reta}, we must use the similarity transformed versions of the projection matrix $V$ and the symmetric cyclic controller. From \eqref{cyclic_modified} and \eqref{Vr_modified}, note that $\mathcal{P}_{n}\mathcal{R}_{\eta}\mathcal{L}_{m\eta}\allowbreak\mathcal{R}^T_{\eta}\mathcal{P}^T_{n} = \mathcal{P}_{n}\mathcal{L}_{m}\mathcal{P}^T_{n}$ since $\mathcal{L}_{m\eta} = \mathcal{R}^{T}_{\eta}\mathcal{L}_{m}\mathcal{R}_{\eta}$. Thus, the analysis does not depend on the arbitrary choice of the orientation of the desired plane normal in the global coordinate system.

Now, given the matrices $\mathcal{P}_n$ and $\mathcal{L}_m$ are $\in \CR$, they have the same set of eigenvectors. Then, for any eigenvector $\bm{v}_i$ in this set, we have the following relation for the corresponding eigenvalues:

\[
\lambda_i\left(\mathcal{P}_{n}\mathcal{L}_{m}\mathcal{P}_{n}^T\right) = \lambda_i(\mathcal{P}_{n})\lambda_i(\mathcal{P}_{n}^T)\lambda_i(\mathcal{L}_{m}).
 \]

To derive the eigenvalues of $\mathcal{P}_{n}$, we note that $\mathcal{P}_{n} = L_1 \otimes I_3 + (L_1 - L_2)\otimes R_{2\pi/n}$ which is the sum of two $\CR$ matrices and thus the eigenvalues of $\mathcal{P}_{n}$ must be the sum of the eigenvalues of $L_1 \otimes I_3$ and $(L_1 - L_2)\otimes R_{2\pi/n}$. From Theorem~\ref{circulant_eig}, we have the general result:
\begin{equation}
\label{eq:m_circulant_eig}
\lambda_i(L_{m}) = 1 - e^{\frac{2m\pi}{n}(i-1)j},\  i=1,\ldots,n.
\end{equation}
Additionally, from Theorem~\ref{circulant_eig}, the eigenvalues of $L_1 - L_2$ are the sum of eigenvalues of $L_1$ and $-L_2$. The eigenvalues of $R_{2\pi/n}$ are $\{1,e^{\pm\frac{2\pi}{n}j}\} = e^{\frac{2k\pi}{n}j}$ where $k\in\{-1,0,1\}$. Thus, the eigenvalues of $\mathcal{P}_{n}$ and $\mathcal{P}_{n}^T$ are the set $i\in\{1,\ldots,n\},k\in\{-1,0,1\}$:
\begin{align}
\lambda(\mathcal{P}_{n},\mathcal{P}_{n}^T) =& \left(1- e^{\pm \frac{2\pi}{n}(i-1)j}\right) \nonumber \\
	&+ \left( e^{\pm \frac{4\pi}{n}(i-1)j} - e^{\pm \frac{2\pi}{n}(i-1)j}\right)e^{\pm \frac{2k\pi}{n}j} ,
\label{eq:eig_L}
\end{align}
where the positive and negative signs are for $\mathcal{P}_{n}$ and $\mathcal{P}_{n}^T$ respectively. In a similar fashion, the eigenvalues for $\mathcal{L}_{m}$ are:
\begin{align}
\label{eq:eig_Lalpha}
\lambda(\mathcal{L}_{m}) =& \left(1- e^{\frac{2m\pi}{n}(i-1)j} \right)e^{k\alpha_mj} \nonumber \\
& + \left(1- e^{\frac{-2m\pi}{n}(i-1)j} \right)e^{-k\alpha_mj} \nonumber \\
=& 2\left(\cos(k\alpha_m) - \cos\left(k\alpha_m + \frac{2\pi m(i-1)}{n}\right)\right).
\end{align}
Multiplying (\ref{eq:eig_L}) and (\ref{eq:eig_Lalpha}) gives the expression for $\lambda\left(\mathcal{P}_{n}\mathcal{L}_{m}\mathcal{P}_{n}^T\right) = \lambda^{(m)}_{ik}(\bm{x},t)$. To obtain the summation form in (\ref{eq:eig_Theorem}), note that $\mathcal{P}_{n}\mathcal{L}_{m}\mathcal{P}_{n}^T \in \CR$ for all $m$. Thus $\lambda\left(\mathcal{P}_{n}\sum\limits_{m=1}^{N} k_m\mathcal{L}_{m}\mathcal{P}^T_{n}\right) = \sum\limits_{m=1}^{N} k_m\left[ \lambda\left( \mathcal{P}_{n}\mathcal{L}_{m}\mathcal{P}^T_{n} \right)\right]$. Having obtained the summation form, (\ref{eq:eig_Theorem}) follows directly from trying to show (\ref{eq:convergence_cond}) given Lemma~\ref{W_n_Lemma}.

\end{proof}

\subsection{Proofs for Section \ref{sec:complex}}

\begin{proof}[Proof of Corollary \ref{corr:manifold_defined_red}]
Necessity is straightforward. We then consider sufficiency. Let us assume without loss of generality that $\bm{n}_k = \bm{e}_z$; the general case follows through a rotation under $R_{\eta_k}$. To prove the claim, we need to show that the rotational constraints are sufficient to ensure that all robots lie in the desired plane for $j = 1,\ldots,|\mathcal{V}_k|$. The claim then follows directly.  Now, the rotational constraints imply
\begin{equation}
	\bm{e}_{z}^{T}(\bm{x}_{i+1}\! -\! \bm{x}_i) = \bm{e}_{z}^{T}( {\bm{x}_{i+2}\! - \!\bm{x}_{i+1}} ),
\label{manifold_defined_2}
\end{equation}
for $i = 1,\ldots,|\mathcal{V}_k|-2$. For two robots with indices $j$ and $j+1$ that lie in the desired plane, we know that $\bm{e}_{z}^{T}\bm{x}_j = \bm{e}_{z}^{T}\bm{x}_{j+1} = a \in \reals$. If $j \in \{1,\ldots,|\mathcal{V}_k|-1\}$, we can recursively use \eqref{manifold_defined_2} to show
\begin{equation}
	\bm{e}_{z}^{T}\bm{x}_i= \bm{e}_{z}^{T}\bm{x}_j,
\label{manifold_defined_3}
\end{equation}
for $i = 1, \ldots, |\mathcal{V}_k|$. For the case where $j=|\mathcal{V}_k|$, we have:
\[
	\bm{e}_z^T\bm{x}_n = \bm{e}_z^T\bm{x}_1. 
\]
However, we note that \eqref{manifold_defined_1} still holds for $i = 1,\ldots,|\mathcal{V}_k|-2$ using the rotational constraints alone. Then, setting $i = |\mathcal{V}_k|-2$, we obtain
\[
	(\bm{e}_z^T\bm{x}_n - \bm{e}_z^T\bm{x}_1) + (|\mathcal{V}_k|-1)\bm{e}_z^T(\bm{x}_2 - \bm{x}_1) = (|\mathcal{V}_k|-1)\bm{e}_z^T(\bm{x}_2 - \bm{x}_1),
\]
which gives $\bm{e}_z^T\bm{x}_2 = \bm{e}_z^T\bm{x}_1$. We can now again use \eqref{manifold_defined_2} recursively to show \eqref{manifold_defined_3}, completing the proof.
\end{proof}

\subsection{Proofs for Section \ref{sec:additional}} 

\subsubsection{Control Over Formation Size}

In order to prove Theorem \ref{thm:size_final}, we first require several intermediate results detailed in the following lemmas. We first prove that $\bm{x}$ converges to the subspace $\mathcal{M}_n$.

\begin{lemma}[Convergence to $\mathcal{M}_n$ with Size Control]
\label{contraction_size}
Assume
\begin{equation}
\begin{split}
&\inf_{\alpha_m\in[m\pi/n-\alpha_{s_0}, m\pi/n+\alpha_{s_0}]}\bigg(\min\limits_{\substack{ 1\leq i\leq n \\ k\in\{-1,0,1\}}} \  \sum\limits_{m=1}^{N} k_m \, \lambda^{(m)}_{ik}\ \bigg)   > \! 0, 
\end{split}
\label{eig_Theorem_size}
\end{equation}
where $\lambda^{(m)}_{ik}$ has the form given in Theorem \ref{contraction_th_orig}.
Then the system $\dot{\bm{x}} = \bm{u}$ under the controller \eqref{symmetric_control_simplified_size} globally converges to $\mathcal{M}_n$.
\end{lemma}
\begin{proof}
To show convergence to $\mathcal{M}_n$, consider the following auxiliary system for the transformed closed-loop dynamics $\dot{\bm{z}}_p = \bar{V}\dot{\bm{x}_p} = \bar{V}\bm{u}$ where $\bm{u}$ is as given in \eqref{symmetric_control_simplified_size}:
\[
\dot{\bm{y}} =  -\bar{V}\mathcal{L}(\bm{x}_p,\tau)\,(\bar{V}^T\bm{y} + \bar{U}^T\bar{U}\bm{x}_p) , 
\]
Note that $\bm{y}=\bm{0}$ is a solution of this auxiliary system since $\bar{V}\mathcal{L}(\bm{x}_p,\tau)\,( \bar{U}^T\bar{U}\bm{x}_p)=\bm{0},\ \forall \bm{x}_p$. Then to show that the auxiliary system is contracting in $\bm{y}$, we require:
\[
\bar{V} \bigg( \mathcal{L}(\bm{x}_p,\tau)\bigg)\bar{V}^T \succ 0, \quad \text{ uniformly},\ \forall \bm{x}_p, \tau. 
\]
By Remark \ref{V_Vbar_T}, the above stability requirement can be re-formulated in terms of $V$ directly, i.e.,  

\begin{equation} \label{convergence_cond_size}
\mathcal{W}_{n} \mathcal{P}_{n} \bigg( \mathcal{L}(\bm{x}_p,\tau)\bigg)  \mathcal{P}^{T}_{n} \mathcal{W}^{T}_{n}  \succ 0,\quad \text{ uniformly} ,\ \forall \bm{x}_p, \tau.
\end{equation}
The above condition is equivalent to (\ref{eq:convergence_cond}) with $\bm{g}=\bm{0}$, $\mathcal{L} = \mathcal{L}(\bm{x}_p,\tau)$ and the additional stipulation that \eqref{convergence_cond_size} holds for all possible $\bm{x}_p$ and time lags $\tau$. Temporarily neglecting this additional condition and applying Theorem \ref{contraction_th_orig} with $\bm{g}=\bm{0}$ gives the expression within the infimum in \eqref{eig_Theorem_size}. Then to account for all possible  $\bm{x}_p$ and time lags $\tau$, we note that the dependance of $\mathcal{L}$ on $\bm{x}_p$ and $\tau$ stems from the rotation angle $\alpha_m = m\pi/n + \alpha_{s}(\bm{x}_p,\tau)$, which changes every $\tau$ seconds. As $|f_{s}(\bar{p})|\leq1$,  it follows that $\alpha_m \in [m\pi/n-\alpha_{s_0}, m\pi/n+\alpha_{s_0}]$, for all time (despite the discrete switches every $\tau$ seconds). Choosing the worst case value for $\alpha_m$ within this interval then completes the expression given in \eqref{eig_Theorem_size}. 
  \end{proof}

\begin{remark}
Note that the continuous/discrete hybrid dynamics induced by the time lag parameter $\tau$ do not affect the convergence analysis for $\mathcal{M}_n$. By taking the infimum in \eqref{eig_Theorem_size}, we have that for any modification of the standard rotation angle $m\pi/n$ by the angle $\alpha_s$, condition \eqref{convergence_cond_size} is still satisfied.
\end{remark}

We now study the behavior of the system within the invariant subspace $\mathcal{M}_n$, with respect to the desired subset $\mathcal{M}_{n\rho}$. Here, the time lag $\tau$ will play a crucial role in governing the convergence of the inter-robot separation to the desired distance. 

\begin{lemma}[Formation Size Discrete Dynamics]\label{Lemma:Mn_to_Mnrho}
Assume $\bm{x}$ lies in the subspace $\mathcal{M}_n$. Let $\beta =\sqrt{2(1-\cos(2\pi/n))}$, $\Gamma = \sum\limits_{m=1}^{N}k_m\gamma_m$ where $\gamma_m$ are fixed constants derivable from geometry, specifically,
\[
	\gamma_1 = 1,\text{ and } \gamma_m \in \begin{cases} \left(1, \dfrac{1}{\sin(\frac{\pi}{n})}\right] &\mbox{ if }  $n$ \text{ is even,} \\
					     \left(1, \dfrac{1}{2\sin(\frac{\pi}{2n})}\right] &\mbox{ if }  $n$ \text{ is odd,}
			\end{cases}
\]
for $m = 2,\ldots,N$, and let $C = 2\beta\Gamma$. Then, under the action of the modified cyclic controller given in \eqref{symmetric_control_simplified_size}, each inter-robot distance error $p_i$ admits the following recursive solution:
\begin{equation}
\label{size_discrete}
	p_{i_{k+1}}= (p_{i_k} - 1)e^{C\sin(\alpha_s(p_{i_{k-1}}))\tau}+1,
\end{equation}
where $p_{i_k} = p_i(k\tau)$.
\end{lemma}
\begin{proof}
First note that by Lemma \ref{Lemma:size_invariant}, $\mathcal{M}_n$ is flow-invariant with respect to the control law given in \eqref{symmetric_control_simplified_size}. Thus, we will exploit the symmetry properties attributed to regular polygons. In particular, recall the definition $p_{i}:= 1 - \|\bm{x}_{i+1} - \bm{x}_i\|/\rho$. We note that for $\bm{x}\in\mathcal{M}_n$, $p_1 = p_2 = \ldots = p_n = \bar{p}$. Consider the time-derivative of $p_i$ within the time-range $[k\tau, (k+1)\tau)$:

\[
	\dot{p}_i = -\dfrac{(\bm{x}_{i+1}-\bm{x}_i)^T(\dot{\bm{x}}_{i+1}-\dot{\bm{x}}_i)}{\rho\|\bm{x}_{i+1}-\bm{x}_i\|}.
\]
To evaluate $(\dot{\bm{x}}_{i+1}-\dot{\bm{x}}_i)$, consider Figure \ref{fig:cyclic_size}. 
\begin{figure}[htbp]
\vspace{-2mm}
\begin{center}
	\includegraphics[scale=0.7]{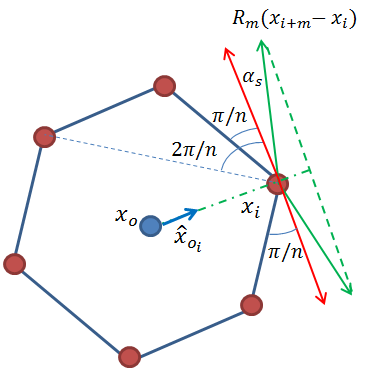}
\end{center}	
	\caption{Cyclic control with $\alpha_s$}
	\label{fig:cyclic_size}
\vspace{-2mm}
\end{figure}

From Remark \ref{Remark:cyclic_fixed}, we see the geometric realization of the forward and rear components of the cyclic controller cancelling for $\bm{x} \in \mathcal{M}_n$, when $\alpha_m = m\pi/n$. We can also see that the vector $R_m(\bm{x}_{i+m}-\bm{x}_i)$ where $R_m$ is a rotation matrix with angle $m\pi/n + \alpha_s$, is simply a rotation of $R_{m\pi/n}(\bm{x}_{i+m}-\bm{x}_i)$ about the plane normal by the angle $\alpha_s$. Thus, we can write:

\[
	\begin{split}
		\dot{\bm{x}}_i &= 2\sin(\alpha_s)\left(\sum\limits_{m=1}^{N}k_m\|\bm{x}_{i+m}-\bm{x}_i\|\right) \hat{\bm{x}}_{o_i} \\
		\dot{\bm{x}}_{i+1} &= 2\sin(\alpha_s)\left(\sum\limits_{m=1}^{N}k_m\|\bm{x}_{i+1+m}-\bm{x}_{i+1}\|\right) \hat{\bm{x}}_{o_{i+1}} \\
					    &= 2\sin(\alpha_s)\left(\sum\limits_{m=1}^{N}k_m\|\bm{x}_{i+m}-\bm{x}_i\|\right) \hat{\bm{x}}_{o_{i+1}},
	\end{split}
\]
where $\hat{\bm{x}}_{o_i}$ is the unit vector pointing outward from the instantaneous geometric center $\bm{x}_0$ towards robot $i$. Exploiting symmetry, we can re-write $\|\bm{x}_{i+m}-\bm{x}_i\|$ as $\gamma_m\|\bm{x}_{i+1}-\bm{x}_{i}\|$, where $\gamma_1 = 1$ and $\gamma_m$ for $m = 2,\ldots,N$ may be derived from geometry. Using these simplifications, we obtain 
\[ 
	\begin{split}
 	(\dot{\bm{x}}_{i+1}-\dot{\bm{x}}_i) &= 2\sin(\alpha_s)\left(\underbrace{\sum\limits_{m=1}^{N}k_m\gamma_m}_{:=\Gamma>0}\right)\|\bm{x}_{i+1}-\bm{x}_i\| \\ &\times \underbrace{\sqrt{2(1-\cos(2\pi/n))}}_{:=\beta>0} \dfrac{(\bm{x}_{i+1}-\bm{x}_i)}{\|\bm{x}_{i+1}-\bm{x}_i\|}. 
	\end{split}
\]
We now recall that $\alpha_s = \alpha_s(\bar{p}_{k-1})$, and is constant over the time range $[k\tau, (k+1)\tau)$. Thus, we can write $\dot{p}_i$ as: 
\begin{equation}
\begin{split}
\dot{p}_i &= -2\beta\Gamma\sin\big(\alpha_s(\bar{p}_{k-1})\big)\dfrac{\|\bm{x}_{i+1}-\bm{x}_{i}\|}{\rho}\\
	       &= C\sin\big(\alpha_s(p_{i_{k-1}})\big)(p_{i}-1).
\label{size_cont}
\end{split}
\end{equation}
Solving the above equation exactly gives the recursive formula given in \eqref{size_discrete}.
 \end{proof}

We are now ready to derive sufficient conditions for the time lag parameter $\tau$ to ensure that the recursive solution given in \eqref{size_discrete} converges to 0. 

\begin{lemma}[Convergence to $\mathcal{M}_{n\rho}$ within $\mathcal{M}_n$] \label{Mn_Mnrho}
Assume $\bm{x}$ lies in the subspace $\mathcal{M}_n$. Let $T$ denote a positive constant such that $(1/2)T|p_i| \leq |\sin\left(\alpha_{s_0}f_s(p_i)\right)| \leq T |p_i|$ for $|p_i|<1$ (if $f_s$ is the saturation function, $T = \alpha_{s_0}$). Suppose the time lag parameter $\tau$ satisfies the following bound:
\begin{equation}
	\tau < \min \left\{ \dfrac{1}{C}, \dfrac{1}{8CT}\right\}.
\label{tau_constraint}
\end{equation}
Then the neighbor inter-robot distance governed by the discrete dynamics given in \eqref{size_discrete} converges to the desired distance $\rho$. That is, $\bm{x}$ converges to the subspace $\mathcal{M}_{n\rho}$.
\end{lemma}
\begin{proof}
The layout of the proof is as follows: we first show that there exists some time index $k'$ such that the inter-robot distance error $p_{i_k}$ has the same sign in two consecutive time steps. Next, we prove that if $p_{i_k}$ has the same sign \emph{for all} $k>k'$, then $p_{i_k}$ converges to zero. Finally, we provide sufficient conditions to ensure that either (a) $p_{i_k}$ possesses the same sign for all $k>k'$, or (b) if there exists a time index $k''$ where $p_{i_k}$ changes sign, then it never crosses 0 again. We begin with the following claim:
\begin{claim}
Given the inter-robot distance error discrete dynamics given in \eqref{size_discrete}, there exists a time index $k'$ such that $p_{i_{k'-1}}p_{i_{k'}}>0$. That is, the distance error has the same sign in two consecutive time steps. 		
\end{claim}
\begin{proof}
Assume by way of contradiction that $p_{i_{k-1}}p_{i_{k}}<0$,  $p_{i_{k-1}}<0$ and $1>p_{i_k}>0$. Then, $\sin(\alpha_s(p_{i_{k-1}}))<0$. From \eqref{size_discrete}, we have $1>p_{i_{k+1}}>p_{i_{k}}>0$. Alternatively, assume that  $p_{i_{k-1}}p_{i_{k}}<0$,  $p_{i_{k-1}}>0$ and $p_{i_k}<0$. Then, $\sin(\alpha_s(p_{i_{k-1}}))>0$. Using \eqref{size_discrete}, we have $p_{i_{k+1}}<p_{i_{k}}<0$. Setting $k' = k+1$ gives the desired contradiction.
 \end{proof}
Given there exists a time index $k'$, such that $p_{i_{k'-1}}p_{i_{k'}}>0$, consider the following result:
\begin{claim} \label{z_mono}
Let $k'$ be a time index such that $p_{i_{k'-1}}p_{i_{k'}}>0$. If
\begin{enumerate}
	\item $p_{i_{k'}}>0$ and $p_{i_{k}}>0$ for all $k>k'$, or
	\item $p_{i_{k'}}<0$ and $p_{i_{k}}<0$ for all $k>k'$,
\end{enumerate}
then $p_{i_k}$ converges to 0. 
\end{claim}
\begin{proof}
Consider the first case. Without loss of generality, assume $p_{i_{k'-1}}>p_{i_{k'}}$. If this is not true, then take $k''$ to be $k'+1$. Then given \eqref{size_discrete}, we have $p_{i_{k''}} =  p_{i_{k'+1}} < p_{i_{k'}} = p_{i_{k''-1}}$, and we proceed with our analysis for $k>k''$. 

Now if $p_{i_k}>0$ for all $k>k'$, and $p_{i_k}p_{i_{k-1}}>0$ for all $k>k'$, then given \eqref{size_discrete}, we deduce that $p_{i_k}$ is a strictly decreasing monotonic sequence lower-bounded by 0, and therefore has a limit $l\geq 0$. We will show that $l=0$.  

Consider the continuous time differential equation in \eqref{size_cont}. Given the recursive solution in \eqref{size_discrete}, we know that $p_i(t) < p_{i_{k'-1}}$ (and consequently $\alpha_s(p_i(t)) < \alpha_s(p_{i_{k'-1}})$) for all $t> (k'-1)\tau$. Furthermore, since $p_i(t)-1<0$, we obtain the following differential inequality for $t\in [k'\tau, (k'+1)\tau]$:
\begin{equation}
	\begin{split}
	\dot{p_i}(t) &= C\sin\left(\alpha_s(p_{i_{k'-1})}\right)(p_i(t) - 1) \\
			 &< C\sin\left(\alpha_s(p_i(t)) \right)(p_i(t)-1).
	\end{split}
\label{diff_ineq}
\end{equation}
We now use the lower bound: $\sin\left(\alpha_s(p_i(t))\right)> (1/2)Tp_i(t)$. It follows then:
\[
	\dot{p_i}(t)< \dfrac{CT}{2}p_i(t)(p_i(t)-1).
\]
By the comparison theorem \cite{HKK:02}, we obtain for $t \in [k'\tau, (k'+1)\tau]$:
\[
	p_i(t) < \dfrac{1}{1+\left(\dfrac{1}{p_{i_{k'}}}- 1\right)e^{\frac{CT}{2}(t-k'\tau)}}:= \xi(t).
\]
In particular, we note that
\[
	p_{i_{k'+1}} < \xi((k'+1)\tau).
\]
Proceeding inductively, it follows that
\[
	p_{i_{k}} < \xi(k\tau), \forall k> k'.
\]
That is, the recursive solution is bounded above by $\xi(t)$, a Class-$\mathcal{L}$ function \cite{HKK:02}, for all $k>k'$. Consequently, the limit $l$ for $p_{i_k}$ must equal 0.

Consider now the second case. Without loss of generality, assume $p_{i_{k'}}>p_{i_{k'-1}}$. If this is not true, then take $k''$ to be $k'+1$. Then given \eqref{size_discrete}, we have $p_{i_{k''}} =  p_{i_{k'+1}} > p_{i_{k'}} = p_{i_{k''-1}}$, and we proceed with our analysis for $k>k''$. 

Now if $p_{i_k}<0$ for all $k>k'$ and $p_{i_k}p_{i_{k-1}}>0$ for all $k>k'$, then given \eqref{size_discrete}, we deduce that $p_{i_k}$ is a strictly increasing monotonic sequence upper-bounded by 0, and therefore has a limit $-l \leq 0$. We will show that $l=0$.  

We first prove that $l <1$. That is, there exists a time index $k^{*}$ such that for all $k>k^{*}$, $|p_{i_k}|<1$. By way of contradiction, assume that $l \geq1$ and suppose $p_{i_k} =-l - \epsilon$ where $\epsilon>0$ is an infinitesimally small positive number. Given \eqref{size_discrete}, we know that $p_{i_{k-1}} < p_{i_k} < -1$. Then since $C\tau > 0$, we know that $C\sin(\alpha_s(p_{i_{k-1}}))\tau < C\sin(\alpha_{s_0}f_s(-1))\tau<0$, and thus $e^{C\sin(\alpha_s(p_{i_{k-1}}))\tau} < 1$. Since the sequence $p_{i_k}$ is increasing monotonically, for the limit to be $-l$, $p_{i_k}< -l$ for all $k$. To obtain the desired contradiction, we will show that there exists a small enough $\epsilon$ such that $p_{i_{k+1}} > -l$. Using \eqref{size_discrete}, we seek an $\epsilon>0$ such that:
\[
	(-l - \epsilon - 1)e^{C\sin(\alpha_s(p_{i_{k-1}}))\tau} > -l -1. 
\]
Re-arranging the above expression, we require:
\[
	e^{C\sin(\alpha_s(p_{i_{k-1}}))\tau} < \dfrac{l+1}{l + 1 + \epsilon},
\]
Since $e^{C\sin(\alpha_s(p_{i_{k-1}}))\tau} < e^{C\sin(\alpha_{s_0}f_s(-1))\tau}<1$, there exists an $\epsilon>0$ small enough such that the above condition will be satisfied. We now proceed with our analysis for $k>k^{*}$. For $t>k^{*}\tau$, we have that $|p_i(t)| < |p_{i_{k^{*}}}|$. Then, using \eqref{size_cont} and that $|p_i(t)|<1$, we obtain the following differential inequality for $t\in [(k^{*}+1)\tau, (k^{*}+2)\tau]$:
\[
	\begin{split}
	\dot{p_i}(t) &= -C\left|\sin\left(\alpha_s(p_{i_{k^{*}})}\right)\right|(p_i(t) - 1) \\
			& > -C\left|\sin\left(\alpha_s(p_i(t))\right)\right|(p_i(t)-1).
	\end{split}
\]
Using the lower bound: $\left|\sin\left(\alpha_s(p_i(t))\right)\right|>(1/2)T|p_i(t)|$ for $t>k^{*}\tau$, it follows:
\[
		\dot{p_i}(t) > -\dfrac{CT}{2}|p_{i}(t)|(p_i(t)-1) .
\]
Once again, applying the Comparison Theorem \cite{HKK:02}, we obtain for $t\in [(k^{*}+1)\tau, (k^{*}+2)\tau]$:
\[
	|p_{i}(t)| < \left| \dfrac{1}{1 + \left(\frac{1}{p_{i_{k^{*}+1}}}-1\right)e^{\frac{CT}{2}(t-(k^{*}+1)\tau)}} \right|:= \eta(t).
\]
We note that:
\[
	|p_{i_{k^{*}+2}}| < \eta((k^{*}+2)\tau). 
\]
Proceeding inductively, it follows that
\[
	|p_{i_k}| < \eta(k\tau),\ \forall k>k^{*}+1.
\]
Thus, $|p_{i_k}|$ is bounded above by $\eta(t)$, a Class-$\mathcal{L}$ function \cite{HKK:02}, for all $k> k^{*}+1$. Thus, the limit $-l$ for $p_{i_k}$ must equal 0.
 \end{proof}

Now suppose $p_{i_{k'-1}}p_{i_{k'}}>0$ and $p_{i_{k'}}>0$. We now derive sufficient conditions to ensure that $p_{i_k}>0$ for all $k>k'$. Let $k = k'$. Since $\tau<1/C$, and $0<\sin(\alpha_s(p_{i_{k-1}}))<1$, we have $0<C\sin(\alpha_s(p_{i_{k-1}}))\tau<1$. Using the inequality $e^{\zeta}\leq 1 + 2\zeta$ for $\zeta \in [0,1]$, we may deduce:
\[
	e^{C\sin(\alpha_s(p_{i_{k-1}}))\tau} \leq 1 + 2C\sin(\alpha_s(p_{i_{k-1}}))\tau.
\]
Additionally, since $|\sin\left(\alpha_{s_0}f_s(p_{i_{k-1}})\right)| \leq T |p_{i_{k-1}}|$, we have: 
\[
	1 + 2C\sin(\alpha_s(p_{i_{k-1}}))\tau \leq 1 + \underbrace{2CT\tau}_{:=A>0} p_{i_{k-1}}.
\]
Since, $p_{i_k}-1< 0$, it follows
\[
	p_{i_{k+1}} \geq (p_{i_k} - 1)\left(1 + A p_{i_{k-1}}\right) + 1 = p_{i_k} + Ap_{i_{k-1}}p_{i_k} - Ap_{i_{k-1}}.
\]
Equivalently, 
\[
	\dfrac{p_{i_{k+1}}}{p_{i_k}} \geq 1 + Ap_{i_{k-1}} - A\dfrac{p_{i_{k-1}}}{p_{i_k}}.
\]
Note that $p_{i_k}/p_{i_{k-1}}>1>1/2$ (i.e. $p_{i_{k-1}}/p_{i_k}<2$) and $Ap_{i_{k-1}}>0$. Then, if $Ap_{i_{k-1}}/p_{i_k}<1/2$ (i.e. $A<1/4$), then,
\[
	0<(1/2)p_{i_k}<p_{i_{k+1}}<p_{i_k}. 
\]
By induction, we conclude that $p_{i_k}>0$ for all $k>k'$. By Claim \ref{z_mono}, $p_{i_k}$ converges to 0. The requirement $A<1/4$ is guaranteed by \eqref{tau_constraint}.

Finally, we address the scenario where  $p_{i_{k'-1}}p_{i_{k'}}>0$ and $p_{i_{k'}}<0$. There are two sub-cases:
\begin{enumerate}
	\item For all $k>k'$, $p_{i_{k}}<0$. Then by Claim \ref{z_mono} above, $p_{i_k}$ converges to 0. 
	\item There exists a time index $k''$ such that $p_{i_{k''-1}}p_{i_{k''}}<0$ (i.e. $p_{i_{k''}}>0$).
\end{enumerate}
Analyzing the second sub-case, using \eqref{size_discrete}, we have that $p_{i_{k''+1}}>p_{i_{k''}}>0$. Then given $A<1/4$, $p_{i_k}>0$ for all $k>k''$. Thus, $p_{i_k}$ does not change sign again, completing the proof.

\end{proof}
%
We are now ready to prove Theorem \ref{thm:size_final} which links convergence analysis both on and off the subspace $\mathcal{M}_n$.

\begin{proof}[Proof of Theorem \ref{thm:size_final}]
By Lemmas \ref{Lemma:size_invariant}, \ref{Lemma:Mn_to_Mnrho} and \ref{Mn_Mnrho}, we know that if $\bm{x}(0)$ lies in the flow-invariant subspace $\mathcal{M}_n$, then $\bm{x}$ converges to the manifold $\mathcal{M}_{n\rho}$.

We now address the scenario where $\bm{x}(0) \in \reals^{3n}\setminus \mathcal{M}_n$. The state $\bm{x}$ may be decomposed into projections both on and off $\mathcal{M}_n$ as follows:
\[
	\bm{x} = \underbrace{\bar{U}^{T}\bar{U}\bm{x}}_{:=\bm{x}_{M} \in \mathcal{M}_n} + \underbrace{\bar{V}^{T}\bar{V}\bm{x}}_{:=\bm{x}_{M^\perp} \in \mathcal{M}_n^{\perp}}.
\] 
By Lemma \ref{contraction_size}, we have that $\|\bm{x}_{M^\perp}\|$ is decreasing exponentially to 0. Thus $\bm{x} \rightarrow \bm{x}_M$. All that remains is to show that $\bm{x}_M \rightarrow \mathcal{M}_{n\rho}$. 
To prove this result, let us examine the dynamics of $\bm{x}_M$:
\[
	\begin{split}
	\dot{\bm{x}}_M &= \bar{U}^T\bar{U} \bm{u} = -\bar{U}^T\bar{U}\mathcal{L}(\bm{x},\tau)\bm{x}\\
				&= -\mathcal{L}(\bm{x},\tau) \bar{U}^T\bar{U}\bm{x} \\
				&= -\mathcal{L}(\bm{x},\tau)\bm{x}_{M} \\
				&= -\mathcal{L}(\bm{x}_M,\tau)\bm{x}_{M} - \underbrace{\left( \mathcal{L}(\bm{x},\tau) - \mathcal{L}(\bm{x}_M,\tau)\right) \bm{x}_{M}}_{:=\bm{u}_p} ,
	\end{split}
\]
where the order of $\bar{U}^T\bar{U}$ and $\mathcal{L}$ can be swapped since the order of cyclic control and projection onto the subspace $\mathcal{M}_n$ is irrelevant. Note that from the analysis conducted in Lemma \ref{Lemma:Mn_to_Mnrho}, the center control plays no part in the convergence analysis to $\mathcal{M}_{n\rho}$. Thus we see that the dynamics for $\bm{x}_M$ are perturbed by the term $\bm{u}_p$ due to formation error. However since $\|\bm{x}-\bm{x}_M\|$ is decreasing exponentially to 0, we have that $\bm{u}_p$ is bounded and also decreasing exponentially to $\bm{0}$. Thus, there exists some time $t^\star$ following which the conditions given in Lemmas \ref{Lemma:Mn_to_Mnrho} and \ref{Mn_Mnrho} are \emph{sufficient} to ensure that $\bm{x}_M \rightarrow \mathcal{M}_{n\rho}$. 

 \end{proof}
\begin{remark}[Time Lag Parameter Bound]
Note that the bound for the time lag parameter used in the proof above (derived in the proofs for Lemmas \ref{Lemma:Mn_to_Mnrho} and \ref{Mn_Mnrho}) is less restrictive than the bound given in the statement of Theorem \ref{thm:size_final} in the main body of the text. Thus, satisfaction of the constraint in \eqref{tau_constraint_final} is sufficient for the convergence results to hold. We present the more restrictive condition in the main body of the text to simplify the exposition.
\end{remark}

Finally, we present the decoupled framework for controlling the formation size of a Johnson polyhedron by decoupling $\mathcal{Q}_1$ from the set $\mathcal{Q}$. Suppose the tuple $(\bm{n}_1, R_{\eta_1})$ describes the plane of the first polygon $\mathcal{Q}_1$ and let $(\bm{n}_2, R_{\eta_2})$ define the plane of an adjacent polygon $\mathcal{Q}_2$ with robots $j$ and $j+1$ forming the shared edge. Define the variables $\nu_f, \nu_b$ and $k_r$ as follows:
\begin{align}
	k_r &=  (\bm{x}_{j+1}-\bm{x}_j)\cdot R_{\eta_2}^T\begin{bmatrix} 0\\ 0 \\ 1\end{bmatrix} \\
	\nu_f &= \begin{cases}  1 & \mbox {if } k_r > 0, \\
				            0 & \mbox {else}. 
			\end{cases} \\
	\nu_b &= \begin{cases}  1 & \mbox {if } k_r < 0, \\
				            0 & \mbox {else}. 
			\end{cases} 
\end{align}
Then, the control law for robot $i$ in $\mathcal{Q}_1$ is given by:
\begin{equation}
	\begin{split}
	\bm{u}_i^{(1)} &\!=\! \sum\limits_{m=1}^{N_1} k^{(1)}_m \big[R_{m_s}^{(1)}(\bm{x}^{(1)}_{i+m} - \bm{x}^{(1)}_{i})  + R_{m_s}^{(1)^T}(\bm{x}^{(1)}_{i-m} - \bm{x}^{(1)}_{i})\big] \\
				&+ |k_r|\big[\nu_f R_{r_s}^{(1)} (\bm{x}^{(1)}_{i+1} - \bm{x}^{(1)}_{i}) + \nu_b R_{r_s}^{(1)^T} (\bm{x}^{(1)}_{i-1} - \bm{x}^{(1)}_{i})\big],
	\end{split}
\label{cyclic_rotational}
\end{equation}
where $R_{m_s} = R_{\eta_1}^TR_{m\pi/|\mathcal{V}_1| + \alpha_s(\bm{x}^{(1)})}R_{\eta_1}$, and $R_{r_s} = R_{\eta_1}^TR_{\pi/|\mathcal{V}_1|}R_{\eta_1}$. Thus, the control for all robots in $\mathcal{Q}_1$ is augmented by a ``one-sided" cyclic controller (that is flow-invariant with respect to the polygon $\mathcal{M}^{(1)}_{|\mathcal{V}_1|}$) that aligns $\mathcal{Q}_1$ with its neighboring face $\mathcal{Q}_2$. Note that only one of $\nu_f$ or $\nu_b$ will be equal to 1 if $\mathcal{Q}_1$ does not possess the correct in-plane rotation. 

The proof of convergence for $\bm{x}^{(1)}$ to the subspace $\mathcal{M}^{(1)}_{|\mathcal{V}_1|}$ is a trivial extension of Lemma \ref{Lemma:Mn_to_Mnrho} by noting that the net control for $\mathcal{Q}_1$ may be written as:
\[
	\begin{split}
	\dot{\bm{x}}^{(1)} = \bm{u}^{(1)} = &-\mathcal{L}_{\eta}^{(1)}(\bm{x}^{(1)},\tau)\bm{x}^{(1)} \\&- |k_r(\bm{x}^{(1)})|\left(\nu_f L_1\otimes R_{r_s} + \nu_b L_1^T \otimes R_{r_s}^T \right)\bm{x}^{(1)}. \end{split}
\]
A suitable auxiliary system for the above dynamics is given by:
\[
	\begin{split}
	\dot{\bm{y}}^{(1)} = -\bar{V}^{(1)}\bigg[&\mathcal{L}_{\eta}^{(1)}(\bm{x}^{(1)}_{p},\tau) \\+&|k_r(\bm{x}^{(1)}_p)|\left( \nu_f L_1\otimes R_{r_s} + \nu_b L_1^T \otimes R_{r_s}^T \right)\bigg]\\ \times &\left(\bar{V}^{(1)^T}\bm{y}^{(1)} + \bar{U}^{(1)^T}\bar{U}^{(1)}\bm{x}_{p}^{(1)}\right),
	\end{split}
\]
where $\bar{V}^{(1)}$ is the orthonormal counterpart of $\widetilde{V}^{(1)}$ without the matrix $E^{(1)}$. Application of Corollary \ref{co:applied_contraction} then requires:
\[ \begin{split}
	\bar{V}^{(1)}\big[\mathcal{L}_{\eta}^{(1)}(\bm{x}^{(1)}_{p},\tau) + \dfrac{1}{2}|k_r(\bm{x}^{(1)}_p)|\mathcal{L}_{1\eta}^{(1)} \big] \bar{V}^{(1)^T} \succ 0,\ &\text{uniformly}, \\  &\forall \bm{x}_p, \tau. \end{split}
\]
Note that $\bar{V}^{(1)}|k_r(\bm{x}^{(1)}_p)|\mathcal{L}_{1\eta}^{(1)} \bar{V}^{(1)^T} \succ 0,\ \forall \bm{x}_p$. Thus, the above condition reduces to \eqref{convergence_cond_size} in Lemma \ref{contraction_size}. The proof of convergence for polygon size is unchanged since the rotational controller introduced in \eqref{cyclic_rotational} does not affect the expression for $\dot{p}_i$ given in \eqref{size_cont}.

The control law for the remaining robots in the polyhedron is given as before by \eqref{poly_cyclic_ind} and \eqref{poly_control_robot}. Conditions for convergence for the reduced PPS (defined by the set of polygons $\mathcal{Q}\setminus \mathcal{Q}_1$) can be derived by eliminating the first block in the net constraint matrix given in \eqref{V_poly}. The robots in $\mathcal{V}_1$ converge to the polygon $\mathcal{Q}_1$ with the desired size and an orientation consistent with the reduced PPS (as dictated by $\mathcal{Q}_2$). In turn, all polygons in the reduced PPS will then possess the desired size $\rho$ as dictated by the polygon $\mathcal{Q}_1$.

\section{Quadcopter Velocity Tracking Control}
In this section we present the quadcopter velocity tracking controller used for the numerical simulations presented in the main body of the paper. For notational convenience we drop the $i$ subscript notation referring to the $i^{th}$ quadcopter. In this section we assume that all quantities refer to a single quadcopter. 

Let $\bm{v}_d$ denote the desired velocity for a quadcopter, as dictated by the cyclic control laws, and let $\bm{v}$ be the current velocity. A simple proportional-integral-derivative controller converts the velocity error $\bm{e}_v$ given by $\bm{v}_d - \bm{v}$ into a desired acceleration $\bm{a}_d$. To achieve this desired acceleration, a reference attitude quaternion $\bm{q}_d$ must also be generated since the quadcopter only posseses thrust capability in one direction (opposite its body $z$ axis). To do this, we equate the desired acceleration to the actual acceleration dynamics given in \eqref{quad_trans} and linearize the resulting equation assuming small pitch and roll angles. This yields the following system of equations:

\begin{subequations}
\begin{align}
\begin{bmatrix} -T\sin(\psi) & -T\cos(\psi) \\ T\cos(\psi) & -T\sin(\psi) \end{bmatrix} \begin{bmatrix} \phi \\ \theta \end{bmatrix} &= m_q \begin{bmatrix} a_{d_x} \\ a_{d_y} \end{bmatrix} \\
		-T + m_q g &= m a_{d_z},
\end{align}
\end{subequations}
where $T$ is the desired net thrust, and $[a_{d_x},\ a_{d_y},\ a_{d_z}]$ are the components of $\bm{a}_d$. Given a desired yaw angle $\psi$, we can solve the above system of equations for the required thrust ($T$), roll ($\phi$) and pitch ($\theta$). The roll, pitch and yaw angles are then converted into a desired attitude quaternion $\bm{q}_d$ which is fed into the attitude control loop. 

The attitude control loop uses the desired attitude $\bm{q}_d$ and the current measured attitude $\bm{q}$ to compute a desired correction quaternion $\Delta \bm{q}$ given by:
\[
	\begin{split}
	\Delta \bm{q} &= \bm{q}_d \odot \bm{q}^{-1}\\ &= \begin{bmatrix} \cos(\beta/2) & \sin(\beta/2) \hat{\bm{n}}_x & \sin(\beta/2) \hat{\bm{n}}_y & \sin(\beta/2) \hat{\bm{n}}_z \end{bmatrix}^T   .\end{split}
\]
The desired moment is then given by the following proportional-derivative control law:
\[
	\bm{M} = K_p \beta \hat{\bm{n}} - K_d \bm{\omega},
\]
where $K_p$ and $K_d$ are positive semi-definite gain matrices. Finally, the desired thrust and moment are fed into the quadcopter motor controller to compute the desired motor speeds. A block diagram for this controller is shown below in Figure \ref{fig:quad_vel_control}.
\begin{figure}[htbp]
\begin{center}
	\includegraphics[scale=0.33]{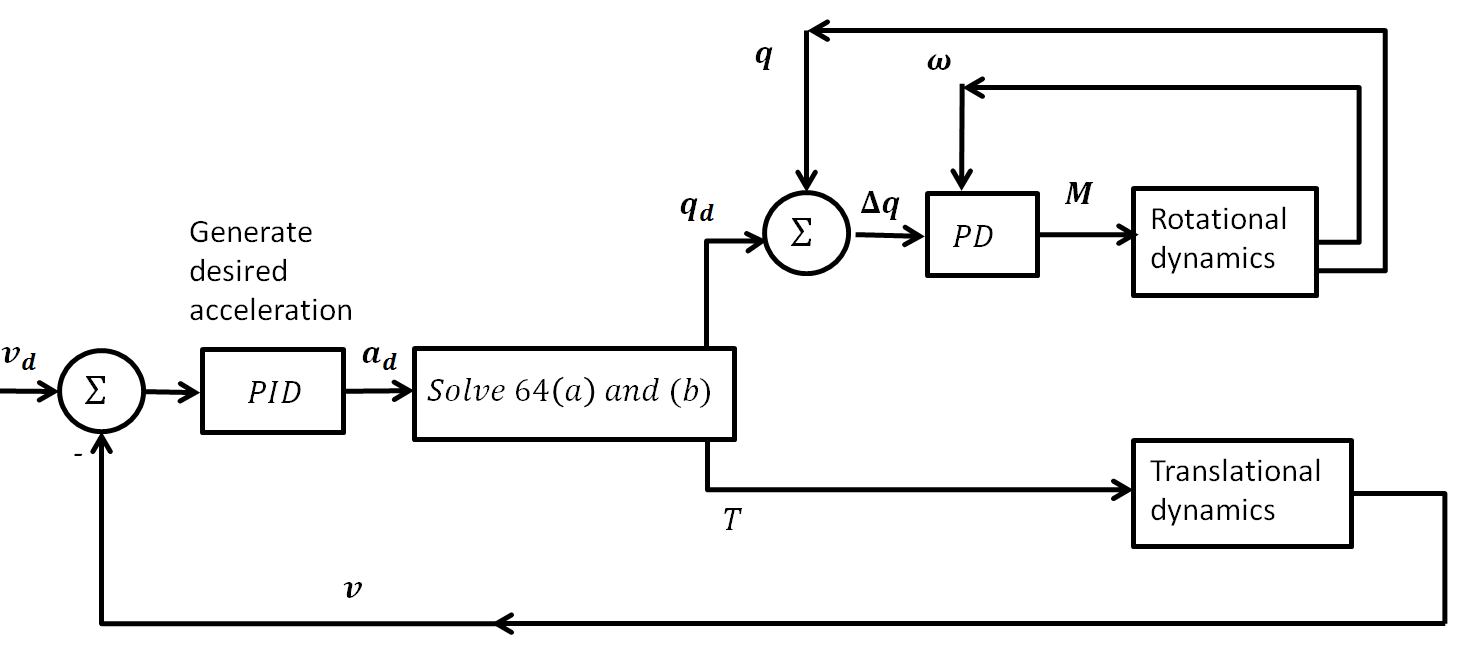}
\end{center}	
	\caption{Quadcopter Velocity Tracking Controller.}
	\label{fig:quad_vel_control}
\end{figure}

}{}
\end{document}